\documentclass[a4paper]{article}
\usepackage{graphicx}
\usepackage{amsmath}
\usepackage{amsthm}

\usepackage{geometry}
\usepackage{amsfonts}
\usepackage{array}
\usepackage{amssymb}
\usepackage{upgreek}
\usepackage{mathrsfs}

\usepackage{enumitem}
\usepackage{cancel}

\usepackage{pgf}
\usepackage{tikz}

\usepackage{pgfplots}
\pgfplotsset{compat=newest}
\usetikzlibrary{arrows}
\usetikzlibrary{matrix}
\usepackage{color}
\usepackage{subcaption}

\definecolor{beaublue}{rgb}{0.74, 0.83, 0.9}
\definecolor{emerald}{rgb}{0.31, 0.78, 0.47}

\usepackage{hyperref}

\usepackage{authblk}

\newtheorem{corollary}{Corollary}
\newtheorem{proposition}{Proposition}
\newtheorem{definition}{Definition}
\newtheorem{lemma}{Lemma}
\newtheorem{example}{Example}
\newtheorem{remark}{Remark}
\newtheorem{theorem}{Theorem}[section]
\numberwithin{equation}{section}

\makeatletter
\def\@maketitle{%
  \newpage
  \null
  \vskip 2em%
  \begin{center}%
  \let \footnote \thanks
    {\Large\bfseries \@title \par}%
    \vskip 1.5em%
    {\normalsize
      \lineskip .5em%
      \begin{tabular}[t]{c}%
        \@author
      \end{tabular}\par}%
    \vskip 1em%
    {\normalsize \@date}%
  \end{center}%
  \par
  \vskip 1.5em}
\makeatother

\title{\textbf{The space of light rays: Causality and $L$--boundary}}

\author{A. BAUTISTA \thanks{E--mail: \texttt{alfredo.bautista@uam.es}}}
\affil{Depto. de An\'alisis Econ\'omico: Econom\'{\i}a Cuantitativa, Univ. Aut\'onoma de Madrid \protect\\ C/ Francisco Tom\'as y Valiente 5, 28049 Madrid, Spain.} 

\author{A. IBORT \thanks{E--mail: \texttt{albertoi@math.uc3m.es}}}
\affil{Depto. de Matem\'aticas, Univ. Carlos III de Madrid \protect\\ Avda. de la Universidad 30, 28911 Legan\'es, Madrid, Spain, and \protect\\ ICMAT, Instituto de Ciencias Matem\'{a}ticas (CSIC-UAM-UC3M-UCM) \protect\\ C/ Nicol\'as Cabrera, 13-15, 28049, Madrid, Spain. }

\author{J. LAFUENTE \thanks{E--mail: \texttt{jlafuente@mat.ucm.es}}}
\affil{Depto. de Geometr\'{\i}a y Topolog\'{\i}a, Univ. Complutense de Madrid \protect\\ Avda. Complutense s/n, 28040 Madrid, Spain.} 

\date{May, 2022}

\begin{document}

\maketitle

\begin{abstract}
The space of light rays $\mathcal{N}$ of a conformal Lorentz manifold $(M,\mathcal{C})$ is, under some topological conditions, a manifold whose basic elements are unparametrized null geodesics. This manifold $\mathcal{N}$, strongly inspired on R. Penrose's twistor theory, keeps all information of $M$ and it could be used as a space complementing the spacetime model. In the present review, the geometry and related structures of $\mathcal{N}$, such as the space of skies $\Sigma$ and the contact structure $\mathcal{H}$, are introduced. The causal structure of $M$ is characterized as part of the geometry of $\mathcal{N}$. A new causal boundary for spacetimes $M$ prompted by R. Low, the $L$-boundary, is constructed in the case of $3$--dimensional manifolds $M$ and proposed as a model of its construction for general dimension. Its definition only depends on the geometry of $\mathcal{N}$ and not on the geometry of the spacetime $M$. The properties satisfied by the $L$--boundary $\partial M$ permit to characterize the obtained extension $\overline{M}=M\cup \partial M$ and this characterization is also proposed for general dimension. 

\end{abstract}

\hrulefill

\tableofcontents

\hrulefill

\section{Introduction}\label{sec:Intro}

The majority of the scientific community dedicated to mathematical physics is aware of the importance of the magnificent Roger Penrose's scientific contribution in the various fields in which he has been working. 
One of these fields is the study of the formation of black holes within the framework of the General Theory of Relativity, for which he was awarded the Nobel Prize in Physics in 2020. 
As part of that goal, Penrose has pioneered various theories that have been remarkably successful for their achievements and originality, for instance, the study of causal relations \cite{GKP68}, \cite{Pe72}, in the sense of describing the global properties of spacetime depending on what events influence to (or are influenced by) others. 
Most of the tools used in this field are quite simple, but the idea is powerful enough so that it has been very fruitful and still remains active today \cite{Ak21}, \cite{He21}, \cite{Ci22}. 
Non--spacelike curves are one of these tools, so if we would want to characterize the causality in any geometrical new model of the spacetime, then it will be necessary to describe such curves. Once the causal relations are determined, then it is possible to look for the $c$--boundary, which is a conformal boundary of the spacetime defined by adding the chronological past and future of inextensible causal curves as ideal points \cite{GKP68}, \cite{Pe72}.  
The construction of boundaries of the spacetimes is motivated, among other reasons, to extend the spacetime in order to study the singularities which could appear ``outside" the model, or to compactify ``infinite" models to study properties  at a finite range \cite{Pe64}.

Another theory developed by Penrose is \emph{twistor theory} \cite{Pe77}, \cite{Pe84}, \cite{Pe88}. 
This theory takes advantage of the geometry of the $4$--dimensional Minkowski spacetime $\mathbb{M}^4$ to establish a complementary framework. In  this new geometry for the spacetime, the basic elements are the paths of massless particles instead of the events of the spacetime. Calling $\mathbb{M}_{\mathbb{C}}\simeq \mathbb{C}^4$ the complexification of $\mathbb{M}^4$, then there exists a \emph{double fibration} 
\begin{equation}\label{double-fibration-1}
\begin{tikzpicture}[every node/.style={midway}]
\matrix[column sep={6em,between origins},
        row sep={2em}] at (0,0)
{ ; &  \node(F)   { $\mathbb{F}_{12}(\mathbb{C}^4)$}  ; &  ; \\
 \node(PT)   { $\mathrm{Gr}_{1}\left(\mathbb{C}^{4}\right)$}; &    ; & \node(M)   { $\mathbb{M}_{\mathbb{C}}$} ;  \\} ; 
\draw[->] (F) -- (PT) node[anchor=south east]  {$\pi_1$};
\draw[->] (F) -- (M) node[anchor=south west]  {$\pi_2$};
\end{tikzpicture}
\end{equation}
where $\mathbb{F}_{12}(\mathbb{C}^4)$ is the \emph{flag manifold} of (complex) $1$-- and $2$--dimensional vector subspaces of $\mathbb{C}^{4}$ and $\mathrm{Gr}_{1}\left(\mathbb{C}^{4}\right)\simeq \mathbb{P}(\mathbb{C}^4)$ is the projective space of $\mathbb{C}^4$.
The spaces $\mathrm{Gr}_{1}\left(\mathbb{C}^{4}\right)$ and $\mathbb{M}_{\mathbb{C}}$ contain complementary information (all contained in the flag manifold) that can be transferred from one to the other by the double fibration. 

The beauty of twistor theory rests on the specific geometry of the complex projective spaces involved in the double fibration (\ref{double-fibration-1}), so this geometry is not applicable to more general spacetimes. 

A new attempt to apply the ideas of twistor theory for non--flat spacetimes arose when R. Low, a PhD student of Penrose at the time, proposed to use real geometry instead of complex one (see, for instance, \cite{Lo88}--\cite{Lo06}). 
Low laid the foundations for building a new framework allowing the study of the conformal properties of general spacetimes. 
This new geometry is the \emph{space of light rays} $\mathcal{N}$. Its elements are the images of all null geodesics in the spacetime $M$ (called light rays). 
This space pretends to be a complement (or even a substitute) of the conformal class of the given spacetime in the study of, for instance, causality or other conformal aspects on physical models (see also \cite{Ch08}, \cite{Ch10}, \cite{Ch18}, \cite{Ba14}, \cite{Ba15}).
It is proven in \cite{Ba14} that the conformal manifold $M$ is encoded in $\mathcal{N}$ as a family $\Sigma$ of submanifolds in $\mathcal{N}$, called the space of skies, such that the conformal manifold $M$ can be univocally recovered knowing how the submanifolds of $\Sigma$ lie in $\mathcal{N}$. These submanifolds are called skies and consist of all light rays passing through a given point in $M$.

In \cite{Lo06}, Low introduces the idea of how to construct a possible new boundary for the spacetime $M$.  This can be done from $\mathcal{N}$ and $\Sigma$ by adding the congruence of light rays ``arriving at" or ``going out from" the ``same" endpoint of null geodesics. 
Some preliminary results have been obtained in \cite{Ba17} by Low et al. showing that the proposed boundary can be consistent and coinciding, in some simple cases, with the part of the $c$--boundary which is accessible by null geodesics. In \cite{Ba18}, this new boundary, called the \emph{light boundary} or the \emph{$L$--boundary}, is studied deeper and characterized whenever it exists. At this point, there are still technical problems to solve, such as, for example, its differentiability.

The present article is motivated by a talk given at the meeting \emph{Singularity theorems, causality, and all that. A tribute to Roger Penrose} on 16th of June of 2021 to pay public tribute to R. Penrose for his influence in Mathematics and Physics of the late mid-twentieth and early twenty-first centuries. All results in this review are strongly inspired by his work on causality, conformal boundaries and twistor theory. Roger Penrose has paved the way and their original ideas will be an inspiration for new generations of researchers for a long time.

This review pretends to introduce the geometry of the space of light rays $\mathcal{N}$ of a conformal manifold $M$ to a wider audience summarizing results which are disperse in the literature. The topics emphasized here are the geometric and topological structures of $\mathcal{N}$, the causal relations of the spacetime $M$ studied from the structures of $\mathcal{N}$ and the construction of the $L$--boundary.  
References are given along this review to balance brevity and a detailed exposition of the subject.

The outline of this article is as follows. After a slight brush-stroke on causality, section \ref{sec:lightrays} is devoted to introduce the space of light rays $\mathcal{N}$ of a conformal manifold $M$ as well as its structures. The hypotheses that $M$ must verify for $\mathcal{N}$ to have  nice properties are justified. 
A characterization of tangent vectors of $\mathcal{N}$ in terms of elements of $M$ is given in section \ref{sec:lightrays-tangent}. They can be seen as a class of Jacobi fields along null geodesics describing light rays. Another canonical structure in $\mathcal{N}$, the contact structure $\mathcal{H}\subset T\mathcal{N}$, is described briefly in section \ref{sec:lightrays-contact}.

Section \ref{sec:lightrays-skies} is addressed to describe the \emph{space of skies} $\Sigma$ which will be diffeomorphic to the spacetime $M$ but its points consist of submanifolds in $\mathcal{N}$. The elements of $\Sigma$ are called \emph{skies} and they consist of the congruence of light rays passing through a given event of $M$. 
There is a topology and a differentiable structure in $\Sigma$ inherited from $\mathcal{N}$, which makes $\Sigma$ diffeomorphic to the conformal manifold $M$. 
Moreover, the way in which the skies of $\Sigma$ are embedded in $\mathcal{N}$ allows to obtain all the information of the conformal structure of $M$ as the Reconstruction theorem states, it means, under suitable hypotheses, the conformal manifold $M$ can be recovered from the geometry of $\mathcal{N}$. 

In section \ref{sec:Causality}, we characterize causal curves in $M$ as a class of Legendrian isotopies of skies in $\mathcal{N}$. This permits to recover all information about causality in $M$ from the ``movement" of skies across the space of light rays $\mathcal{N}$. In the section \ref{sec:Twisted}, we show the existence of curves in $\mathcal{N}$ tangent to skies everywhere in such a way they define null curves in $M$, not geodesic at any point, called \emph{twisted null curves}, connecting chronological related points.  Finally, the property of sky--linking corresponding to pairs of skies of causally related points in $M$ is briefly described in section \ref{sec:Linking}.

The construction of the $L$--boundary for $\dim M = 3$ is done in section \ref{sec:LBoundary}. The original idea, motivated by twistor geometry and due to R. Low is described. Moreover, the hypotheses under which this construction can be done are analysed in sections \ref{sec:LBoundary-idea} and \ref{sec:LBoundary-3D}. 
Then, in section \ref{sec:LBoundary-Ntilde}, an open submanifold $\widetilde{\mathcal{N}}$ of the bundle $\mathbb{P}\left(\mathcal{H}\right)$ of $1$--dimensional vector subspaces of the contact structure $\mathcal{H}$ which are tangent to skies is described. 
Three different distributions are defined in $\mathbb{P}\left(\mathcal{H}\right)$ in section \ref{sec:LBoundary-distrib}, one in $\widetilde{\mathcal{N}}$ and two more in the boundary $\partial\widetilde{\mathcal{N}}$. 
Sections \ref{sec:Lboundary-Smoothness} and \ref{sec:Lboundary-canonical} study the conditions under which the union of these distributions forms a unique and smooth distribution $\overline{\mathcal{D}^{\sim}}$ such that the quotient of the closure $\overline{\widetilde{\mathcal{N}}}$ over $\overline{\mathcal{D}^{\sim}}$ is, whenever regularity holds, a Hausdorff manifold with boundary $\overline{M}$ such that $M\subset \overline{M}$ and $\partial M=\overline{M}-M$.  

Finally, in section \ref{sec:Lextensions}, a characterization of the $L$--boundary constructed in the previous section is done in a dimensional independent way. Some examples illustrating the properties satisfied by the $L$--extensions are given.

\section{The space of light rays}\label{sec:lightrays}

We will start with a $m$--dimensional \emph{Lorentz manifold} $\left(M,\mathbf{g}\right)$ with $m\geq 3$, that is a $m$--dimensional Hausdorff smooth manifold equipped with a non--degenerate metric $\mathbf{g}$ with signature $\left(-+\cdots +\right)$ such that $\left(M,\mathbf{g}\right)$ is time--oriented, that means that there exists a global vector field $T\in \mathfrak{X}\left(M\right)$ such that $\mathbf{g}\left(T,T\right)<0$ at any point of $M$, where $\mathfrak{X}\left(M\right)$ denotes the set of all smooth vector fields in $M$. For brevity, sometimes we will call \emph{spacetime} to $\left(M,\mathbf{g}\right)$.

For a given $\left(M,\mathbf{g}\right)$, we can define its \emph{conformal (Lorentz) structure}
\[
\mathcal{C}_{\mathbf{g}}=\left\{ \overline{\mathbf{g}}=e^{f}\mathbf{g}: f\in \mathfrak{F}\left(M\right) \right\}
\]
where $\mathfrak{F}\left(M\right)$ denotes the ring of smooth function in $M$. 
By definition, the conformal structure $\mathcal{C}_{\mathbf{g}}$ is the set of all Lorentz metrics proportional to the given $\mathbf{g}$ with positive function of proportionality.
So, the pair $\left(M,\mathcal{C}_{\mathbf{g}}\right)$ is called a \emph{conformal (Lorentz) manifold} and it can be denoted by $\left(M,\mathcal{C}\right)$, or even $M$, when the mention to the metric $\mathbf{g}$ is not necessary. 

We will use the standard notation for the tangent bundle of some smooth manifold $N$ and its tangent spaces at $p\in N$, that is $TN$ and $T_p N$ respectively. 

Next, we will construct the space of light rays and all its additional structures for a given $m$--dimensional conformal Lorentz manifold $\left(M,\mathcal{C}\right)$ with $m\geq 3$.
The case $m=2$ can also be constructed, but some of its associated spaces have discrete or trivial structures. 
Anyway, we can use this case for illustrative purpose.

\subsection{Construction of the space of light rays}\label{sec:lightrays-causal}

Let us consider a conformal structure $\mathcal{C}$ in $M$. Fixed a metric $\mathbf{g}\in\mathcal{C}$, it is possible to classify any tangent vector $v\in TM$ depending on the sign of $\mathbf{g}\left(v,v\right)\in \mathbb{R}$.
So, a tangent vector $v\in T_p M$ is called 
\begin{itemize}
\item  \emph{timelike} $\Longleftrightarrow$ $\mathbf{g}\left(v,v\right)<0$
\item  \emph{null} $\Longleftrightarrow$ $\mathbf{g}\left(v,v\right)=0$
\item  \emph{spacelike} $\Longleftrightarrow$ $\mathbf{g}\left(v,v\right)>0$ .
\end{itemize}
We will also say that $v\in T_pM$ is \emph{lightlike} if it is null and $v\neq \mathbf{0}$, and we will say that $v\in T_pM$ is \emph{causal} if it is timelike or null.

Observe that for any other $\overline{\mathbf{g}}=e^f \mathbf{g}\in \mathcal{C}$ with $f\in\mathfrak{F}\left(M\right)$, trivially we have 
\[
\mathrm{sign}\left(\overline{\mathbf{g}}\left(v,v\right) \right)=\mathrm{sign}\left(\mathbf{g}\left(v,v\right) \right) \text{ for all }v\in TM
\]
then, this classification, named \emph{causal character} of tangent vectors is defined at $\left(M,\mathcal{C}\right)$ because it does not depend on the representative metric $\mathbf{g}\in\mathcal{C}$, so the \emph{causality} or \emph{causal character} is a conformal property.

Since $M$ is time--oriented, the global timelike vector field $T\in \mathfrak{X}(M)$ defines the time--orientation of vectors saying that $T$ is \emph{future--directed} and $-T$ \emph{past--directed}. All non--zero causal vectors $u$ can be classified into these two categories of future or past--directed vectors depending on the sign of $\mathbf{g}(u,T)$. 
In particular, if we denote by $\mathbb{N}_p\subset T_p M$ the set of all lightlike vectors at $p\in M$ then it splits in two connected components $\mathbb{N}_p = \mathbb{N}_p^+ \cup \mathbb{N}_p^-$ where
\[
\mathbb{N}^{+}_{p}=\{u\in \mathbb{N}_{p}: \mathbf{g}(u,T)<0\} \quad \text{ and } \quad \mathbb{N}^{-}_{p}=\{u\in \mathbb{N}_{p}: \mathbf{g}(u,T)>0\} .
\]
We will call $\mathbb{N}_p^+$ the set of \emph{future--directed} lightlike vectors at $p$ and $\mathbb{N}_p^-$ the set of the \emph{past--directed} ones.
Hence the disjoint union
\[
\mathbb{N}^{+}=\bigcup_{p\in M}\mathbb{N}_p^{+}\subset TM
\]
denotes the sub--bundle of $TM$ given by the future--directed lightlike vectors whose fibres are $\mathbb{N}_p^{+}$.

The causal character is extended automatically to smooth curves $\alpha:I\rightarrow M$ by the classification of its tangent vector $\alpha'\left(t\right)\in T_{\alpha\left(t\right)}M$ for $t\in I$. 
Then we can talk about timelike, null, spacelike, lightlike or causal curves whenever the corresponding tangent vectors lie in such categories.

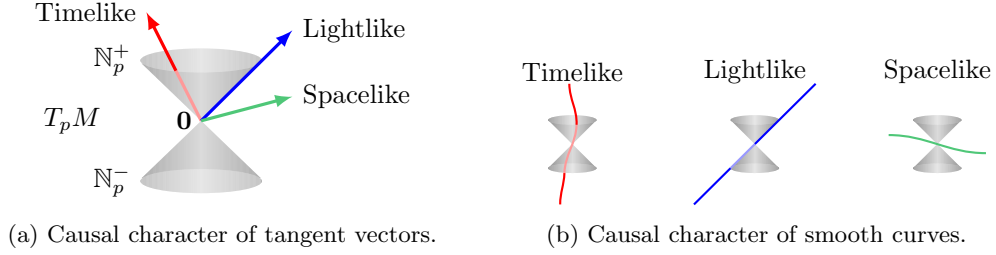
\begin{figure}[h]
\centering
\begin{subfigure}[b]{0.45\textwidth}
\centering
\begin{tikzpicture}[scale=0.8]
\fill[left color=gray!15!black, right color=gray!15!black,middle color=gray!20,shading=axis,opacity=0.15] (-0.98,0.95) -- (0,0) -- (0.98,0.95);
\fill[color=white] (1,1) arc (0:360:1cm and 0.2cm);
\fill[left color=gray!05!black, right color=gray!05!black,middle color=gray!10,shading=axis,opacity=0.1] (1,1) arc (0:360:1cm and 0.2cm);
\fill[left color=gray!15!black, right color=gray!15!black,middle color=gray!20,shading=axis,opacity=0.15] (-0.98,-0.95) -- (0,0) -- (0.98,-0.95);
\fill[color=white] (1,-1) arc (0:360:1cm and 0.2cm);
\fill[left color=gray!05!black, right color=gray!05!black,middle color=gray!10,shading=axis,opacity=0.1] (1,-1)  arc (0:360:1cm and 0.2cm);
\draw[very thick,color=red!40]  (0,0) -- (-0.415,0.83); 
\draw[-latex, very thick, color=red]  (-0.415,0.83) -- (-0.9,1.8) node[anchor=east,color=black] {Timelike};
\draw[-latex, very thick, color=blue]  (0,0) -- (1.5,1.5) node[anchor=west,color=black] {Lightlike};
\draw[-latex, very thick, color=emerald]  (0,0) -- (1.5,0.4) node[anchor=west,color=black] {Spacelike};    
\draw (-1,1) node[anchor=east] {$\mathbb{N}^{+}_p$};
\draw (-1,-1) node[anchor=east] {$\mathbb{N}^{-}_p$};
\draw (0,0) node[anchor=east] {$\mathbf{0}$};
\draw (-1.5,0) node[anchor=east] {$T_p M$};
\end{tikzpicture}
\caption{Causal character of tangent vectors.}
  \label{diapositivas1}
\end{subfigure}
\hspace{5pt}
\begin{subfigure}[b]{0.45\textwidth}
\centering
\begin{tikzpicture}[scale=0.8]

\fill[left color=gray!15!black, right color=gray!15!black,middle color=gray!20,shading=axis,opacity=0.15] (-3.4,0.4) -- (-3,0) -- (-2.6,0.4);
\fill[color=white] (-2.6,0.4) arc (0:360:0.4cm and 0.08cm);  
\fill[left color=gray!05!black, right color=gray!05!black,middle color=gray!10,shading=axis,opacity=0.1] (-2.6,0.4) arc (0:360:0.4cm and 0.08cm);
\fill[left color=gray!15!black, right color=gray!15!black,middle color=gray!20,shading=axis,opacity=0.15] (-3.4,-0.4) -- (-3,0) -- (-2.6,-0.4);
\fill[color=white] (-2.6,-0.4) arc (0:360:0.4cm and 0.08cm);
\fill[left color=gray!05!black, right color=gray!05!black,middle color=gray!10,shading=axis,opacity=0.1] (-2.6,-0.4) arc (0:360:0.4cm and 0.08cm);
\draw[thick,color=red!40] (-3.13,-0.48) to[out=88, in=-92] (-2.93,0.32);  
\draw[thick,color=red] (-2.93,0.32) to[out=88, in=-89] (-3.05,1);  
\draw[thick,color=red] (-3.2,-1) to[out=90, in=-92] (-3.13,-0.48);  

\fill[left color=gray!15!black, right color=gray!15!black,middle color=gray!20,shading=axis,opacity=0.15] (-0.4,0.4) -- (0,0) -- (0.4,0.4);
\fill[color=white] (0.4,0.4) arc (0:360:0.4cm and 0.08cm);
\fill[left color=gray!05!black, right color=gray!05!black,middle color=gray!10,shading=axis,opacity=0.1] (0.4,0.4) arc (0:360:0.4cm and 0.08cm);
\fill[left color=gray!15!black, right color=gray!15!black,middle color=gray!20,shading=axis,opacity=0.15] (-0.4,-0.4) -- (0,0) -- (0.4,-0.4);
\fill[color=white] (0.4,-0.4) arc (0:360:0.4cm and 0.08cm);
\fill[left color=gray!05!black, right color=gray!05!black,middle color=gray!10,shading=axis,opacity=0.1] (0.4,-0.4) arc (0:360:0.4cm and 0.08cm);
\draw[thick,color=blue] (0,0) -- (1,1);  
\draw[thick,color=blue!40] (-0.4,-0.4) -- (0,0);  
\draw[thick,color=blue] (-1,-1) -- (-0.4,-0.4);  
\fill[left color=gray!15!black, right color=gray!15!black,middle color=gray!20,shading=axis,opacity=0.15] (2.6,0.4) -- (3,0) -- (3.4,0.4);
\fill[color=white] (3.4,0.4) arc (0:360:0.4cm and 0.08cm);
\fill[left color=gray!05!black, right color=gray!05!black,middle color=gray!10,shading=axis,opacity=0.1] (3.4,0.4) arc (0:360:0.4cm and 0.08cm);
\fill[left color=gray!15!black, right color=gray!15!black,middle color=gray!20,shading=axis,opacity=0.15] (2.6,-0.4) -- (3,0) -- (3.4,-0.4);
\fill[color=white] (3.4,-0.4) arc (0:360:0.4cm and 0.08cm);
\fill[left color=gray!05!black, right color=gray!05!black,middle color=gray!10,shading=axis,opacity=0.1] (3.4,-0.4) arc (0:360:0.4cm and 0.08cm);
\draw[thick,color=emerald] (2.2,0.15) to[out=0, in=180] (3.8,-0.15);  
\draw (-3,1.2) node {Timelike};
\draw (0,1.2) node {Lightlike};
\draw (3,1.2) node {Spacelike};
\end{tikzpicture}
\caption{Causal character of smooth curves.}
  \label{diapositivas2}
\end{subfigure}
  \caption{Causal character.}
  \label{diapositivas1y2}
\end{figure}

It is known by \cite[Lem. 2.1]{Ku88} that any null geodesic related to the metric $\mathbf{g}\in \mathcal{C}$ is also a pregeodesic for any metric $\overline{\mathbf{g}}\in \mathcal{C}$. 
This implies that if $\gamma:I\rightarrow M$ is a null geodesic for the metric $\mathbf{g}\in \mathcal{C}$, then there exists a reparametrization $ \overline{\gamma}:I'\rightarrow M$ of $\gamma$ such that $\overline{\gamma}$ is a null geodesic for the metric $\overline{\mathbf{g}}\in \mathcal{C}$.
Since the images of $\gamma$ and $\overline{\gamma}$ coincide, then we can define a \emph{light ray} in $\left(M,\mathcal{C}\right)$ as the image of a maximal null geodesic related to some (and therefore, to any) metric $\mathbf{g}\in \mathcal{C}$. 
Hence, the \emph{set of light rays} of $M$ is defined as
\[
\mathcal{N}=\{\gamma\left(I\right)\subset M: \quad \gamma:I\rightarrow M \text{ is a maximal null geodesic } \} 
\] 
which is a conformal definition.

We will abuse of the notation and denote a light ray defined by a null geodesic $\gamma:I\rightarrow M$ related to some metric $\mathbf{g}\in \mathcal{C}$ by the same lower greek letter as an element $\gamma\in \mathcal{N}$ or as subset $\gamma\subset M$, so we can interpret a light ray as an unparametrized null geodesic.

\subsection{Causality conditions}\label{sec:causal-conditions}

The causal structure provides topological properties to the conformal manifold called \emph{causality conditions}. 
In fact, these conditions constitute a hierarchy such that the stronger conditions also verify the weaker ones, that is why this hierarchy is also called \emph{the causal ladder}.
For a complete description, see \cite{Mi08} and also \cite{Pe72}, \cite{HE}, \cite{On83}, \cite{BE96}. 
We will only mention some steps of this \emph{causal ladder}.

We say that $\left(M,\mathcal{C}\right)$ satisfies the \emph{chronological condition} if and only if there are no closed timelike curves in $M$. 
This condition permits the existence of closed causal curves (not strictly timelike, but still causal). 
To avoid this situation, we can require the next step of the causal ladder: the \emph{causal condition} which implies the non--existence of closed causal curves. 
But still, causal curves can be almost closed (see \cite[Fig. 38]{HE}). 
When there is not such kind of curves, then the conformal manifold is said to be \emph{strongly causal} (or to verify the \emph{strong causality condition}).
A physically relevant spacetime is supposed to be strongly causal because there is not physical experience of the existence of almost closed causal curves.
There are several steps before reaching the top of the hierarchy, that is the \emph{global hyperbolicity condition}. 
This is the strongest condition and it consists in the existence of a differentiable spacelike hypersurface $C\subset M$, called \emph{Cauchy surface}, such that each inextensible causal curve $\lambda$ intersects $C$ at exactly one point \cite{BS1}.
The existence of such Cauchy surface allows to fix in it initial data of Cauchy problems of differential equations in order to determine their solutions in the entire spacetime.
A globally hyperbolic conformal manifold is less general than a strongly causal one.

The study of the causality conditions is done by checking what different properties are satisfied by the sets that can be influenced by (or influence to) others by signals travelling with a speed not exceeding the speed of light.

\begin{definition}\label{def-causal-sets}
Let $A$ be a subset of a Lorentzian manifold $M$.
\begin{enumerate}
\item The \emph{chronological future of $A$} is the set $I^+\left(A\right)$ of all points in $M$ that can be connected from $A$ by a future--directed timelike curve.
The \emph{chronological past of $A$} is defined analogously and it will be denoted by $I^-\left(A\right)$.
\item The \emph{causal future of $A$}, denoted by $J^{+}(A)$,  is the union of $A$ and the set of all points in $M$ that can be connected from $A$ by a future--directed causal curve. Analogously we have $J^{-}\left(A\right)$, the \emph{causal past of $A$}.
\end{enumerate}
\end{definition}

In a equivalent way, the following notation is used in the literature.
\begin{equation}
\begin{tabular}{l}
$p\ll q  \Leftrightarrow q\in I^{+}(p)$ \\
$p<q \Leftrightarrow$ there exists a future--directed causal curve from $p$ to $q$ \\
$p\leq q \Leftrightarrow q\in J^{+}(p)$
\end{tabular}
\end{equation}

The following theorem is a basic result to study the causal structure of spacetimes and it can be found in \cite[Prop. 10.46]{On83}.

\begin{theorem}\label{1T5}
Let $M$ be a Lorentzian manifold. If $\alpha $ is a causal curve joining the points $p,q\in M$ but not a null pregeodesic, then in any neighbourhood of $\alpha $ there exists a timelike curve $\beta $ connecting the points $p$ and $q$.
\end{theorem}

As immediate consequences of definition \ref{def-causal-sets} and theorem \ref{1T5}, we get the following properties (see \cite[Cor. 14.1 \& Lem. 14.2]{On83} for proofs).

\begin{corollary}\label{1C3}
For $p\in M$ and $A\subset M$ we have 
\begin{enumerate}
\item $I^{+}\left( A\right)\subset J^{+}\left(A\right)$.
\item $I^{+}\left( A\right)=I^{+}\left(I^{+}\left( A\right)\right)$.
\item $J^{+}\left( A\right)=J^{+}\left(J^{+}\left( A\right)\right)$.
\item If $r\in J^{+}\left( q\right) $ and $q\in I^{+}\left( p\right) $, or also $r\in I^{+}\left( q\right) $ and $q\in J^{+}\left( p\right) $, then we have that $r\in I^{+}\left( p\right) $.
\item $p\in I^{+}\left( q\right)\Leftrightarrow q\in I^{-}\left( p\right)$ and moreover $p\in J^{+}\left( q\right)\Leftrightarrow q\in J^{-}\left( p\right)$.
\end{enumerate}
The statements 1--4 are also true when we consider the chronological and causal past $I^{-}$ , $J^{-}$.
\end{corollary}

The following proposition can be found in \cite[Lems. 14.3 \& 14.6]{On83}, and it shows that the causality and the topology of the spacetime are closely related.

\begin{proposition}
For any $A\subset M$ we have $I^{+}\left( A\right)$ is an open set in $M$. Moreover, $\mathrm{int}~J^{+}\left( A\right)=I^{+}\left( A\right)$ and $J^{+}\left( A\right)\subset \overline{I^{+}\left( A\right)}$.
\end{proposition}

There is a lot of literature on causality theory, but we will only present the few basic elements that we will need throughout this review.

For our purposes, we will assume that $\left(M,\mathcal{C}\right)$ is strongly causal. By \cite[Lem. 3.21 \& Rmk. 3.23 (2)]{Mi08}, we can state the following result.
\begin{proposition}\label{prop-Mingu-Sanch}
$\left(M,\mathcal{C}\right)$ is strongly causal if and only if for any $p\in M$ there exists a topological basis consisting of globally hyperbolic, causally convex and normal neighbourhoods.
\end{proposition}

We will keep in mind the meaning of properties of the neighbourhoods of proposition \ref{prop-Mingu-Sanch}.
\begin{itemize}
\item If $V\subset M$ is globally hyperbolic, then there is a Cauchy surface $C$ in $V$.
\item If $V$ is \emph{normal}, then for any $q\in V$ there exists a neighbourhood $W\subset T_q M$ of $0$ such that the exponential map $\mathrm{exp}_q:W\rightarrow V$ is a diffeomorphism. 
\item If $V$ is \emph{causally convex}, then the intersection of any causal curve $\lambda$ with $V$, if non--empty, has exactly one connected component.
\end{itemize}

\begin{remark}\label{rmk-basic-neighb}
We warn the readers that such kind of globally hyperbolic, normal and causally convex neighbourhoods $V\subset M$ will be used as a tool to achieve the local constructions along this paper. In addition, we will assume that $V$ is small enough to consider the Cauchy surface $C\subset V$ diffeomorphic to $\mathbb{R}^{m-1}$.
\end{remark}

\subsection{Topological and differentiable structure of \texorpdfstring{$\mathcal{N}$}{N}}\label{sec:lightrays-struct-N}

A detailed explanation of this section can be found in \cite[Sec. 2.2]{Ba15b}. 
Notice that we will denote by $\pi_B^A:A\rightarrow B$ the canonical projection of the bundle $A$ onto the base $B$.
Let us fix some auxiliary metric $\mathbf{g}\in \mathcal{C}$.
Observe that any future--directed lightlike vector $v\in \mathbb{N}^{+}$ defines the future--directed null geodesic $\gamma_v\subset M$ such that $\gamma_v\left(0\right)=\pi_{M}^{TM}\left(v\right)$ and $\gamma'_v\left(0\right)=v$, so $v\in \mathbb{N}^{+}$ defines the light ray $\gamma_{v}\in \mathcal{N}$ determined by the null geodesic $\gamma_v$.
Also, for any $a>0$, the vector $av\in \mathbb{N}^{+}$ defines another null geodesic but the same light ray, this means $\gamma_v =\gamma_{av}\in\mathcal{N}$. 
Then, if $\mathcal{D}_{\Delta}$ is the distribution in $\mathbb{N}^{+}$ such that its leaf passing by $v\in \mathbb{N}^{+}$ is $\left[v\right]=\mathrm{span}\{v\}\subset \mathbb{N}^{+}$, since $\mathcal{D}_{\Delta}$ is a regular distribution, then the \emph{bundle of null directions} $\mathbb{PN}=\mathbb{N}^{+} / \mathcal{D}_{\Delta}$ is a differentiable manifold such that the projection $\pi_{\mathbb{PN}}^{\mathbb{N}^{+}}:\mathbb{N}^{+}\rightarrow\mathbb{PN}$ given by $\pi_{\mathbb{PN}}^{\mathbb{N}^{+}}\left(v\right)=\left[v\right]$ is a submersion.
Now, recall that given two vectors $u,v\in T_p M$ such that $u=av$ with $a>0$, then the geodesics $\gamma_u$ and $\gamma_v$ verify $\gamma_u\left(s\right)=\gamma_{av}\left(s\right) =\gamma_v\left(as\right)$. 
Hence, the elevation of the null geodesics from $M$ to $\mathbb{PN}$ defines a distribution $\mathcal{D}_{G}$ in $\mathbb{PN}$ such that the leaves are, precisely, the unparametrized null geodesics, that is, the light rays. 
Since $M$ is assumed to be strongly causal, then $\mathcal{D}_{G}$ is regular and then $\mathcal{N}=\mathbb{PN}/ \mathcal{D}_{G}$ is a differentiable manifold such that the quotient map $\boldsymbol{\upgamma}:\mathbb{PN}\rightarrow \mathcal{N}$ given by $\boldsymbol{\upgamma}\left(\left[v\right]\right)=\gamma_{\left[v\right]}=\gamma_{v}$ is a submersion.

For any set $W\subset M$, we will use the following notation
\[
\mathbb{N}^{+}\left(W\right)=\{v\in \mathbb{N}^{+}:\pi_{M}^{TM}\left(v\right)\in W \} \quad \text{ and } \quad \mathbb{PN}\left(W\right)=\{\left[v\right]\in \mathbb{PN}:\pi_{M}^{\mathbb{PN}}\left(\left[v\right]\right)\in W \}   .
\]

In order to build a coordinate chart for $\mathcal{N}$, by proposition \ref{prop-Mingu-Sanch}, we will take a globally hyperbolic, normal, causally convex open neighbourhood $V\subset M$ at some $p\in M$. 
Given a timelike vector field $T\in\mathfrak{X}\left(M\right)$, we consider the restriction $\Omega\left(V\right)\subset\mathbb{N}^{+}\left(V\right)$ defined by
\begin{equation}\label{eq-Omega-V}
\Omega\left(V\right)=\{v\in \mathbb{N}^{+}\left(V\right):\mathbf{g}\left(v,T\right)=-1 \}
\end{equation}
which is diffeomorphic to $\mathbb{PN}\left(V\right)$, open set in $\mathbb{PN}$. 
If $C\subset V$ is a Cauchy surface in $V$, then any light ray $\gamma\in \mathcal{N}$ passing through $V$ is determined by the intersection point $c=\gamma\cap C$ and the corresponding null direction at $c\in C$. 
Then, with the abuse in the notation because $\Omega\left(V\right)\simeq\mathbb{PN}\left(V\right)$, the restricted maps $\left.\boldsymbol{\upgamma}\right|_{\mathbb{PN}\left(C\right)}$ and $\left.\boldsymbol{\upgamma}\right|_{\Omega\left( C\right)}$ are diffeomorphisms, hence we have the following diagram
\begin{equation}\label{diagram-charts}
\begin{tikzpicture} [every node/.style={midway}]
\matrix[column sep={8em,between origins},
        row sep={3em}] at (0,0)
{ \node(PN1)   {$\mathbb{PN}\left(V\right)\simeq\Omega\left( V\right)$}  ; & \node(N) {$\mathcal{N}_{V}\subset \mathcal{N}$}; \\
  \node(PN2) {$\mathbb{PN}\left(C\right)\simeq\Omega\left(C\right)$};    \\};
\draw[->] (PN1) -- (N) node[anchor=south]  {$\boldsymbol{\upgamma}$};
\draw[->] (PN2) -- (N) node[anchor=north]  {$\xi$};
\draw[<-right hook] (PN1)   -- (PN2) node[anchor=east] {$i$};
\end{tikzpicture}
\end{equation}
where $i:\Omega\left(C\right)\hookrightarrow\Omega\left(V\right)$ is the inclusion map and moreover $\xi=\left.\boldsymbol{\upgamma}\right|_{\Omega\left(C\right)}\circ i$ and $\xi=\left.\boldsymbol{\upgamma}\right|_{\mathbb{PN}\left(C\right)}\circ i$ are diffeomorphisms.
Therefore, a coordinate chart $\varphi$ for $\Omega\left(C\right)$ provides us a coordinate system $\left( \mathcal{N}_V, \psi\right)$ in $\mathcal{N}$ such that 
\begin{equation}\label{eq-chart-N}
\psi=\varphi\circ \xi^{-1}  .
\end{equation}
So, the topology of $\mathcal{N}$ is inherited from $\mathbb{PN}$ (locally represented by $\Omega\left(V\right)$) as a quotient space.
Observe that if, at some $p\in M$, we take some orthonormal basis $\{\mathbf{E}_i\}_{i=0,\ldots,m-1}$ related to some $\mathbf{g}\in\mathcal{C}$ such that $\mathbf{E}_0$ is future--directed and timelike (here we consider $T=\mathbf{E}_0$ at $p$) and $\mathbf{E}_i$ are spacelike for $i=1,\ldots,m-1$, then any null direction $[v]\in\mathbb{PN}_p$ can be defined by a null vector 
\[
v=\mathbf{E}_0 + v^1 \mathbf{E}_1 + \cdots + v^{m-1}\mathbf{E}_{m-1}\in \mathbb{N}^{+}_{p}
\]
verifying $\left(v^1\right)^2+\ldots +\left(v^{m-1}\right)^2=1$ since $\mathbf{g}(v,v)=0$. 
Therefore the fibres $\mathbb{PN}_p$ of $\mathbb{PN}$ are diffeomorphic to the sphere $\mathbb{S}^{m-2}$ and then the bundle of null directions $\mathbb{PN}(C)$ is diffeomorphic to the bundle $ST(C)$ of $(m-2)$--spheres on $C$. If $C$ is diffeomorphic to $\mathbb{R}^{m-1}$ then we can consider $\mathcal{N}\simeq C \times \mathbb{S}^{m-2}$.
\begin{remark}
If $M$ is globally hyperbolic, there exists a global smooth spacelike Cauchy surface $C\subset M$. So, the diffeomorphism $\xi$ of diagram (\ref{diagram-charts}) is global and therefore $\mathcal{N}\simeq \mathbb{PN}(C)$.  
\end{remark}

\begin{remark}\label{remark-v-en-Omega}
It is important to note that when working in coordinates $\psi:\mathcal{U}\subset \mathcal{N}\rightarrow \mathbb{R}^{2m-3}$ built by the diffeomorphism $\xi$ of diagram (\ref{diagram-charts}), we will set the parametrization of light rays as null geodesics such that their initial vectors at the local Cauchy surface $C\subset V$ given by $\gamma'(0)=v\in \Omega(C)$ verify $\mathbf{g}(v,T)=-1$ as in equation (\ref{eq-Omega-V}).   
\end{remark}

\begin{example}\label{ex-4-minkowski}
Consider the $4$--dimensional Minkowski spacetime $\mathbb{M}^4$, that is $\mathbb{R}^4$ equipped with the metric $\mathbf{g}=-dt\otimes dt + dx\otimes dx+dy\otimes dy+dz\otimes dz$ for standard coordinates $(t,x,y,z)$. The hypersurface $C\equiv\{t=0\}$ is a spacelike Cauchy surface of $M$ then any light ray intersects $C$ at exactly one point $q=(0,x,y,z)$. Using the spherical coordinates $\theta$, $\phi$, a null direction is generated by the vector
\[
v=\left(\frac{\partial}{\partial t}\right)_q + \cos \theta \sin \phi \left(\frac{\partial}{\partial x}\right)_q + \sin \theta \sin \phi  \left(\frac{\partial}{\partial y}\right)_q + \cos \phi \left(\frac{\partial}{\partial z}\right)_q   .
\]
If $\gamma\in \mathcal{N}$ is the light ray defined in $\mathbb{M}^4$ by the null geodesic $\gamma$ such that $\gamma(0)=q=(0,x,y,z)$ and $\gamma'(0)=v$, the map 
\[
\begin{tabular}{rrcl}
$\psi:$ & $\mathcal{N}$ & $\rightarrow$ & $\mathbb{R}^3 \times [0,2\pi) \times (0,\pi)$ \\
  & $\gamma$ & $\mapsto$ & $\left(x,y,z,\theta,\phi \right)$
\end{tabular}
\] 
is a coordinate system for $\mathcal{N}$ of $\mathbb{M}^4$.
\end{example}

\begin{remark}\label{remark-non-null-convex}
It is not ensured that the topology of $\mathcal{N}$ is Hausdorff. As an example, consider the $m$--dimensional Minkowski space $\mathbb{M}^{m}$, that is $\mathbb{R}^{m}$ with the metric $\mathbf{g}=-dx^1\otimes dx^1 + dx^2\otimes dx^2+\ldots+dx^m\otimes dx^m$, where $\left(x^1,\ldots,x^m\right)$ are the standard coordinates in $\mathbb{R}^m$. 
Now, if a point $p\in\mathbb{M}^{m}$ is removed obtaining another spacetime $\mathbb{M}_0$, then the resulting spacetime is still strongly causal, but its corresponding space of light rays $\mathcal{N}_0$ is not Hausdorff because there exists a sequence of light rays $\{ \gamma_n\}\subset \mathcal{N}_0$ converging to two different light rays limits $\mu,\lambda\in\mathcal{N}_0$. This example is illustrated in the figure \ref{diapositiva3}.
\end{remark}

\begin{figure}[h]
\centering
\begin{tikzpicture}[scale=1]
\foreach \n in {0.5,1,2,3,4}{
      \draw[thick] plot[samples=50, smooth, domain=-1:1.5] ({(\x)-(1/(\n))},\x ); } 
\draw[thick,color=blue] plot[samples=50, smooth, domain=0.3:1.5] (\x +0.3,\x ); 
\draw[thick,color=red] plot[samples=50, smooth, domain=-1:0.3] (\x +0.3,\x ); 
\draw (0.15,-0.3) node[anchor=east] {$\cdots$};
\draw (1.05,0.6) node[anchor=east] {$\cdots$};
\draw[thick,fill=white] (0.6,0.3) circle (2pt);
\draw (0.7,0.3) node[anchor=west] {$p$ (removed)};
\draw (1.3,0.9) node[anchor=west] {$\mu$};
\draw (-0.3,-0.7) node[anchor=west] {$\lambda$};
\draw (1.3,1.5) node[anchor=south] {$\gamma_n$};
\draw (-0.5,1.5) node[anchor=south] {$\gamma_1$};
\draw (0.5,1.5) node[anchor=south] {$\gamma_2$};
\draw (-3,0.5) node {$\mathbb{M}_0$};
\end{tikzpicture}
 \caption{The space of light rays of $\mathbb{M}_0$ is not Hausdorff. The light rays $\mu$ and $\lambda$ are limits of the sequence $\{ \gamma_n\}$.}
  \label{diapositiva3}
\end{figure}
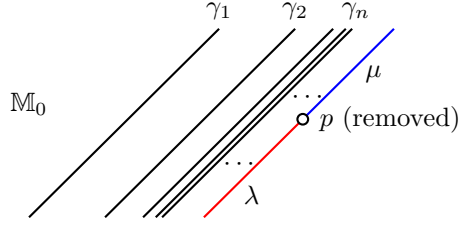

Now, we wonder what condition makes of the space of light rays of a strongly causal conformal manifold a Hausdorff manifold.
An answer to this question was given by R. Low in \cite[Prop. 3.2 et seq.]{Lo90b} and it is necessary to introduce the technical property of \emph{null pseudo--convexity} of the manifold.

\begin{definition}
A conformal manifold $M$ is said to be \emph{null pseudo--convex} if for any compact $K\subset M$ there exists a compact $K'\subset M$ such that any segment of light ray with endpoints in $K$ is contained in $K'$.
\end{definition}

In the example $\mathbb{M}_0$ of the remark \ref{remark-non-null-convex}, it is possible to consider a compact set $K=K_1 \cup K_2$, as in the figure \ref{diapositiva3-2}, such that any set $K'$ containing all segments of light rays with endpoints in $K$ can not be compact because the hole left by the removed point $p$ prevents it.

\begin{figure}[h]
\centering
\begin{tikzpicture}[scale=1]
\filldraw[draw=red,fill=red!20, opacity=0.5] (1.6,1.2) -- (2,0.8) -- (-0.2,-1.2) -- (-1.6,-1.2) -- (-2,-0.8) -- (0.2,1.2) -- cycle;
\filldraw[draw=emerald,fill=emerald!60, opacity=0.5] (0.4,0.5) -- (1.4,0.5) -- (1.4,1) -- (0.4,1) -- cycle;
\filldraw[draw=emerald,fill=emerald!60, opacity=0.5] (-0.4,-0.5) -- (-1.4,-0.5) -- (-1.4,-1) -- (-0.4,-1) -- cycle;
\draw (-1.2,-0.7) -- (0.45,0.95);
\draw (-1.1,-0.8) -- (0.6,0.9);
\draw (-0.5,-0.9) -- (1.2,0.8);
\draw (-0.4,-1) -- (1.3,0.7);
\draw (-1.2,-0.7) -- (0.45,0.95);
\draw[dashed] (-1,-1) -- (1,1);
\draw[thick,fill=white] (0,0) circle (2pt);
\draw[color=emerald] (1.95,0.8) node[anchor=east] {$K_1$};
\draw[color=emerald] (-1.95,-0.8) node[anchor=west] {$K_2$};
\draw[color=red] (1,-0.3) node[anchor=west] {$K'$};
\draw (-2,0.5) node {$\mathbb{M}_0$};
\draw (0.05,0) node[anchor=west] {$p$};
\end{tikzpicture}
 \caption{The space of light rays of $\mathbb{M}_0$ is not null pseudo--convex. Although $K=K_1 \cup K_2$ is compact, $K'$ can not be so.}
  \label{diapositiva3-2}
\end{figure}
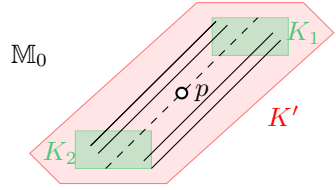

The mentioned proposition stated by Low is the following.
\begin{proposition}\label{prop-hausdorff-null}
Let $M$ be a strongly causal conformal manifold. The space of light rays $\mathcal{N}$ of $M$ is Hausdorff if and only if $M$ is null pseudo--convex.
\end{proposition}

A related result is that the lack of hausdorffness of $\mathcal{N}$ implies that there are naked singularities in $M$ \cite[Prop. 2.2]{Lo89}. 
Recall that a \emph{naked singularity} occurs at the future of a inextensible causal curve $\lambda$ if there exists a point $p\in M$ such that its chronological past sets verify $I^{-}(\lambda)\subset I^{-}(p)$. 

It is not difficult to find examples of spacetimes in which there exists naked singularities but the space of light rays is Hausdorff. 

\begin{example}
Consider the $3$--dimensional Minkowski spacetime $\mathbb{M}^3$ and its space of light rays $\mathcal{N}_{\mathbb{M}^3}$. Now, we restrict it to the cylinder given by
\[
M = \left\{ (t,x,y)\in \mathbb{M}^3: x^2+y^2<1 \right\}  .
\]

Clearly, the space of light rays $\mathcal{N}_{M}$ of $M$ is contained in $\mathcal{N}_{\mathbb{M}^3}$ because $M\subset \mathbb{M}^3$ is open and any light ray in $\mathcal{N}_{M}$ is a segment of light ray in $\mathcal{N}_{\mathbb{M}^3}$.
Observe that, for any inextensible light ray $\gamma\in \mathcal{N}_{M}$, there always exists a point $p=(t_0,0,0)\in M$ with $t_0\in \mathbb{R}$ big enough such that $\gamma\subset I^{-}(p)$ and, by corollary \ref{1C3}, we obtain that $I^{-}\left(\gamma\right)\subset I^{-}(p)$ and therefore $M$ is nakedly--singular (see figure \ref{figura-nakedly-sing}).

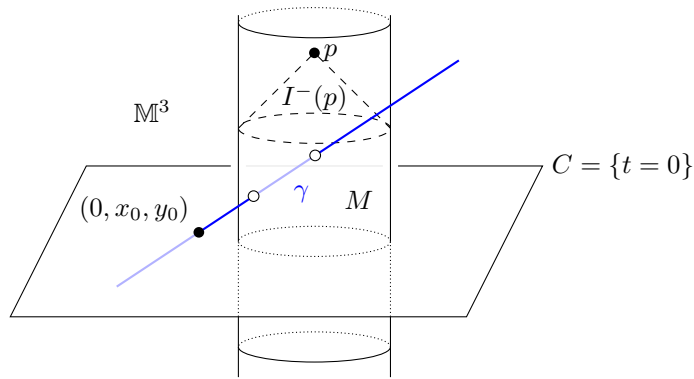
\begin{figure}[h]
\centering
\begin{tikzpicture}[scale=1]
\draw (-1.1,1) -- (-3,1) -- (-4,-1) -- (2,-1) -- (3,1) -- (1.1,1);  
\draw[color=gray!20] (-0.9,1) -- (0.9,1);
\draw (1,3) -- (1,0);        
\draw (-1,3) -- (-1,0);
\draw[densely dotted] (1,0) -- (1,-1);
\draw[densely dotted] (-1,0) -- (-1,-1);
\draw (1,-1) -- (1,-1.8);
\draw (-1,-1) -- (-1,-1.8);
\draw[densely dotted]	 (1,2.9) arc (0:180:1cm and 0.2cm);
\draw	 (1,2.9) arc (0:-180:1cm and 0.2cm);
\draw[densely dotted]	 (1,-1.4) arc (0:180:1cm and 0.2cm);
\draw	 (1,-1.4) arc (0:-180:1cm and 0.2cm);
\draw[densely dotted]	 (1,0) arc (0:360:1cm and 0.2cm);
\draw[dashed]	 (1,1.5) arc (0:360:1cm and 0.2cm);
\draw[dashed] (-1,1.5) -- (0,2.5) -- (1,1.5);
\fill (0,2.5) circle (2pt);
\draw (0,2.5) node[anchor=west] {$p$};
\draw (0,1.9) node {$I^{-}(p)$};
\draw[color=blue!30,thick] (-2.6,-0.6)--(-1.52,0.12);
\draw[color=blue,thick] (-1.52,0.12)--(-0.8,0.6) ;
\draw[color=blue!30,thick] (-0.8,0.6) -- (0.01,1.14);
\draw[color=blue,thick] (0.01,1.14) -- (1.9,2.4);
\filldraw[draw=black,fill=white] (-0.8,0.6) circle (2pt);
\filldraw[draw=black,fill=white] (0.01,1.14) circle (2pt);
\fill (-1.52,0.12) circle (2pt);
\draw[color=blue] (-0.395,0.87) node[anchor=north west] {$\gamma$};
\draw (-1.52,0.12) node[anchor=south east] {$(0,x_0,y_0)$};
\draw (0.9,0.3) node[anchor=south east] {$M$};
\draw (-2.5,2) node[anchor=north west] {$\mathbb{M}^3$};
\draw (3,1) node[anchor=west] {$C=\{t=0\}$};
\end{tikzpicture}
 \caption{$M$ is nakedly--singular but $\mathcal{N}_M$ is Hausdorff.}
  \label{figura-nakedly-sing}
\end{figure}

The light ray $\gamma\in\mathcal{N}_{M}$ can be written by $\gamma(s)=(s,x_0+s\cdot\cos\theta_0,y_0+s\cdot\sin\theta_0)$ for $s\in(a,b)\subset \mathbb{R}$ an open interval. We are using $\gamma\simeq(x_0,y_0,\theta_0)$ as coordinates for $\mathcal{N}_{\mathbb{M}^3}$ and $\mathcal{N}_{M}$. Then we have
\[
\left(x_0+s\cdot\cos\theta_0\right)^2+ \left(y_0+s\cdot\sin\theta_0\right)^2 <1
\]
and because this is an open condition, there exist $\delta>0$, $\epsilon>0$ and $\eta>0$ such that if $\vert x-x_0 \vert <\delta$, $\vert y-y_0 \vert <\epsilon$ and $\vert \theta-\theta_0 \vert <\eta$ then the light ray $\beta(s)=(s,x+s\cdot\cos\theta,y+s\cdot\sin\theta)$ intersects $M$ for $s\in(a',b')\subset \mathbb{R}$ open. Then $\mathcal{N}_{M}\subset\mathcal{N}_{\mathbb{M}^3}$ is an open set contained in a Hausdorff space, therefore $\mathcal{N}_{M}$ is also Hausdorff.
\end{example}  

The assumption of $M$ being strongly causal is not necessary to obtain that $\mathcal{N}$ is Hausdorff. In fact, it is possible to find spacetimes which do not verify the chronological condition, that is, with closed timelike curves, but its space of light rays is still Hausdorff. 
For example, Zollfrei manifolds $Z$ are compact Lorentz manifold such that all its null geodesics are closed. The compacity of $Z$ implies that  such manifolds do not satisfy the chronological condition \cite[Thm. 3.6]{Mi08} but their spaces of light rays $\mathcal{N}_Z$ are, indeed, smooth manifolds (see \cite{Gu89}, \cite{Su13}, \cite{Ma21}).

The following example shows a simple non--compact non--chronological spacetime with Hausdorff space of light rays.

\begin{example}
Let $\mathbb{M}^2$ be the $2$--dimensional Minkowski spacetime equipped with the metric $\mathbf{g}=-dt\otimes dt + dx\otimes dx$ in standard coordinates $(t,x)\in \mathbb{R}^2$. Consider the restriction
\[
M = \left\{ (t,x)\in \mathbb{M}^2: \vert t \vert \leq 1 \right\}  
\]
identifying the points $(1,x)\sim (-1,x)$. For fixed $y_0\in \mathbb{R}$, trivially, the curve $\lambda(s)=  (s,y_0)$ is a closed timelike curve, then $M$ do not satisfy any causality condition enumerated in section \ref{sec:Causalcurves}.
Observe that, any light ray can be defined by a null geodesic $\gamma_{\pm}(s)=(s,x_0 \pm s)$ where the signs $+$ or $-$ determine if the light rays travel to the right or to the left respectively (see figure \ref{figura-vicious-hausdorff}).  

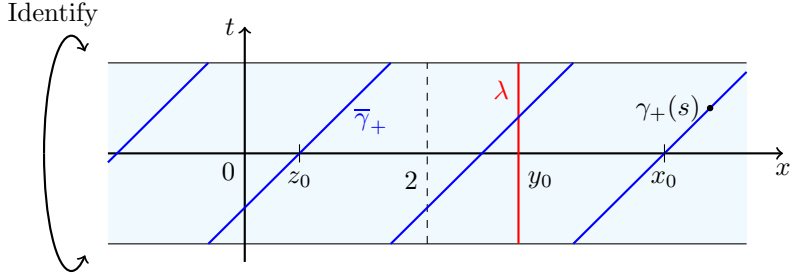
\begin{figure}[h]
\centering
\begin{tikzpicture}[scale=1.2]
\fill[color=cyan!5] (-1.5,-1) -- (5.5,-1) -- (5.5,1) -- (-1.5,1) -- cycle;  
\draw (-1.5,-1) -- (5.5,-1);
\draw (5.5,1) -- (-1.5,1);
\draw[->,thick] (0,-1.2) -- (0,1.4) node[anchor=east] {$t$};
\draw[->,thick] (-1.5,0) -- (5.9,0) node[anchor=north] {$x$};
\draw[dashed] (2,-1) -- (2,1) ;
\draw[color=red,thick] (3,-1) -- (3,1);
\draw[color=blue,thick] (-1.5,-0.1) -- (-0.4,1);
\draw[color=blue,thick] (-0.4,-1) -- (1.6,1);
\draw[color=blue,thick] (1.6,-1) -- (3.6,1);
\draw[color=blue,thick] (3.6,-1) -- (5.5,0.9);
\draw (3,-0.1) node[anchor=north west] {$y_0$};
\draw[color=red] (3,0.7) node[anchor=east] {$\lambda$};
\draw (0.6,-0.1) node[anchor=north] {$z_0$} -- (0.6,0.1);
\draw (4.6,-0.1) node[anchor=north] {$x_0$} -- (4.6,0.1);
\draw (0,0) node[anchor=north east] {$0$};
\draw (2,-0.1) node[anchor=north east] {$2$};
\draw (5.1,0.5) node[anchor=east] {$\gamma_{+}(s)$};
\draw[color=blue] (1.1,0.6) node[anchor=north west] {$\overline{\gamma}_{+}$};
\fill (5.1,0.5) circle (1pt);
\draw[->,thick]	(-2.2,0) arc (180:60:0.3cm and 1.3cm);
\draw[->,thick]	(-2.2,0) arc (-180:-60:0.3cm and 1.3cm);
\draw(-2.1,1.3) node[anchor=south] {Identify};
\end{tikzpicture}
 \caption{$M$ is not chronological but $\mathcal{N}_M$ is Hausdorff.}
  \label{figura-vicious-hausdorff}
\end{figure}

Notice that, by the identification $(1,x)\sim (-1,x)$, then $t=t+2k$ for all integer $k\in \mathbb{Z}$, and moreover, since for any $x_0\in \mathbb{R}$ there exists $k_0\in \mathbb{Z}$ such that $z_0=x_0-2k_0\in\left[0,2\right)$ then 
\[
\gamma_{+}(s)=(s,x_0+s)=(s+2k_0,x_0+s)=(s+2k_0,z_0+2k_0+s)=\overline{\gamma}_{+}(s+2k_0)=\overline{\gamma}_{+}(\tau)
\]
where $\tau=s+2k_0$ and $\gamma_{+}$ and $\overline{\gamma}_{+}$ are the same light ray but with two different parametrizations. This shows that the set $\mathcal{N}_{+}$ of light rays of $M$ travelling to the right can be globally determined by $z_0\in[0,2)\subset \mathbb{R}$. Analogously, we have $\mathcal{N}_{-} \simeq[0,2)$ for the light rays travelling to the left. 
Then $\mathcal{N}=\mathcal{N}_{+} \cup \mathcal{N}_{-}$ is diffeomorphic to two disjoint copies of the interval $[0,2)$ therefore $\mathcal{N}$ is Hausdorff. 
\end{example}

In virtue of proposition \ref{prop-hausdorff-null}, from now on, we will assume that the conformal manifold $M$ is strongly causal and null pseudo--convex.

\subsection{Tangent spaces \texorpdfstring{$T_{\gamma}\mathcal{N}$}{TN}}\label{sec:lightrays-tangent}

It is necessary to describe the tangent vectors of $\mathcal{N}$ in such a way that we can carry out calculations with them.
A detailed exposition of this description can be found in \cite[Sec. 3]{Ba15b}. 
We will only introduce to the reader a glimpse of the matter.

First, observe that a tangent vector $\mathbf{v}\in T_{\gamma}\mathcal{N}$ can be defined by the derivative of a smooth curve $\Gamma:\left(-\epsilon,\epsilon\right)\rightarrow \mathcal{U}\subset\mathcal{N}$ such that $\Gamma\left(s\right)=\gamma_s\in \mathcal{N}$ where $\gamma_0=\gamma$, $\Gamma'\left(0\right)=\mathbf{v}$ and $\mathcal{U}$ is a neighbourhood of $\gamma=\gamma_0$ diffeomorphic to $\Omega(C)\subset \mathbb{N}$ by the map $\xi$ in (\ref{diagram-charts}). 
But this curve $\Gamma$ can be got by a variation of null geodesics $\mathbf{f}:\left(-\epsilon,\epsilon\right)\times I\rightarrow M$ such that $\mathbf{f}\left(s,\tau\right)=\gamma_s\left(\tau\right)$ where each $\gamma_s$ is parametrized as a null geodesic with initial vector $\gamma'_s(0)\in \Omega(C)\subset \mathbb{N}$. 
Then $\Gamma'\left(0\right)\in T_{\gamma}\mathcal{N}$ corresponds with the variational field $J\in \mathfrak{X}_{\gamma}$ of $\mathbf{f}$ along $\gamma$, that is 
\[
\displaystyle{ J\left(\tau\right)= \frac{\partial \mathbf{f}}{\partial s}\left(0,\tau\right) \in T_{\gamma\left(\tau\right)}M } 
\]
where $\mathfrak{X}_{\gamma}$ denotes the set of all smooth vector fields along $\gamma$.

Notice that, by \cite[Lem. 8.3]{On83}, this variational vector field along $\gamma$ satisfies the equation of \emph{Jacobi fields} given by
\begin{equation}\label{eq-Jacobi-field}
J'' + \mathbf{R}\left(J,\gamma'\right)\gamma' = 0
\end{equation}
where prime $'$ over $J$ denotes the covariant derivative $\frac{D}{d\tau}$ along $\gamma$ respect to the parameter $\tau$, and $\mathbf{R}$ is the Riemann curvature tensor\footnote{The expression of the equation of Jacobi fields depends on the sign of definition of the Riemann curvature tensor $\mathbf{R}$ and on the order of the arguments. So, \[ J'' - \mathbf{R}\left(J,\gamma'\right)\gamma' = 0, \hspace{2mm}J'' + \mathbf{R}\left(\gamma',J\right)\gamma' = 0, \hspace{2mm}J'' - \mathbf{R}\left(\gamma',J\right)\gamma' = 0 \] are other expressions we can find in the literature.}.
The solutions of equation (\ref{eq-Jacobi-field}) are called \emph{Jacobi fields along the geodesic $\gamma$} and they are fully determined by fixing initial vectors $J\left(0\right)=u\in T_{\gamma\left(0\right)}M$ and $J'\left(0\right)=v\in T_{\gamma\left(0\right)}M$.

A way to construct a variation of null geodesics is using the exponential map for a fixed auxiliary metric $\mathbf{g}\in\mathcal{C}$. For a base curve $\lambda:I\rightarrow M$ and a lightlike vector field $W:I\rightarrow \mathbb{N}^{+}$ such that $W(s)\in \mathbb{N}_{\lambda(s)}^{+}$ we can define $\mathbf{f}\left( s,\tau\right) =\mathrm{exp}_{\lambda \left( s\right) }\left(\tau W\left( s\right) \right)$. If $\lambda(0)=p\in M$ and $W(0)=v\in\mathbb{N}_{p}^{+}$, then $\mathbf{f}$ is a variation of the light ray $\gamma=\gamma_{[v]}\in\mathcal{N}$.

\begin{lemma}
\label{lemmaDC92} Let $\lambda :\left( -\epsilon ,\epsilon
\right)\rightarrow M$ be a smooth curve and $W\left( s\right)\in T_{\lambda(s)}M$ a non-vanishing vector field along $\lambda$. 
Then, the Jacobi field $J$ of the geodesic variation given by
\begin{equation}
\mathbf{f}\left( s,\tau\right) =\mathrm{exp}_{\lambda \left( s\right) }\left(\tau W\left( s\right) \right)
\end{equation}%
along $\gamma(\tau)=\mathbf{f}\left( 0,\tau\right) =\mathrm{exp}_{\lambda \left( 0\right) }\left(\tau W\left( 0\right) \right)$ has initial vectors $J\left( 0\right) =\lambda'(0)$ and $J'\left( 0\right)=\frac{DW}{ds}\left( 0\right)$ .
\end{lemma}

\begin{proof}
Notice that $J(0)=\frac{\partial\mathbf{f}}{\partial s}\left(0,0\right)$ is the tangent vector to the curve $\mathbf{f}\left(s,0\right)$ at $s=0$, then $\mathbf{f}\left(s,0\right)=\mathrm{exp}_{\lambda\left(s\right)}\left(0\cdot W\left(s\right)\right)= \lambda\left(s\right)$, then we have
\[
J\left(0\right)=\lambda ^{\prime }\left(0\right)
\]
On the other hand, recall that $J'\left(0\right)=\frac{D}{d \tau}\frac{\partial\mathbf{f}}{\partial s}\left(0,0\right)$ and  $\frac{D}{ds}\frac{\partial\mathbf{f}}{\partial \tau}\left(0,0\right)$ is the covariant derivative of $\frac{\partial\mathbf{f}}{\partial \tau}\left(s,0\right)=W\left(s\right)$ for $s=0$ along the curve $\mathbf{f}\left(s,0\right)=\lambda\left(s\right)$. Therefore, we have
\[
J'\left(0\right)=\frac{D}{d\tau}\frac{\partial\mathbf{f}}{\partial s}\left(0,0\right)=\frac{D}{ds}\frac{\partial\mathbf{f}}{\partial \tau}\left(0,0\right)= \frac{DW}{ds}\left(0\right)    .
\]
\qed
\end{proof}

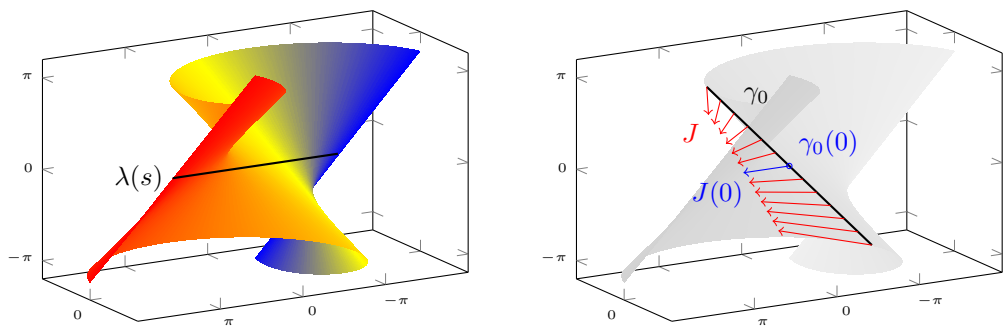
\begin{figure}[h]
\centering
\begin{subfigure}[t]{0.45\textwidth}
\centering
\begin{tikzpicture}[scale=1]
  \begin{axis}[xtick = {-180,0,180},xticklabels={$-\pi$,$0$,$\pi$},
    ytick = {0},ztick = {-180,0,180},zticklabels={$-\pi$,$0$,$\pi$},
    ticklabel style = {font = \tiny},axis equal image,
    view={150}{15}]
    \addplot3 [surf,domain=-180:180,samples=120,y domain=-180:180,samples y=2,line join=round,
       shader=interp,
      variable=\t,point meta={t}]
   ( {t+cos(t)*y},{sin(t)*y},{y} );
   \addplot3[thick] coordinates {( -180,0,0 ) (180,0,0)};
   \addplot3[color=black] ( 180,0,0 ) node[anchor=east] {$\lambda(s)$};
  \end{axis}
\end{tikzpicture}
\caption{$\gamma_{s}(\tau)=\mathbf{f}(s,\tau)=(\tau,s+\tau\cos(s),\tau\sin(s))$ is a variation of light rays in $\mathbb{M}^3$ defined by the null vectors $W(s)=(1,\cos(s),\sin(s))\in \mathbb{N}_{\lambda(s)}$ with base curve $\lambda(s)=(0,s,0)\in M$. Observe that each single colour corresponds to one value of the parameter $s$.}
  \label{figura-variacion-rayos-luz}
  \end{subfigure}
\hspace{5pt}
  \begin{subfigure}[t]{0.45\textwidth}
\centering
\begin{tikzpicture}[scale=1]
  \begin{axis}[xtick = {-180,0,180},xticklabels={$-\pi$,$0$,$\pi$},
    ytick = {0},ztick = {-180,0,180},zticklabels={$-\pi$,$0$,$\pi$},
    ticklabel style = {font = \tiny},axis equal image,colormap={custom}{rgb255(0cm)=(240,240,240);
            rgb255(1cm)=(230,230,230);rgb255(2cm)=(220,220,220);rgb255(3cm)=(210,210,210);},    
    view={150}{15}]
    \addplot3 [surf,domain=-180:180,samples=120,y domain=-180:180,samples y=2,line join=round,
       shader=interp,
      variable=\t,point meta={t}]
   ( {t+cos(t)*y},{sin(t)*y},{y} );
   \addplot3[thick] coordinates {( -180,0,-180 ) (180,0,180)};
   \addplot3[->,color=red] coordinates {( -180,0,-180 ) (100-180,-180,-180)};
   \addplot3[->,color=red] coordinates {( -150,0,-150 ) (100-150,-150,-150)};
   \addplot3[->,color=red] coordinates {( -120,0,-120 ) (100-120,-120,-120)};
   \addplot3[->,color=red] coordinates {( -90,0,-90 ) (100-90,-90,-90)};
   \addplot3[->,color=red] coordinates {( -60,0,-60 ) (100-60,-60,-60)};
   \addplot3[->,color=red] coordinates {( -30,0,-30 ) (100-30,-30,-30)};
      \addplot3[->,color=blue] coordinates {( 0,0,0 ) (100,0,0)};
   \addplot3[->,color=red] coordinates {( 180,0,180 ) (100+180,180,180)};
   \addplot3[->,color=red] coordinates {( 150,0,150 ) (100+150,150,150)};
   \addplot3[->,color=red] coordinates {( 120,0,120 ) (100+120,120,120)};
   \addplot3[->,color=red] coordinates {( 90,0,90 ) (100+90,90,90)};
   \addplot3[->,color=red] coordinates {( 60,0,60 ) (100+60,60,60)};
   \addplot3[->,color=red] coordinates {( 30,0,30 ) (100+30,30,30)};
   \addplot3[color=blue] ( 0,0,0 ) node[anchor=south west] {$\gamma_0(0)$};
   \addplot3[color=blue] (80,0,0) node[anchor=north east] {$J(0)$};
   \addplot3[color=black] ( 120,0,120 ) node[anchor=south west] {$\gamma_0$};
   \addplot3[color=red] (100+180,180,180) node[anchor=north east] {$J$};
   \addplot3[color=blue] ( 0,0,0) circle (1pt);
     \end{axis}
\end{tikzpicture}
\caption{The Jacobi field $J(\tau)=\frac{\partial\mathbf{f}}{\partial s}\left(0,\tau\right)=(0,1,\tau)$ along the light ray $\gamma_0\in\mathcal{N}$ defined by the variation $\mathbf{f}$ in figure \ref{figura-variacion-rayos-luz}.}
  \label{figura-Jacobi-rayos-luz}
  \end{subfigure}
  \caption{Variation of light rays in $\mathbb{M}^3$.}
  \label{figura-rayos-luz}
\end{figure}

Recall that a light ray can be defined by a null geodesic related to some metric $\mathbf{g}\in\mathcal{C}$, then a change of its affine parameter or of metric in $\mathcal{C}$ results on a different Jacobi field $\overline{J}$ along the corresponding null geodesic $\overline{\gamma}$, but they are related by 
\[
\overline{J}=J\left(\mathrm{mod}\hspace{1mm}\gamma'\right)  .
\]
Therefore, a tangent vector $\mathbf{v}\in T_{\gamma}\mathcal{N}$ can be represented by the \emph{class} of Jacobi fields along $\gamma$ modulo $\gamma'$ defined by the variational field $J\left(\tau\right)= \frac{\partial \mathbf{f}}{\partial s}\left(0,\tau\right)$ of a null geodesic variation $\mathbf{f}$ for any auxiliary metric $\mathbf{g}\in\mathcal{C}$ and any affine parameter.

Indeed, let $\mathcal{J}\left(\gamma\right)$ be the vector space of Jacobi fields along the null geodesic $\gamma\in \mathcal{N}$.
If $\mathbf{f}(s,\tau)=\gamma_s (\tau)$ is a variation of null geodesics where $\gamma_0 = \gamma$, then the Jacobi field $J\in\mathcal{J}\left(\gamma\right)$  along $\gamma$ corresponding to $\mathbf{f}$ verifies 
\begin{equation}\label{eq-J-gamma-constant}
\mathbf{g}\left(J\left(\tau\right),\gamma'\left(\tau\right)\right)= \text{constant for all } \tau\in I .
\end{equation}
This fact can be straightforwardly proven because $J$ is a variational field of a geodesics variation such that $\mathbf{g}\left(\gamma'_s(\tau),\gamma'_s(\tau)\right)=0$ constant for any geodesic $\gamma_s$ in the variation.
See \cite[Lem. 2.1]{Ch18} or \cite[p. 142-143]{Ba15b} for proofs.

Notice that, since $\mathbf{g}\left(J\left(\tau\right),\gamma'\left(\tau\right)\right)$ is constant for all $\tau\in I$, then 
\[
0=\frac{d}{d\tau} \mathbf{g}\left(J\left(\tau\right),\gamma'\left(\tau\right)\right)=\mathbf{g}\left(J'\left(\tau\right),\gamma'\left(\tau\right)\right)+\mathbf{g}\left(J\left(\tau\right),\gamma''\left(\tau\right)\right)
\]
and since $\gamma''\left(\tau\right)=\frac{D\gamma'}{d\tau}\left(\tau\right)=0$, therefore
\begin{equation}\label{eq-Jprima-gammaprima}
\mathbf{g}\left(J'\left(\tau\right),\gamma'\left(\tau\right)\right)=0   .
\end{equation}

Now, we denote by 
\[
\mathcal{J}_L\left(\gamma\right)=\left\{ J\in \mathcal{J}\left(\gamma\right): \mathbf{g}\left(J\left(t\right),\gamma'\left(t\right)\right)=\mathrm{constant} \right\}\simeq \mathbb{R}^{2m-1}
\]
the vector space of Jacobi fields of null geodesic variations of $\gamma$ and if 
\[
\mathcal{J}_0\left(\gamma\right)=\left\{ J\in \mathcal{J}_L\left(\gamma\right): J\left(t\right)=\left(\alpha t + \beta\right)\gamma'\left(t\right), \alpha,\beta\in \mathbb{R} \right\}\simeq \mathbb{R}^{2}
\]
is the vector subspace of $\mathcal{J}_L\left(\gamma\right)$ of Jacobi fields proportional to $\gamma'$, where $\mathbf{g}\in\mathcal{C}$ and the parametrization of $\gamma\in \mathcal{N}$ are auxiliary elements, then we have the following statement whose proof can be seen in \cite[Prop. 3.16]{Ba15b}. 
\begin{proposition}
$T_{\gamma}\mathcal{N}$ is isomorphic to $\mathcal{L}\left(\gamma\right)=\mathcal{J}_L\left(\gamma\right) / \mathcal{J}_0\left(\gamma\right)\simeq \mathbb{R}^{2m-3}$.
\end{proposition}

We will denote the class of Jacobi fields which is identified with the corresponding vector in $T_{\gamma}\mathcal{N}$ by  $\langle J \rangle =\left\{ \overline{J}\in \mathcal{J}_L\left(\gamma\right): \overline{J}=J~(\mathrm{mod}~\gamma') \right\}\in \mathcal{L}\left(\gamma\right)$ so, by abusing of the notation, we will say $\langle J \rangle\in T_{\gamma}\mathcal{N}$.

\subsection{Contact structure in \texorpdfstring{$\mathcal{N}$}{N}}\label{sec:lightrays-contact}

The space of light rays $\mathcal{N}$ has an additional canonical structure: a \emph{contact structure} $\mathcal{H}\subset T\mathcal{N}$.
For a broad introduction about contact structures, see \cite[Appx. 4]{Ar89} or \cite[Ch. 5]{LM87}.

Consider a $\left(2k+1\right)$--dimensional smooth manifold $P$ with a smooth distribution of hyperplanes $\mathcal{H}\subset TP$. 
It is known that $\mathcal{H}$ can be locally defined as the kernel of a $1$--form $\alpha$, that is $\mathcal{H}_p = \mathrm{ker}\hspace{1mm}\alpha_p$ (see \cite[Lem. 1.1.1]{G08}).

We will say that $\mathcal{H}$ is \emph{maximally non--integrable} if there is no integral submanifold $N\subset P$ such that $T_q N = \mathcal{H}_q$ for all $q\in N$. This property can be characterized by the contact form $\alpha$ defining it, so $\mathcal{H}$ is maximally non--integrable if and only if $\alpha$ restricted to $\mathcal{H}$ is non--degenerate, but this is equivalent to $\alpha \wedge \left( d\alpha \right)^k \neq 0$ (see \cite[Prop. 10.3]{Ca01} for proof).

\begin{definition}
Given a $\left(2k+1\right)$--dimensional smooth manifold $P$, a \emph{contact structure} $\mathcal{H}=\mathrm{ker}\hspace{1mm}\alpha$ in $P$ is a smooth distribution of hyperplanes $\mathcal{H}_{p}\subset T_{p}P$ such that $\alpha \wedge \left( d\alpha \right)^k \neq 0$. Each hyperplane $\mathcal{H}_p$ is called a \emph{contact element} or \emph{contact hyperplane} and the $1$--form $\alpha$ is called a \emph{contact form}.
\end{definition}

Notice that a contact form defining a contact structure is not unique, in fact, given a contact form $\alpha$, for every non--vanishing differentiable function $f\in \mathfrak{F}\left(P\right)$ then $f\alpha$ is also a contact form defining the same contact structure $\mathcal{H}$ (see \cite[Sect. V.4.1]{LM87}).

In the case we are studying, $P=\mathcal{N}$ and $k=m-2$. Let us fix some auxiliary metric $\mathbf{g}\in \mathcal{C}$ and consider a coordinate chart $\left(\mathcal{N}_V, \psi\right)$, as in equation (\ref{eq-chart-N}). Recall that, according to remark \ref{remark-v-en-Omega}, we are considering the light rays in $\mathcal{N}_V$ as parametrized as null geodesics such that their initial vectors at the local Cauchy surface $C\in V$ are in $\Omega(C)\in \mathbb{N}$. 
Prevented from this, in \cite[Sec. 4]{Ba15b}, a contact form $\alpha$ is constructed using the quotient $\mathcal{N}\simeq \mathbb{PN}/\mathcal{D}_G$ described in section \ref{sec:lightrays-struct-N}, obtaining that 
\begin{equation}\label{eq-alpha-contact}
\alpha\left(\langle J \rangle\right)=\mathbf{g}\left(J,\gamma'\right) \text{ for } \langle J \rangle\in T_{\gamma}\mathcal{N}   .
\end{equation}
Although the expression (\ref{eq-alpha-contact}) for $\alpha$ depends on the affine parameter of $\gamma$, since the affine parameter is fixed, then $\alpha$ is well--defined as a contact form.  
Observe that, regardless of parameterization, the sign of $\alpha\left(\langle J \rangle\right)$ is unambiguously and globally determined for any $J$ in its equivalence class $\langle J \rangle$, then it induces a co--orientation in $\mathcal{N}$. The existence of this co--orientation is equivalent to the existence of contact form globally defined  \cite[Lem. 1.1.1]{G08}. See also \cite[Sect.~4]{Ba15b} for a explicit construction of a global contact form $\alpha$ in $\mathcal{N}$.
The kernel of $\alpha$ at $\gamma\in\mathcal{N}$ defines the contact hyperplane $\mathcal{H}_{\gamma}$ at $\gamma$ by
\begin{equation}\label{eq-contact-struct}
\mathcal{H}_{\gamma}= \left\{\langle J \rangle\in T_{\gamma}\mathcal{N}: \mathbf{g}\left(J , \gamma' \right)= 0  \right\}   .
\end{equation}

The construction of $\mathcal{H}$ can also be done in other ways. For example, in \cite[Sec. 5]{Ba15b} the contact structure is also obtained by co--isotropic reduction and, in \cite{Ke09}, the symplectic reduction is used. In both cases, basic elements of symplectic theory are needed. 

The calculus of the exterior derivative of $\alpha$, carried out in \cite[Sec. 4]{Ba15b}, shows that $\left.\omega\right|_{\mathcal{H}}=\left.-d\alpha\right|_{\mathcal{H}}$ is non--degenerated, and it can be written by 
\begin{equation}\label{eq-dalpha-expression}
\omega_{\gamma}\left(\langle J \rangle,\langle K \rangle\right)=\mathbf{g}\left(J\left(0\right),K'\left(0\right)\right)-\mathbf{g}\left(J'\left(0\right),K\left(0\right)\right) \qquad \text{ for } \langle J \rangle , \langle K \rangle\in T_{\gamma}\mathcal{N}
\end{equation}
where the initial vectors $J\left(0\right),K\left(0\right),J'\left(0\right),K'\left(0\right)\in T_{\gamma\left(0\right)}C$ are tangent to the local Cauchy surface $C\subset U$ in $M$ and where $\gamma$ is parametrized such that $\gamma'(0)\in \Omega(C)$ according to remark \ref{remark-v-en-Omega}.

This implies that $\left(\mathcal{H}_{\gamma},\left.\omega\right|_{\mathcal{H}_{\gamma}}\right)$ is a symplectic vector space for each $\gamma\in \mathcal{N}$.

\begin{example}\label{ex-contact-4M}
Consider the $4$--dimensional Minkowski spacetime $\mathbb{M}^{4}$ with all the notation as in example \ref{ex-4-minkowski}.  If $\gamma\in \mathcal{N}_{\mathbb{M}^4}$ is the light ray with coordinates $(x_0,y_0,z_0,\theta_0, \phi_0)$, then we can write 
\[
\gamma(s)=\left(s,x_0+s \cos\theta_0\sin\phi_0,y_0+s \sin\theta_0\sin\phi_0,z_0+s \cos\phi_0  \right)\in \mathbb{M}^{4}
\]
where $\gamma(0)\in C\equiv\{t=0\}$.  Calling $v(\theta,\phi)=\left(1,\cos \theta \sin \phi , \sin \theta \sin \phi ,  \cos \phi  \right)$, then for fixed $s$ and $\left(\theta, \phi\right)$ we can define 
\[
\mu\left(\theta, \phi, s ,\tau\right) = \gamma\left(s\right) + \tau \cdot v(\theta,\phi)  .
\]
Observe that $\mu$ defines a variation of null geodesics for the parameters $s$, $\theta$ and $\phi$, where $\tau$ is the affine parameter, that is $\mu_{\left(\theta, \phi, s \right)}(\tau)=\mu\left(\theta, \phi, s ,\tau\right)$ is the corresponding null geodesic of the variation. It verifies $\mu_{\left(\theta_0,\phi_0,s\right)}(0)=\gamma(s)$ and $\mu'_{\left(\theta_0,\phi_0,s\right)}(0)=v(\theta_0,\phi_0)$.

It is straightforward to check that 
\begin{equation}\label{eq-Jacobi-4M}
\left\{
\begin{tabular}{l}
$\frac{\partial \mu}{\partial \theta}\left(\theta_0,\phi_0,s,\tau\right)=(0,-\tau\sin\theta_0 \sin\phi_0,\tau \cos\theta_0\sin\phi_0,0)$ \\
$\frac{\partial \mu}{\partial \phi}\left(\theta_0,\phi_0,s,\tau\right)=(0,\tau\cos\theta_0 \cos\phi_0,\tau \sin\theta_0\cos\phi_0,-\sin\phi_0)$
\end{tabular}
\right.
\end{equation}
so, we have 
\[
\left\{
\begin{tabular}{l}
$\mathbf{g}\left( \frac{\partial \mu}{\partial \theta}\left(\theta_0,\phi_0,s,\tau\right),\gamma'(s)\right) =0 $ \\
$\mathbf{g}\left(\frac{\partial \mu}{\partial \phi}\left(\theta_0,\phi_0,s,\tau\right),\gamma'(s)\right) =0$
\end{tabular}
\right.
\]
for all $s\in \mathbb{R}$. Then $J_{\left(s,\theta_0,\phi_0\right)}^{\theta}(\tau)=\frac{\partial \mu}{\partial \theta}\left(\theta_0,\phi_0,s,\tau\right)$ and $J_{\left(s,\theta_0,\phi_0\right)}^{\phi}(\tau)=\frac{\partial \mu}{\partial \phi}\left(\theta_0,\phi_0,s,\tau\right)$ are Jacobi fields along $\gamma$ defining vectors in the contact hyperplane $\mathcal{H}_{\gamma}$.

We can notice that $\mu\left(\theta, \phi, s ,-s\right)\in C$, and therefore the coordinates of the light ray defined by $\mu_{\left(\theta, \phi, s \right)}\in \mathcal{N} $ are
\[
\left\{
\begin{array}{l}
x(\mu_{\left(\theta, \phi, s \right)}) =  x_0 +s\left( \cos \theta_0 \sin \phi_0 - \cos \theta \sin \phi  \right)  \\
y(\mu_{\left(\theta, \phi, s \right)}) =  y_0 +s\left( \sin \theta_0 \sin \phi_0 - \sin \theta \sin \phi  \right)  \\
z(\mu_{\left(\theta, \phi, s \right)}) =  z_0 +s\left( \cos \phi_0  - \cos \phi  \right)  \\
\theta(\mu_{\left(\theta, \phi, s \right)}) =  \theta \\
\phi(\mu_{\left(\theta, \phi, s \right)}) =  \phi 
\end{array}
\right. 
\]

If we compute the derivatives of the expressions above with respect to $\theta$ and $\phi$ at $\left(\theta_0, \phi_0 ,s\right)$, we obtain the expression in coordinates of the tangent vectors in $T_{\gamma}\mathcal{N}$ defined by $J_{\left(s,\theta_0,\phi_0\right)}^{\theta}$ and $J_{\left(s,\theta_0,\phi_0\right)}^{\phi}$. Then we have
\[
\left\{
\begin{array}{l}
\langle J_{\left(0,\theta_0,\phi_0\right)}^{\theta} \rangle  = \textstyle{ \left( \frac{\partial}{\partial \theta} \right)_{\gamma}  } \\
\langle J_{\left(s,\theta_0,\phi_0\right)}^{\theta} \rangle  = \textstyle{ s\left( \sin\theta_0 \sin \phi_0 \left( \frac{\partial}{\partial x} \right)_{\gamma} - \cos\theta_0 \sin \phi_0 \left( \frac{\partial}{\partial y} \right)_{\gamma} \right) + \left( \frac{\partial}{\partial \theta} \right)_{\gamma}  } \\
\langle J_{\left(0,\theta_0,\phi_0\right)}^{\phi} \rangle  = \textstyle{ \left( \frac{\partial}{\partial \phi} \right)_{\gamma}   } \\
\langle J_{\left(s,\theta_0,\phi_0\right)}^{\phi} \rangle  = \textstyle{ s\left( -\cos\theta_0 \cos \phi_0 \left( \frac{\partial}{\partial x} \right)_{\gamma} - \sin\theta_0 \cos \phi_0 \left( \frac{\partial}{\partial y} \right)_{\gamma}  +\sin \phi_0 \left( \frac{\partial}{\partial z} \right)_{\gamma} \right) + \left( \frac{\partial}{\partial \phi} \right)_{\gamma}   }
\end{array}
\right. 
\]
which are linearly independent for $s\neq 0$. Therefore, since $\mathrm{dim}\left(\mathcal{H}_{\gamma}\right)=4$, these four vectors have to be a basis of the hyperplane $\mathcal{H}_{\gamma}$. 
Whence the contact hyperplane at $\gamma$ can be written as
\begin{align*}
\mathcal{H}_{\gamma} & = \textstyle{ \mathrm{span}\left\{  \sin\theta_0 \left( \frac{\partial}{\partial x} \right)_{\gamma} - \cos\theta_0 \left( \frac{\partial}{\partial y} \right)_{\gamma}, \left( \frac{\partial}{\partial \theta} \right)_{\gamma}, \left( \frac{\partial}{\partial \phi} \right)_{\gamma} , \right. } \\
& \textstyle{ \left. \cos\theta_0 \cos \phi_0 \left( \frac{\partial}{\partial x} \right)_{\gamma} + \sin\theta_0 \cos \phi_0 \left( \frac{\partial}{\partial y} \right)_{\gamma}  - \sin \phi_0 \left( \frac{\partial}{\partial z} \right)_{\gamma} \right\}  }
\end{align*}
and a contact form whose kernel is $\mathcal{H}$ can be
\[
\alpha = \cos \theta \sin \phi \cdot dx + \sin \theta \sin \phi \cdot dy + \cos \phi \cdot dz \, .
\]
\end{example}

\begin{example}\label{example-contact-3M}
We can find the contact structure of the $3$--dimensional Minkowski spacetime $\mathbb{M}^3$ by restriction of $\mathbb{M}^4$ to the hyperplane $z=0$ and the angle $\phi$ to describe null directions must be fixed to $\phi=\frac{\pi}{2}$.  Then, the variation of null geodesic is written by 
\[
\mu\left(\theta, s ,\tau\right) =\left( s+\tau , x_0 +s\cos \theta_0 +\tau \cos \theta , y_0 +s \sin \theta_0 +\tau \sin \theta  \right)\in M
\] 
and whence the Jacobi fields $J_{\left(s,\theta_0\right)}=\frac{\partial \mu}{\partial \theta} \left(\theta_0, s ,\tau\right)$ satisfy
\[
J_{\left(s,\theta_0\right)}\left(\tau\right)=(0,-\tau\sin\theta_0,\tau \cos\theta_0)  \quad \text{ and } \quad J'_{\left(s,\theta_0\right)}\left(\tau\right)=(0,-\sin\theta_0, \cos\theta_0)   .
\]
Observe that for all $\theta\in\left[0,2\pi\right)$ we get $\mu\left(\theta, s ,-s\right)\in C\equiv\{t=0\}$, then its initial vectors at $C$ are
\begin{equation}\label{eq-Jacobi-3M-initial}
J_{\left(s,\theta\right)}\left(0\right)=(0,s\sin\theta,-s \cos\theta_0)   \quad \text{ and } \quad J'_{\left(s,\theta\right)}\left(0\right)=(0,-\sin\theta, \cos\theta)   
\end{equation}
for all $\theta\in\left[0,2\pi\right)$.

In an analogous way as in example \ref{ex-contact-4M}, we obtain that
\begin{equation}\label{eq-contact-3-M}
\mathcal{H}  = \textstyle{ \mathrm{span}\left\{  \sin\theta  \frac{\partial}{\partial x} - \cos\theta \frac{\partial}{\partial y} , \frac{\partial}{\partial \theta}  \right\}  }
\end{equation}
and a contact form can be written by
\[
\alpha = \cos \theta \cdot dx + \sin \theta \cdot dy   .
\] 

Since $\mathcal{N}_{\mathbb{M}^3}$ is $3$--dimensional, we can graphically represent objects from this space of light rays. For any $\gamma\simeq(x,y, \theta)$, the basis of the contact hyperplane in equation (\ref{eq-contact-3-M}) does not depend on $(x,y)$, so for any fixed $(x_0,y_0)$ the contact hyperplanes rotate as the $\theta$--coordinate varies as shown in figure \ref{figura-contact-struct}. The first vector in the basis of (\ref{eq-contact-3-M})(in red) takes a complete turn when $\theta$ moves on intervals of length $2\pi$, but observe that expression of the hyperplane repeats when $\theta$ varies in periods of $\pi$. 
 
\begin{figure}[h]
\centering
\begin{tikzpicture}[scale=1.5]
\draw[->] (0.6,-0.3) -- (0.6-0.9,-0.3-0.9) node[anchor=north west] {$x$};
\draw[->] (0.6,-0.3) -- (0.6,0.8) node[anchor=south west] {$y$};
\draw[->] (0.6,-0.3) node[anchor=north west] {$0$} -- (6.883152+0.6,-0.3) node[anchor=north] {$\theta$};
\draw[densely dotted] (0,0) -- (360*0.0174532,0);
\draw[densely dotted] (0.6,0.5485) -- (0,0) -- (0,-0.9) ;
\draw (0.65,0.5485) -- (0.55,0.5485) node[anchor=east] {$y_0$};
\draw (0.05,-0.9) -- (-0.05,-0.9) node[anchor=east] {$x_0$};
\draw (3.701576,-0.33) node[anchor=north west] {$\pi$} -- (3.781576,-0.24);
\draw (6.843152,-0.33) node[anchor=north west] {$2\pi$} -- (6.923152,-0.24);

\draw[densely dotted] (0.6,0.5485)-- (6.883152,0.5485);
\draw[densely dotted] (0,-0.9) -- (6.283152,-0.9);
\draw[densely dotted] (3.141576,0) -- (3.141576,-0.9) -- (3.741576,-0.3) -- (3.741576,0.5485)-- cycle ;
\draw[densely dotted] (6.283152,0) -- (6.283152,-0.9) -- (6.883152,-0.3) -- (6.883152,0.5485)-- cycle ;
 
\foreach \x in {0,30,...,360}
{
 \fill[color=gray!20] ({(\x*0.0174532)-(0.2) +0.2*0.7071*sin(\x)},{0.2*(cos(\x))}) -- 
          ({(\x*0.0174532)+(0.2)+0.2*0.7071*sin(\x)},{0.2*(cos(\x))}) --
          ({(\x*0.0174532)+(0.2)-0.2*0.7071*sin(\x)},{-0.2*(cos(\x))}) -- 
          ({(\x*0.0174532)-(0.2)-0.2*0.7071*sin(\x)},{-0.2*(cos(\x))}) -- cycle ;
\draw[->,thick,color=blue] ({(\x*0.0174532)},0) -- ({(\x*0.0174532)+0.2},0);
\draw[->,thick,color=red] ({(\x*0.0174532)},0) -- ({(\x*0.0174532)-0.2*0.7071*sin(\x)},{-0.2*(cos(\x))});
}
\end{tikzpicture}
 \caption{The contact structure for any fixed coordinates $(x_0,y_0)$ of $\mathcal{N}_{\mathbb{M}^3}$}
  \label{figura-contact-struct}
\end{figure}
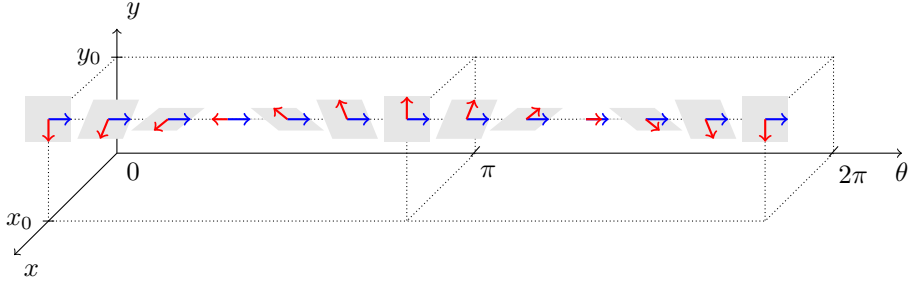

\end{example}

\section{The space of skies \texorpdfstring{$\Sigma$}{S} and the Reconstruction theorem}\label{sec:lightrays-skies}

If we want to replace the conformal manifold $M$ by the space of light rays $\mathcal{N}$ as alternative framework in the study of conformal properties, we will need to identify single points (events) in $M$ as objects in $\mathcal{N}$. 
Notice that it is reasonable to assume that the skies watched by two observers located in different points of $M$ are not exactly the same because the light rays arriving at their eyes differ. 
So, the congruence of light rays $S\left(p\right)\subset \mathcal{N}$ passing by a point $p\in M$ characterizes the point $p$. 
We will call the \emph{sky of $p$} to the set 
\[
S\left(p\right)=\left\{ \gamma\in\mathcal{N}: p\in\gamma\subset M \right\}  .
\]
We will also use capital letters to denote skies of points denoted by lower case, i.e. $S\left(x\right)=X$.

For a given $p\in M$, using the diagram (\ref{diagram-charts}), we can write $S\left(p\right)=\boldsymbol{\upgamma}\left(\mathbb{PN}_p\right)$ and since the fibre $\mathbb{PN}_p$ is diffeomorphic to the sphere $\mathbb{S}^{m-2}$ then we have that every sky $S\left(p\right)$ is a smooth submanifold of $\mathcal{N}$ diffeomorphic to $\mathbb{S}^{m-2}$.

The set of all skies is
\[
\Sigma= \left\{ S\left(x\right): x\in M  \right\}
\]
and it is called \emph{the space of skies} and it permits to define the \emph{sky map} by
\[
\begin{tabular}{rccl}
$S:$ & $M$ & $\rightarrow$ & $\Sigma$ \\
     & $x$ & $\mapsto$ & $X=S\left(x\right)$
\end{tabular}
\]
which is surjective by definition.
Since we want to identify points in $M$ with skies in $\mathcal{N}$ we require $S$ to be injective, saying that $M$ is \emph{sky--separating}.

Consider $x\in M$ and $\gamma \in X=S\left(x\right)$  such that  $\gamma \left(s_{0}\right) =x$.  By the description of $T_{\gamma}\mathcal{N}$ in section \ref{sec:lightrays-tangent}, any $\langle J \rangle\in T_{\gamma}X\subset T_{\gamma}\mathcal{N}$ can be represented by a Jacobi field $J$ defined by a null geodesic variation with $x$ as a fixed point. So, we have $J\left( s_{0}\right) =0\left( \mathrm{{mod} \gamma ^{\prime }}\right) $ and then 
\begin{equation}\label{eq-sky-tangent}
T_{\gamma }X=\{\langle J \rangle\in T_{\gamma }\mathcal{N}:J\left( s_{0}\right) =0\left(\mathrm{{mod}\gamma ^{\prime }}\right) \} .
\end{equation}
Moreover, by eq. (\ref{eq-J-gamma-constant}), we have that $\mathbf{g}\left( J,\gamma ^{\prime}\right)=0$ and therefore $T_{\gamma }X\subset \mathcal{H}_{\gamma }$.

\begin{remark}\label{remark-light-non-conjugate}
It is clear \cite[Prop. 10.10]{On83} that in a normal neighbourhood $U\subset M$ there are no conjugate points along any geodesic, then there are no $t_1,t_2\in \mathbb{R}$ and $\langle J \rangle\in T_{\gamma}\mathcal{N}$ where $\gamma\left(t_1\right), \gamma\left(t_2\right)\subset U \subset M$ with $\gamma\in \mathcal{N}$ such that $J\left(t_1\right)=0\left(\mathrm{{mod}\gamma ^{\prime }}\right)$ and $J\left(t_2\right)=0\left(\mathrm{{mod}\gamma ^{\prime }}\right)$. This implies that $T_{\gamma}S\left(\gamma\left(t_1\right)\right)\cap T_{\gamma}S\left(\gamma\left(t_2\right)\right)= \{0\}$.
\end{remark}

\begin{definition}
We will say that $V\subset M$ is \emph{light non--conjugate} if
\[
T_{\gamma}X \cap T_{\gamma}Y \neq \left\{0_{\gamma}\right\} \Longrightarrow X=Y\in \Sigma
\]
for all $x,y\in V$ and $X=S\left(x\right),Y=S\left(y\right)$.
When this property, natural for normal neighbourhoods, is extended to the entire manifold $M$, then we will say that $M$ is light non--conjugate.
\end{definition}

A consequence of the existence of non--conjugate points along a null geodesic is that the contact hyperplane $\mathcal{H}_{\gamma}$ can be constructed as a direct sum of tangent spaces of skies. Indeed, if $x,y\in \gamma$ are not conjugated points along $\gamma\in \mathcal{N}$ then we have 
\begin{equation}\label{eq-H-TX-TY}
\mathcal{H}_{\gamma}=  T_{\gamma}X \oplus T_{\gamma}Y  .
\end{equation}
In this relation, it is easy to observe that the contact structure $\mathcal{H}$ does not depend on the representative metric in the conformal structure, that is, the contact hyperplanes are given by conformal objects.

\begin{remark}\label{remark-contact-skies}
In example \ref{ex-contact-4M}, we have found the contact structure $\mathcal{H}$ of the $4$--dimensional Minkowski spacetime $\mathbb{M}^4$ computing bases of the tangent spaces of two skies as equation (\ref{eq-H-TX-TY}) suggests, because we have that 
\[
T_{\gamma}S\left(\gamma(s)\right)=\mathrm{span}\{ \langle J_{\left(s,\theta_0,\phi_0\right)}^{\theta} \rangle ,\langle J_{\left(s,\theta_0,\phi_0\right)}^{\phi} \rangle \}
\]
for all $s\in\mathbb{R}$.
\end{remark}

\subsection{Topology and differentiable structure in \texorpdfstring{$\Sigma$}{S}}\label{sec:Topology}

If we want to identify points with skies in a suitable way, we will need to give a topology and a differentiable structure to $\Sigma$ to make of the sky map $S$ a diffeomorphism.

How can we define the required topology and differentiable structure of $\Sigma$ intrinsically, that is, in terms of the geometry of $\mathcal{N}$?

The chosen topology in $\Sigma $ will be induced by the topology of $\mathcal{N}$. 
If $\mathcal{U}\subset \mathcal{N}$ is an open set, then we denote the set of all skies $X\in \Sigma$ such that $X \subset \mathcal{U}$ by 
\[
\Sigma\left(\mathcal{U}\right)= \left\{ X\in \Sigma: X\subset \mathcal{U} \right\} .
\]

\begin{definition}
The \emph{Reconstructive} or \emph{Low's topology} $\mathfrak{T}_L$ on $\Sigma$ is the topology generated by the basis $\{ \Sigma (\mathcal{U}) \mid \mathcal{U} \subset \mathcal{N} \, \, \mathrm{open}\,   \}$.
\end{definition}

As a first step, we will state the topological equivalence.

\begin{proposition}
\label{prop-sky-homeo} If $\Sigma $ is equipped with the Reconstructive topology, then the
sky map $S:M\rightarrow \Sigma $ is a homeomorphism.
\end{proposition}

The first proof of this proposition was given by Kinlaw in \cite[Prop. 4.3]{Ki11} under the hypotheses of non-refocusing in $M$, that is, we say that $M$ is \emph{refocusing} at $x\in M$ if there exists an open neighbourhood $V$ of $x$ such that for all open $U$ with $x\in U\subset V$, there exists $y\notin V$ such that all light rays through $y$ enter $U$, see figure \ref{figure-refocussing}.
When this property is not satisfied at any $x\in M$ we shall say that $M$ is \emph{non--refocusing}.
Another and simpler proof can be found in \cite[Prop. 3]{Ba14}.

In a later paper, it is shown \cite[Cor. 20]{Ba15} that the non-refocusing hypothesis is unnecessary under the hypotheses of strong causality, null pseudo-convexity and sky--separation. So, we will not use this concept in the present review, but it must be considered when weaker hypotheses are assumed.

\begin{figure}[h]
\centering
\begin{tikzpicture}[scale=1]
\shade[left color = gray,right color = gray!80,middle color = gray!30,shading angle=75,opacity=0.6] (-0.3,0) to[out=110,in=210] (1,3) to[out=-120,in=70] (0.4,0) arc (0:-180:0.35cm and 0.15cm);
\draw[densely dashed,color=gray!80] (0.4,0) arc (0:360:0.35cm and 0.15cm);
\draw[color=gray!40] (-0.3,0) to[out=110,in=210] (1,3) to[out=-120,in=70] (0.4,0);
\draw[color=gray!40] (-0.1,-0.12) to[out=100,in=220] (1,3) to[out=-120,in=90] (0.2,-0.12);
\draw[color=gray!40] (0,-0.15) to[out=90,in=230] (1,3) ;
\shade[shading=ball, ball color=red!40,opacity=0.3] (-1.5,0) to[out=87,in=182] (-0.3,1.1) to[out=2,in=92] (1.6,0.2) to[out=-88,in=4](-0.1,-1) to[out=184,in=273](-1.5,0);
\draw[color=red!50] (-1.5,0) to[out=87,in=182] (-0.3,1.1) to[out=2,in=92] (1.6,0.2) to[out=-88,in=4](-0.1,-1) to[out=184,in=273](-1.5,0);
\draw[densely dashed,color=red!50] (-1.5,0) to[out=90,in=180] (0.1,0.9) to[out=0,in=90] (1.6,0.2) to[out=-90,in=0](0,-0.8) to[out=180,in=270](-1.5,0);
\shade[shading=ball, ball color=blue!40,opacity=0.3] (0,0) circle (0.7cm); 
\draw[color=blue!60] (0,0) circle (0.7cm);
\draw[densely dashed,color=blue!60] (0.7,0) arc (0:360:0.7cm and 0.3cm);
\fill (1,3) node[anchor=north west] {$y$} circle (1pt);
\fill (0,0) node[anchor=east] {$x$} circle (1pt);
\draw (1.1,1.8) node {$S(y)$};
\draw[color=blue] (0.8,0.4) node {$U$};
\draw[color=red] (-1.3,1) node {$V$};
\end{tikzpicture}
 \caption{Refocusing at $x\in M$. All light rays in the sky of some point $y\in M$ enter $U\subset V \subset M$.}
  \label{figure-refocussing}
\end{figure}
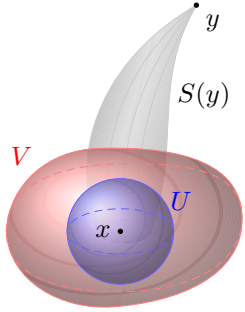

Now, we want to get the differentiable structure of $\Sigma$ compatible with the Reconstructive topology. We will describe $\Sigma$, locally, as a quotient manifold.

By proposition \ref{prop-Mingu-Sanch}, we can consider a globally hyperbolic, normal and causally convex open set $V\subset M$ (as in diagram (\ref{diagram-charts})). Notice that, by remark \ref{remark-light-non-conjugate}, $V$ is light non-conjugate, and by Proposition \ref{prop-sky-homeo},  $S\left( V\right) = \{ S(x) \mid x\in V \}$ is open in $\Sigma$. Then we can translate the property of light non-conjugation to open sets in $\Sigma$ by the following property. 

\begin{definition}\label{light-non-conjugate-sigma}
An open set $U\subset \Sigma$ in the Reconstructive topology is called \emph{light non-conjugate} if for
every $X,Y\in U $ and every $\gamma \in X\cap Y$ such that $T_{\gamma }X\cap
T_{\gamma }Y\neq \left\{ 0\right\} $ implies that $X=Y$. 
\end{definition}

All the convex normal neighbourhoods at $x\in M$ set up a basis for the
topology of $M$ at $x$  \cite[Prop. 5.7]{On83} then, by proposition \ref{prop-sky-homeo}, all light non-conjugate neighbourhoods in $\Sigma$ also constitute a basis for the Reconstructive topology. But this property is not good enough to construct the differentiable structure on $\Sigma$.

In what follows we will denote by 
\begin{equation}\label{eq-widehat-T}
\widehat{T}W = TW-\{0\}
\end{equation}
the non--zero tangent vectors at the manifold $W$.

\begin{definition}\label{regular_neigh}
A light non-conjugate open set $U\subset \Sigma $ is said to be a \emph{regular} open set if $U$ verifies that $\widehat{U}=\bigcup\limits_{X\in U}\widehat{T}X \subset T\mathcal{N}$ is a regular submanifold of $\widehat{T}\mathcal{U}$, $\mathcal{U=}\bigcup\limits_{X\in U}X$.
\end{definition}

Next theorem \ref{theorem1}, is a keystone in order to show that $S$ is a diffeomorphism. In the version of the statement below, we offer additional results contained in the proof of \cite[Thm. 1]{Ba14}.
\begin{theorem}
\label{theorem1} Let $X\in \Sigma$ be any sky with $x\in M$ such that $X=S(x)$. If $V\subset M$ is a relatively compact, globally hyperbolic, causally convex and normal open neighbourhood of $x$ then $U=S(V)$ is a regular open neighbourhood $U\subset \Sigma $ of $X$. Moreover, there exists a coordinate system $\overline{\varphi}$ in $\widehat{U}$ adapted to the leaves $\widehat{T}X \subset \widehat{U}$.
\end{theorem}

\begin{remark}\label{remark-carta-U-hat}
As shown in \cite{Ba14}, the adapted coordinate system $\overline{\varphi}$ of theorem \ref{theorem1} can be obtained by fixing a coordinate chart $\varphi=\left( x^{1},\ldots ,x^{m}\right)$ in $V$ and a orthonormal frame $\lbrace \mathbf{E}_{1},\ldots ,\mathbf{E}_{m} \rbrace$ in $V$, related to some auxiliary metric $\mathbf{g}\in\mathcal{C}$, such that $\mathbf{E}_{1}$ is timelike. 
Any null direction $[w]\in \mathbb{PN}(V)$ can be defined by a vector
\[
w=\mathbf{E}_1(p) + u^2 \mathbf{E}_2(p) + \cdots + u^m \mathbf{E}_m(p)
\]
where since $(u^2)^2+ \cdots +(u^m)^2 =1$ then, with no lack of generality, $(u^3,\ldots , u^m)$ define the components of $w$. 
On the other hand, if $\langle J \rangle \in\widehat{T}_{\gamma}S(p)$ then $J(s)=0 \left(\mathrm{mod}~\gamma'(s)\right)$ for $\gamma(s)=p$, then the transversal component to $\gamma'$ of the initial vector $J'(s)=\sum\limits_{j=1}^{m}v^{j}\mathbf{E}_{j}\left( p\right) $ defines $\langle J \rangle$. So, the coordinate chart $\overline{\varphi}:\widehat{U} \rightarrow \mathbb{R}^{3m-4}$ is given by:
\begin{equation}\label{overlinevarphi}
\overline{\varphi}\left(\langle J \rangle\right) = \left( \mathbf{x}, \mathbf{u}, \mathbf{v} \right)=\left( x^1, x^{2},\ldots ,x^{m},
u^{3},\dots ,u^{m} ,\left\langle v^{1}\dots ,v^{m}\right\rangle
\right) \in \mathbb{R}^{3m-4}
\end{equation}
where $\mathbf{v}=\left\langle v^{1},\ldots ,v^{m}\right\rangle=\left(v^1,\ldots,v^m\right)\left(\mathrm{mod}~\gamma^{\prime }\right)$. 
Observe that $\mathbf{x}$ defines the point $p\in V$ to whose sky $\langle J \rangle$ is tangent, $\mathbf{u}$ defines the null direction $[w]\in \mathbb{PN}(V)$ such that $\gamma=\gamma_{[w]}\in\mathcal{N}$ and, finally, $\mathbf{v}$ defines the value of $\langle J \rangle$ in the fibre $T_{\gamma}S(p)$.
\end{remark}

An immediate consequence \cite[Cor. 1]{Ba14} is that regular open sets constitute a basis for the topology of $\Sigma$. A refined notion of regular sets can be found in \cite[Sec. 4.1]{Ba15}.

Now, we can state the following theorem.

\begin{theorem}
\label{teo-sky-difeo} 
The sky map $S:M\rightarrow \Sigma $ is a diffeomorphism when $\Sigma $ is equipped with the smooth structure defined by regular sets.
\end{theorem}

\begin{proof}
Consider $x\in M$ and $V\subset M$ a relatively compact, globally hyperbolic, causally convex and normal open neighbourhood of $x$, then $U=S(V)\subset \Sigma$ is open and light non-conjugate. By theorem \ref{theorem1}, $\widehat{U}=\bigcup\limits_{X\in U}\widehat{T}X$ is a regular submanifold of $\widehat{T}\mathcal{N}$ foliated by leaves $\widehat{T}X$ with $X\in U$. Hence, if we call $\widetilde{U}=\{\widehat{T}X:X\in U\}$, then the map $\widetilde{S}:V\rightarrow \widetilde{U}$ given by $\widetilde{S}\left( x\right) =\widehat{T}X$ is a diffeomorphism. Moreover, since $U$ is light non--conjugate, there is a bijection $U\rightarrow\widetilde{U}$ defined by $X \mapsto \widehat{T}X$ and therefore $U$ inherits from $\widetilde{U}$ a differentiable structure such that the sky map $S:V\rightarrow U$ is a diffeomorphism since $\widetilde{S}$ is also. The global bijectiveness of the sky map $S:M\rightarrow \Sigma$ gives us that $S$ is a diffeomorphism.
\qed
\end{proof}

The compatibility between the Reconstructive topology and the differentiable structure of $\Sigma$ is explicitly stated in \cite[Cor. 17]{Ba15}.

\begin{proposition}\label{differ}
The family of regular sets $\left\{ U \mid U \subset \Sigma \text{ regular}\right\}$ is a basis for the reconstructive topology of $\Sigma$. 
Moreover, the differentiable structure in $\Sigma$ inherited from $\widetilde{U}\simeq U$ is the only one making of $S:M\rightarrow \Sigma$ a diffeomorphism.
\end{proposition}

It is important to notice that the differentiable structure given in $U$ is inherited from $\widetilde{U}\simeq \widehat{U}/\mathcal{D}$ with data coming from $T\mathcal{N}$, so it is canonically obtained from elements of $\mathcal{N}$ and it is not necessary to define it by the geometry of $M$.

\subsection{The Reconstruction theorem}\label{sec:Reconstruction}

This section is addressed to discuss the conditions under which a conformal manifold can be reconstructed from its space of light rays.    

Notice that ``isomorphic'' data in the geometry of $\mathcal{N}$ must provide the same reconstruction of the conformal manifold $M$.  This observation leads to the Reconstruction theorem \cite[Lem. 2 \& Thm. 3]{Ba14}, which is a Malament-Hawking--like theorem (see \cite[Sect. 4.3.4]{Mi19}). Observe that there is no explicit mention to causal conditions of the conformal manifold in the statement of the following version of the Reconstruction theorem. The causal structure is implicit in the geometry of the spaces of light rays and skies as we will see in section \ref{sec:Causality}.

\begin{theorem}\textbf{(Reconstruction theorem)}\label{teo-reconstruction}
Let $\left( M,\mathcal{C}\right) $, $\left( \overline{M},\overline{\mathcal{C}}\right) $ be two strongly causal, null pseudo-convex and sky--separating conformal manifolds and $\left(\mathcal{N},\Sigma \right) $, $\left( \overline{\mathcal{N}},\overline{\Sigma }\right) $ their corresponding pairs of spaces of light rays and skies.
Let $\phi :\mathcal{N}\rightarrow \overline{\mathcal{N}}$ be a diffeomorphism such that $\phi \left( X\right) \in \overline{\Sigma }$ for all $X\in \Sigma$.
Then the map 
\[
\varphi =\overline{S}^{-1}\circ \Phi \circ S:M\rightarrow \overline{M}
\]
is a conformal diffeomorphism onto its image, where $\overline{S}:\overline{M} \rightarrow \overline{\Sigma } $ is the sky map of $\overline{M}$ and $\Phi:\Sigma \rightarrow \overline{\Sigma}$ is the map induced by $\phi$ defined by $\Phi\left(X\right)=\phi\left(X\right)$.
\end{theorem}

\begin{proof}
First, we will show that $\Phi$ is a diffeomorphism onto its range. 

Trivially, $\Phi$ is well defined and injective. 
Consider any open set $\overline{U} \subset \overline{\Sigma}$ and denote $U=\Phi^{-1}(\overline{U})$. By definition of Reconstructive topology, there exists $\overline{\mathcal{W}}\subset \overline{\mathcal{N}}$ open such that $\overline{U} = \Sigma (\overline{\mathcal{W}})$. Since $\phi$ is a diffeomorphism, then $\mathcal{W} = \phi^{-1}(\overline{\mathcal{W}})$ is open in $\mathcal{N}$ and moreover, for each sky $X\subset \mathcal{W}$ we have that $\phi(X) \subset \overline{\mathcal{W}}$ is a sky in $\overline{U}$, hence $\Phi(X) \in \overline{U}$.  So, $U = \Sigma (\mathcal{W})$ and $U\subset\Sigma$ is open. Therefore $\Phi$ is continuous.

Now, consider $X\in \Sigma $ and $\overline{X}=\phi \left( X \right) \in \overline{\Sigma }$.  By proposition \ref{differ} and the continuity of $\Phi$ there exist regular neighbourhoods $U=\Sigma(\mathcal{U})\subset \Sigma $ of $X$ and $\overline{U}=\Sigma(\overline{\mathcal{U}})\subset \overline{\Sigma }$ of $\overline{X}$ such that $\Phi \left( U\right)  \subset \overline{U}$, then we can assume that $\phi \left( \mathcal{U}\right) \subset \overline{\mathcal{U}}$ 
\footnote{
This is not an automatic property for any $\mathcal{U}\in\mathcal{N}$ and $\overline{\mathcal{U}}\in\overline{\mathcal{N}}$. If $U=\Sigma(\mathcal{V}) $ and $\overline{U}=\Sigma(\overline{\mathcal{U}})$ such that $\Phi\left(U\right)\subset \overline{U}$, then we can choose $\mathcal{U}=\phi^{-1}\left(\phi\left(\mathcal{V}\right)\cap \overline{\mathcal{U}}\right)\subset\mathcal{V}$. 
Now, if $X\subset \mathcal{V}$, since $\Phi\left(U\right)\subset \overline{U}$ then $\phi\left(X\right)\subset \phi\left(\mathcal{V}\right)\cap\overline{\mathcal{U}}$ and since $\phi$ is a diffeomorphism, then $X\in \mathcal{U}$. So we have $U=\Sigma\left(\mathcal{U}\right)= \Sigma\left(\mathcal{V}\right)$ and $\phi\left(\mathcal{U}\right)\subset \overline{\mathcal{U}}$.
}.

Since $\phi\colon \mathcal{N}\to \mathcal{\overline{N}}$ is a diffeomorphism, then $\phi_*\colon T\mathcal{N}\to T\mathcal{\overline{N}}$ is also a diffeomorphism and its restriction $\phi _{\ast }\colon \widehat{T}\mathcal{U}\rightarrow\widehat{T}\overline{\mathcal{U}}$ is a diffeomorphism onto its image. Since $U$ and $\overline{U}$ are regular, then $\widehat{U}$ and $\widehat{\overline{U}}$ are regular submanifolds, and 
\[
\phi _{\ast }\left( \widehat{U}\right) =\phi _{\ast }\left( \bigcup\limits_{%
\overline{X}\in U}\widehat{T}X\right) =\bigcup\limits_{X\in U}\phi _{\ast
}\left( \widehat{T}X\right) =\bigcup\limits_{X\in U}\widehat{T}\phi \left(
X\right) \subset \widehat{\overline{U}} 
\]
so, the restriction $\phi _{\ast }:\widehat{U}\rightarrow \widehat{\overline{U}}$ is another diffeomorphism onto its range.

If we denote by $\widehat{\mathcal{D}}$ and $\widehat{\overline{\mathcal{D}}}$ the distributions in $\widehat{U}$ and $\widehat{\overline{U}}$ such that their spaces of leaves are $\{\widehat{T}X:X\in U\}$ and $\{\widehat{T}\overline{X}:\overline{X}\in \overline{U}\}$ respectively, since $\phi$ maps skies to skies and $\phi _{\ast }\left(\widehat{T}X\right)=\widehat{T}\overline{X}$, then $\phi_{\ast }:\widehat{U}\rightarrow \widehat{\overline{U}}$ induces the canonical quotient map
$$\overline{\phi _{\ast}}  \colon \widehat{U}/\widehat{\mathcal{D}} \rightarrow \widehat{\overline{U}}/\widehat{\overline{\mathcal{D}}} ,$$ 
which is smooth. Then we obtain the following commutative diagram:
\begin{equation*}
\begin{tikzpicture}[every node/.style={midway}]
\matrix[column sep={6em,between origins},
        row sep={1em}] at (0,0)
{ \node(Ug1)   {$\widehat{U}$}  ; & \node(Ug2) {$\widehat{\overline{U}} $}; \\
  \node(UD1) {$\widehat{U}/\widehat{\mathcal{D}}$}; & \node(UD2) {$\widehat{\overline{U}}/\widehat{\overline{\mathcal{D}}}$};                      \\
  \node(U1) {$U$}; & \node(U2) {$\overline{U}$};  \\};
\draw[->] (Ug1) -- (Ug2) node[anchor=south]  {$\phi _{\ast }$};
\draw[->] (UD1) -- (UD2) node[anchor=south]  {$\overline{\phi _{\ast }}$};
\draw[->] (U1)   -- (U2) node[anchor=south] {$\Phi$};
\draw[->] (Ug1) -- (UD1) node[anchor=south]  {};
\draw[->] (UD1) -- (U1) node[anchor=north]  {};
\draw[->] (Ug2) -- (UD2) node[anchor=south]  {};
\draw[->] (UD2) -- (U2) node[anchor=north]  {};
\end{tikzpicture}
\end{equation*}
where the maps $\widehat{U}/\widehat{\mathcal{D}}\rightarrow U$ and $\widehat{\overline{U}}/\widehat{\overline{\mathcal{D}}}\rightarrow \overline{U}$ are diffeomorphisms because of proposition \ref{differ}. Then $\Phi :U\rightarrow\overline{U}$ a diffeomorphism onto its image, by injectiveness, $\Phi :\Sigma\rightarrow\overline{\Sigma}$  is a diffeomorphism onto its image and, by theorem \ref{teo-sky-difeo}, therefore $\varphi =\overline{S}^{-1}\circ \phi \circ S:M\rightarrow \overline{M}$ too.

Finally, we will show that $\varphi$ maps light rays into light rays.  The image of a light ray $\gamma$ under the sky map $S$ can be written by 
\[
S\left( \gamma \right) =\{\beta \in \mathcal{N}:\exists  \, X\in \Sigma 
\,\,\, \mathrm{ such ~ that} \,\, \gamma ,\beta \in X\}   .
\]
Then 
\[
\Phi \left( S\left( \gamma \right) \right) =\phi \left( S\left( \gamma
\right) \right) =\{\phi \left( \beta \right) \in \overline{\mathcal{N}}%
:\exists \, X\in \Sigma \,\,\, \mathrm{ such ~ that} \,\, \gamma ,\beta \in X\}
\]
and since $\phi $ is a diffeomorphism preserving skies
\[
\Phi \left( S\left( \gamma \right) \right) =\{\phi \left( \beta \right) \in
\overline{\mathcal{N}}:\exists \, \Phi \left( X\right) \in \overline{%
\Sigma }\,\,\, \mathrm{ such ~ that} \,\, \phi \left( \gamma \right) ,\phi \left( \beta
\right) \in \Phi \left( X\right) \} = \overline{S}\left(\phi\left(\gamma\right)\right)
\]
So, we have $\varphi \left( \gamma \right) =\overline{S}^{-1}\circ \Phi \circ S\left( \gamma \right) =\overline{S}^{-1}\circ \overline{S}\circ \phi \left(
\gamma \right) =\phi \left( \gamma \right) \in \overline{\mathcal{N}}$ is a light ray. Therefore, by \cite[Sec. 3.2]{HE}, $\varphi $ is a conformal diffeomorphism onto its image.
\qed
\end{proof}

In virtue of Hawking and Malament's theorems, weakening the hypotheses of the Reconstruction theorem is a plausible idea, but it is not evident that it can be done in a simple way. For example, under the hypothesis that $\phi$ is a sky--preserving homeomorphism, determining whether $\Phi$ is an open map is, so far, an open problem that we establish as a conjecture. Anyway, we can state the following proposition.

\begin{proposition}\label{prop-reconstruction}
Let $\left( M,\mathcal{C}\right) $, $\left( \overline{M},\overline{\mathcal{C}}\right) $ be two strongly causal, null pseudo-convex and sky--separating conformal manifolds and $\left(\mathcal{N},\Sigma \right) $, $\left( \overline{\mathcal{N}},\overline{\Sigma }\right) $ their corresponding pairs of spaces of light rays and skies.
Let $\phi :\mathcal{N}\rightarrow \overline{\mathcal{N}}$ be a homeomorphism such that $\phi \left( X\right) \in \overline{\Sigma }$ for all $X\in \Sigma$, the map $\Phi:\Sigma \rightarrow \overline{\Sigma}$ is defined by $\Phi\left(X\right)=\phi\left(X\right)$ and $\varphi=\overline{S}^{-1}\circ \Phi \circ S$ where $S:M\rightarrow \Sigma$ and $\overline{S}:\overline{M}\rightarrow \overline{\Sigma}$ are the sky maps.
Then the following conditions are equivalent:
\begin{enumerate}
\item \label{itm:prop-cond-1} $\Phi$ is an open map.
\item \label{itm:prop-cond-2} $\varphi$ is a conformal diffeomorphism onto its image.
\item \label{itm:prop-cond-3} $\phi$ is a diffeomorphism.
\end{enumerate}
\end{proposition}

\begin{proof}
\ref{itm:prop-cond-1}) $\Rightarrow$ \ref{itm:prop-cond-2})

When $\phi$ is a homeomorphism, the same arguments used in the proof of theorem \ref{teo-reconstruction} can be used to show that $\Phi$ is continuous (as well as well--defined and injective) and $\varphi$ maps light rays of $\mathcal{N}$ into light rays of $\overline{\mathcal{N}}$. 
If, moreover, $\Phi$ is open, then it is a homeomorphism onto its image and therefore, by composition, $\varphi$ is also a homeomorphism onto its image which maps light rays into light rays. By the Hawking's theorem \cite[Thm. 4.61]{Mi19}, $\varphi$ is a conformal diffeomorphism.

Trivially we have \ref{itm:prop-cond-2}) $\Rightarrow$ \ref{itm:prop-cond-3}) and, as seen in the proof of theorem \ref{teo-reconstruction}, we get \ref{itm:prop-cond-3}) $\Rightarrow$ \ref{itm:prop-cond-1}). 
\qed
\end{proof}

The contact structure $\mathcal{H}$ is not sufficient to recover the conformal manifold $M$. The space of skies is also needed as done in theorem \ref{teo-reconstruction}. The following example (see \cite[Ex. 2]{Ba14} for details) shows that there can be a diffeomorphism preserving the contact structure between the spaces of light rays of non-equivalent conformal manifolds.

\begin{example}
Let $\mathbb{M}^3$ be the 3--dimensional Minkowski space--time with standard coordinates given by $\left(t,x,y\right)\in \mathbb{R}^3$ and let $\mathcal{N}_{\mathbb{M}^3}$ be its space of  light rays. 

For $\epsilon \geq 0$, we consider $M_{\epsilon}=\left\{\left(t,x,y\right)\in \mathbb{R}^3 : t<\epsilon \right\}$ equipped with the metric
\[
\mathbf{g}_{\epsilon} = -\left(1+f\left(t\right)\right)dt\otimes dt + 2f\left(t\right)dt\otimes dx + \left(1-f\left(t\right)\right)dx\otimes dx+dy\otimes dy
\]
where $f$ is a smooth function verifying $f\left(t\right)=0$ for every $t\leq 0$ which gives us a small perturbation $\mathbf{g}_{\epsilon}$ of the metric $\mathbf{g}$ of $\mathbb{M}^3$ for $0<t<\epsilon$. It is possible to choose $f$ small enough to keep $M_{\epsilon}$ globally hyperbolic. 

Trivially, $\mathbb{M}^3$ and $M_{\epsilon}$ are two space--times extending $M_0$ the corresponding spaces of light rays $\mathcal{N}_{\mathbb{M}^3}$, $\mathcal{N}_{\epsilon}$ and $\mathcal{N}_0$ can be fully determined by a common Cauchy surface, for example $C \equiv \left\{t=-1\right\}$.
Since $M_0$ (equipped with its metric) is contained in $\mathbb{M}^3$ and $M_{\epsilon}$, then every light ray $\gamma_0$ in $M_0$ passing through $C$, defines light rays $\gamma\subset \mathbb{M}^3$ and $\gamma_{\epsilon}\subset M_{\epsilon}$. 
Then we can state a diffeomorphism $\mathcal{N}_{\mathbb{M}^3}\rightarrow\mathcal{N}_{\epsilon}$ such that the contact structures are preserved $\left(\mathcal{H}_{\mathbb{M}^3}\right)_{\gamma}\simeq\left(\mathcal{H}_{\epsilon}\right)_{\gamma_{\epsilon}}$.

It is known that $\mathbb{M}^3$ is flat, but if we compute the Cotton tensor $\mathbf{C}_{\epsilon}$ of $M_{\epsilon}$, defined by $C_{ijk}=\nabla _{k}R_{ij}-\nabla _{j}R_{ik}+\frac{1}{4}\left( \nabla_{j}Rg_{ik}-\nabla _{k}Rg_{ij}\right)$ where $R_{ij}$ and $R$ denote the Ricci curvature and the scalar curvature, we obtain that $\mathbf{C}_{\epsilon}\neq 0$ therefore, by \cite[Thm. 9]{La88}, $\mathbb{M}^3$ and $M_{\epsilon}$ are not conformally equivalent because $M_{\epsilon}$ is not conformally flat.
\end{example}

\section{\texorpdfstring{Causality in $\mathcal{N}$}{Causality in N}}\label{sec:Causality}

In this section we will see how the causal structure of $M$ is encoded in $\mathcal{N}$ by \emph{legendrian isotopies}. 
Any curve $\lambda=\lambda(s)\in M$ is mapped to the curve of skies $S\left(\lambda(s)\right)\in \Sigma$ under the sky map $S$. We can see every $\Lambda_s=S\left(\lambda(s)\right)$ as a smooth submanifold in the space of light rays $\mathcal{N}$, so $S\circ \lambda$ defines a variation of smooth submanifolds in $\mathcal{N}$.

\subsection{Legendrian isotopies}\label{sec:isotopies}

We will introduce some basic and general definition and results. For a more detailed description, see \cite{Ch10b}.

\begin{definition}
Let $\left( \mathcal{M},\mathcal{H}\right) $ be a co-oriented $\left(2n+1\right) $--dimensional contact manifold with contact structure $\mathcal{H}$.

\begin{itemize}
\item A smooth submanifold $\Lambda\subset \mathcal{M}$ is called \emph{Legendrian} if $T_p \Lambda \subset \mathcal{H}_p$ for all $p\in \Lambda$.
\item A differentiable family $\{\Lambda _{s}\}_{s\in \left[0,1\right] }$ of Legendrian submanifolds is called a \emph{Legendrian isotopy}.
\item A \emph{parametrization} of a Legendrian isotopy $\{\Lambda _{s}\}_{s\in \left[0,1\right] }$ is a map $F:\Lambda _{0}\times \left[ 0,1\right] \rightarrow \mathcal{M}$ such that, for all $s \in [0,1]$, we have that $F\left(\Lambda _{0}\times \{s\}\right) =\Lambda _{s}\subset \mathcal{M}$ and $F_s \colon \Lambda_0 \to \Lambda_s$ given by $F_s (\lambda) = F(s,\lambda)$ is a diffeomorphism.
\item Two Legendrian isotopies are \emph{equivalent} if their corresponding parametrizations $F,\widetilde{F}:\Lambda_0 \times \left[0,1\right] \rightarrow \mathcal{M}$ verify $F_s\left(\Lambda_0\right)=\widetilde{F}_s\left(\Lambda_0\right)$ for every $s\in \left[0,1\right]$.
\end{itemize} 
\end{definition}

\begin{definition}\label{def-sign-legendrian}
We will say that a parametrization $F$ of a Legendrian isotopy is \emph{non--negative} (respectively, \emph{non--positive}, \emph{positive}, \emph{negative}) if $\left( F^{\ast }\alpha \right) \left( \frac{\partial}{\partial s}\right) \geq 0$ (respectively $\leq 0$, $>0$, $<0$), where $\alpha$ is a contact 1--form such that $\mathcal{H}=\mathrm{ker}~\alpha$. 
\end{definition}

If the sign of a parametrization of a Legendrian isotopy is defined in the sense of definition \ref{def-sign-legendrian}, then it does not depend on  the parametrization. This result is shown in \cite[Lem. 3]{Ba14} and it allows to define the sign of a Legendrian isotopy in terms of the sign of any parametrization.
 
\begin{lemma}\label{lemma00250} 
Let $F,\widetilde{F}:\Lambda _{0}\times \left[0,1\right]\rightarrow \mathcal{M}$ be two equivalent parametrizations of a Legendrian isotopy $\{\Lambda_{s}\}_{s\in \left[0,1\right]}$. If $F$ is non-negative (respectively non-positive, positive, negative) then so is $\widetilde{F}$.
\end{lemma}

\subsection{Causal curves}\label{sec:Causalcurves}

As it was mentioned in the introduction of section \ref{sec:Causality}, we will study the Legendrian isotopies in the space of null
geodesics $\mathcal{N}$ defined by causal curves.

First notice that the co-orientation can be defined by appointing the sign of $\mathbf{g}\left( J,\gamma ^{\prime
}\right) $ to the sign of $\langle J \rangle  \in T_{\gamma}\mathcal{N}$, where $\langle J \rangle =J \left( \mathrm{mod}~\gamma' \right)$ with $\gamma \in \mathcal{N}$ and $\mathbf{g}\in \mathcal{C}$.

By equations (\ref{eq-contact-struct}) and (\ref{eq-sky-tangent}), any sky $X=S(x)\in \Sigma$ is a Legendrian submanifold and, by diagram (\ref{diagram-charts}), diffeomorphic to $S_{x}=\lbrace \left[ u\right] :u\in \mathbb{N}^{+}_{x} \rbrace = \mathbb{PN}_{x} \simeq \mathbb{S}^{m-2}$. So, given a Legendrian isotopy $\{X_{s}\}_{s\in \left[ 0,1\right] }$ where $X_{s} $ is the sky of $x_{s}\in M$ for $s\in \left[ 0,1\right] $, we can consider a parametrization $F$, given by null directions, as
\[
F:S_{x_0}\times \left[ 0,1\right] \rightarrow \mathcal{N}
\]
since $\boldsymbol{\upgamma}(S_{x_0})=S(x_0)\subset \mathcal{N}$ is the sky of $x_0\in M$ and $\boldsymbol{\upgamma}$ is the submersion of diagram (\ref{diagram-charts}).

\begin{lemma}\label{lemma00300}
Let $\mu:\left[0,1\right]\rightarrow M$ be a curve and $F:S_{\mu(0)} \times \left[0,1\right] \rightarrow \mathcal{N}$ be a Legendrian isotopy such that $F\left(S_{\mu(0)} \times \lbrace s\rbrace \right)= S\left(\mu\left(s\right)\right)\in \Sigma$. Then $\mu$ is differentiable.
\end{lemma}

\begin{proof}
Observe that, since $F$ is a Legendrian isotopy, then $F_{s}:S_{\mu(0)}\rightarrow S\left( \mu \left( s\right) \right) \subset \mathcal{N}$ given by $F_{s}\left( z\right) =F\left(z,s\right) $ is a diffeomorphism for any $s\in \left[ 0,1\right] $. 
Since $F$ and $F_{s}$ are smooth maps then, for $z_{0}\in S_{\mu(0)}$ and $w \in T_{z_{0}}S_{\mu(0)}$, the curve
\[
j\left( s\right) =\left( dF_{s}\right) _{z_{0}}\left( w \right) \in
T_{F\left( z_{0},s\right) }S\left( \mu \left( s\right) \right)
\]
is also differentiable in $\widehat{T}\mathcal{N}$. 
Take $s_{0}\in \left[0,1\right] $ and $U=S\left( V\right) $ be a regular open neighbourhood of $S\left(\mu \left( s_{0}\right)\right)\in \Sigma $.
Then $j\left( s\right) \in \widehat{U}$ for $s$ close to $s_0$ and then the restriction of $j$ to $\widehat{U}$ is differentiable for those $s$ close enough to $s_{0}$. 

If $\left( \widehat{U},\overline{\varphi }=\left( \mathbf{x},\mathbf{u},\mathbf{v}\right) \right) $ is the coordinate system (\ref{overlinevarphi}) of theorem \ref{theorem1} adapted to the leaves $\{\widehat{T}X\}$ with $\left( V,\varphi =\mathbf{x}\right) $ the corresponding coordinate chart in $M$, then $\mu \left( s\right) =\varphi ^{-1}\circ \mathbf{x}\left( j\left( s\right) \right) \in V$, therefore $\mu $ is differentiable.
\qed
\end{proof}

The following result will be used as a technical lemma. The converse statement appears in \cite[Lem. 6]{Ba14} but here we offer a simpler proof.

\begin{lemma}
\label{lemma00350} Let $M$ be a Lorentz manifold and $p\in M$. A non--zero vector $v\in T_p M$ is timelike past--directed (respectively, causal past--directed) if and only if $\mathbf{g}\left(u,v\right)> 0$ for all $u\in \mathbb{N}^{+}_{p}$ (respectively $\mathbf{g}\left(u,v\right)\geq 0$).
\end{lemma}

\begin{proof}
The direct statement is a consequence of \cite[Lem. 5.26]{On83} since the orthogonal complement of $[v]=\mathrm{span}\{v\}$ is a spacelike subspace separating $\mathbb{N}^{+}_{p}$ and $\mathbb{N}^{-}_{p}$.

To show the converse, recall that there exist coordinates $(x^0,\ldots , x^{m-1})$ such that the metric can be written by
\[
\mathbf{g}_p = -dx^0_p \otimes dx^0_p +dx^1_p \otimes dx^1_p + \cdots + dx^{m-1}_p \otimes dx^{m-1}_p
\]
at $p\in M$,  with $\left(\frac{\partial}{\partial x^0}\right)_p$ timelike future--directed. 
We can write $v=v_k \left(\frac{\partial}{\partial x^k}\right)_p$ and $u=\lambda \left(\frac{\partial}{\partial x^0}\right)_p + \lambda\sum_{k=1}^{m-1} Z_k(\overline{\theta}) \left(\frac{\partial}{\partial x^k}\right)_p$ where $\lambda>0$ and $\overline{\theta}=\left(\theta_1 ,\ldots , \theta_{m-2}\right)$ are the corresponding angles in the expressions of the spherical coordinates of the sphere $\mathbb{S}^{m-2}\subset \mathbb{R}^{m-1}$ of radius $r=1$ given by 
\[
\left\{
\begin{tabular}{ll}
$Z_1(\overline{\theta})=\cos\theta_1$ & \\
$Z_k(\overline{\theta})=\prod_{j=1}^{k-1} \sin \theta_j \cos\theta_k$ & , for $k=2,\ldots , m-2$ \\
$Z_{m-1}(\overline{\theta})=\prod_{j=1}^{m-2} \sin \theta_j$ & 
\end{tabular} 
\right.
\]

If $\mathbf{g}\left(u,v\right)>0$, then $-\lambda v_0 + \lambda v_1 Z_1 + \cdots + \lambda v_{m-1} Z_{m-1}>0$ and we have
\begin{equation}\label{eq-causal-past}
 v_0 <  v_1 Z_1 + \cdots + v_{m-1} Z_{m-1} = \overline{v}\cdot Z  \qquad \text{ for all } \overline{\theta} .
\end{equation}

Since $\overline{v}\cdot Z$ can be seen as the standard scalar product in $\mathbb{R}^{m-1}$ of the vectors $\overline{v}=(v_1,\ldots , v_{m-1})$ and $Z=(Z_1, \ldots , Z_{m-1})$, then $\overline{v}\cdot Z = \vert \overline{v} \vert \cdot \vert Z \vert \cdot \cos \alpha$ where $\alpha$ is the angle between $\overline{v}$ and $Z$, then 
\begin{equation}\label{eq-causal-past2}
\underset{\overline{\theta}}{\mathrm{min}}~\overline{v}\cdot Z = -\vert \overline{v} \vert =-\sqrt{v^2_1 + \ldots + v^2_{m-1}} \leq 0
\end{equation}
because $\vert Z \vert=1$. By equations (\ref{eq-causal-past}) and (\ref{eq-causal-past2}), we obtain that 
\[
\left\{
\begin{tabular}{l}
$v_0^2 > v^2_1 + \ldots + v^2_{m-1}$ \\
$v_0<0$ 
\end{tabular} 
\right.
\Rightarrow 
\left\{
\begin{tabular}{l}
$v$ is timelike \\
$v$ is past--directed
\end{tabular} 
\right.
\]

The causal case is shown analogously.
\qed
\end{proof}

The following proposition characterizes the causality of $M$ in terms of Legendrian isotopies of skies in $\mathcal{N}$.

\begin{proposition}\label{prop00300} 
A regular curve $\mu:\left[ a,b\right] \rightarrow M$ is causal past--directed (respectively causal
future--directed, timelike past--directed, timelike future--directed) if and only if $ \{S\left(\mu(s)\right)\}_{s\in [a,b]}$ is a non-negative (respectively
non-positive, positive, negative) Legendrian isotopy.
\end{proposition}

\begin{proof}

Fix an auxiliary metric $\mathbf{g} \in \mathcal{C}$ in $M$ and denotes by $\mathscr{P}:T_{\mu\left(0\right)}M\times\left[0,1\right]\rightarrow TM$
the parallel transport along $\mu$ with respect to $\mathbf{g}$ given by $\mathscr{P}\left(u,s\right)=u_s\in T_{\mu\left(s\right)}M$. Since $\mathscr{P}$ is differentiable and the map $\mathscr{P}_s:T_{\mu\left(0\right)}M\rightarrow T_{\mu\left(s\right)}M$ defined by $\mathscr{P}_s\left(u\right)=\mathscr{P}\left(u,s\right)$ is a linear isometry, then for any $s\in [0,1]$
\[
\mathbf{g}\left( u_{s},u_{s}\right) =\mathbf{g}\left( u,u\right) =0, \qquad u \in \mathbb{N}^{+}_{\mu(0)}
\]
hence $u_s \in \mathbb{N}^{+}_{\mu(s)}$ and $\mathscr{P}_{s}\left( \mathbb{N}^{+}_{\mu \left(0\right) }\right) =\mathbb{N}^{+}_{\mu \left( s\right) }$ for any metric $\mathbf{g}$ in the conformal structure $\mathcal{C}$. 

Now, consider the submersion $\boldsymbol{\upgamma}:\mathbb{PN}\rightarrow \mathcal{N}$ given by $\boldsymbol{\upgamma}\left([u]\right)=\gamma_{\left[u\right]}$ (see section \ref{sec:lightrays-struct-N}). By composition, the map $\boldsymbol{\upgamma}\circ \pi_{\mathbb{PN}}^{\mathbb{N}^{+}}\circ \mathscr{P}$
is differentiable and, since $\mathscr{P}_s$ is linear then 
\[
\pi_{\mathbb{PN}}^{\mathbb{N}^{+}}\circ \mathscr{P}(u)=\left[ \mathscr{P}_s (u) \right]=\left[ \mathscr{P}_s (\lambda u) \right]=\pi_{\mathbb{PN}}^{\mathbb{N}^{+}}\circ \mathscr{P}(\lambda u)\in \mathbb{PN}
\]
for all $\lambda\neq 0$, so it induces a map $F^\mu:S_{\mu(0)}\times \left[ 0,1\right] \rightarrow\mathcal{N}$ given by
\[
F^{\mu }\left( \left[ u\right] ,s\right) =\gamma _{\left[ u_{s}\right] }
\]
with $S_{\mu(0)}=\lbrace\left[ u\right] :u\in \mathbb{N}^{+}_{\mu \left(0\right) }\rbrace \subset \mathbb{PN}$.
So, for $s\in \left[ 0,1\right] $ we have
\begin{align*}
F^{\mu }\left( S_{\mu(0)}\times \{s\}\right) & =  \{F^{\mu }\left( \left[ u\right]
,s\right) \in \mathcal{N}:u\in \mathbb{N}^{+}_{\mu \left( 0\right) }\}
= \{\gamma _{\left[ u_{s}\right] }\in \mathcal{N}:u\in \mathbb{N}^{+}_{\mu
\left( 0\right) }\}= \\
& = \{\gamma _{\left[ v\right] }\in \mathcal{N}:v\in \mathbb{N}^{+}_{\mu \left(
s\right) }\}
= S\left( \mu \left( s\right) \right)
\end{align*}%
Hence, $F^{\mu }$ is a parametrization of the Legendrian isotopy $ \{S\left(\mu(s)\right)\}_{s\in [a,b]}$.

Now, notice that, since $F^{\mu}\left(\left[u\right],s\right)=\gamma_{\left[u_s\right]}$, we can use the exponential map to give an affine parametrization to all light rays in the variation $\gamma_{\left[u_s\right]}$, so we write
\[
F^{\mu}\left(\left[u\right],s\right)\left(t\right)=\gamma_{\left[u_s\right]}\left(t\right)=\mathrm{exp}_{\mu\left(s\right)}\left(tu_s\right)  .
\]
By lemma \ref{lemmaDC92} and since $u_s$ is the parallel transport of $u$ along $\mu$, we have that the Jacobi field $J_{s}\left(t\right)$ along $\gamma_{\left[u_s\right]}$ verifies
\[
J_{s}\left(0\right)=\mu ^{\prime }\left(s\right), \qquad J^{\prime}_{s}\left(0\right)=\left.\frac{D}{ds}\right|_{s}u_s =0  .
\]
Hence, since
\[
F_{*}^{\mu}\left(\frac{\partial}{\partial s}\right)_{\left(\left[u\right],s\right)}= \frac{\partial F^{\mu}}{\partial s}\left(\left[u\right],s\right)=\langle J_{s}\rangle
\]
we have that
\begin{align*}
\alpha\left(F_{*}^{\mu}\left(\frac{\partial}{\partial s}\right)
\right)_{\left(\left[u\right],s\right)} & = \alpha\left(
\langle J_{s}\rangle\right)=\mathbf{g}\left(J_{s}\left(t\right),\gamma
^{\prime }_{\left[u_{s}\right]}\left(t\right)\right)= \\
& =\mathbf{g}\left(J_{s}\left(0\right),\gamma ^{\prime }_{\left[u_{s}\right]}\left(0\right)\right) =\mathbf{g}\left(\mu ^{\prime
}\left(s\right),u_{s}\right)  .
\end{align*}
Therefore, the statement follows from the equality  
\begin{equation}\label{eq-causal-isotopy}
\alpha\left(F_{*}^{\mu}\left(\frac{\partial}{\partial s}\right)
\right)_{\left(\left[u\right],s\right)}  = \mathbf{g}\left(\mu ^{\prime}\left(s\right),u_{s}\right) 
\end{equation}
in virtue of lemmas \ref{lemma00250} and \ref{lemma00350}.
\qed
\end{proof}

\begin{remark}
Recall that a causal curve $\mu :\left[ a,b\right] \rightarrow M$ is assumed to be continuous, piecewise differentiable and time--oriented. Even being $\mu$ smooth, it may not be regular at some $s=s_0$, , that is, $\mu'(s_0)=0$. In this case, by equation (\ref{eq-causal-isotopy}), we have $\alpha\left(F_{*}^{\mu}\left(\frac{\partial}{\partial s}\right)
\right)_{\left(\left[u\right],s_0\right)} =0$ in spite of $\mu$ could be timelike.
\end{remark}

Therefore, if $\mu$ is not regular, the statement of proposition \ref{prop00300} only can be enunciated as: 

\begin{corollary}
A smooth curve $\mu=\mu(s)$ is causal past--directed (respectively future--directed) if and only if $ \{S\left(\mu(s)\right)\}$ is a non-negative (respectively non-positive) Legendrian isotopy.
\end{corollary}

These results characterize the causality of $M$, that is 
\begin{equation}\label{eq-causal-struct}
q\in J^{+}\left(p\right) \Longleftrightarrow S\left(p\right)\leq_{\Sigma} S\left(q\right)
\end{equation}
where $\leq_{\Sigma}$ denotes the existence of a non--positive Legendrian isotopy consisting of skies.

In \cite{Ch16}, an analogous result is stated under the hypotheses of $M$ simply connected and globally hyperbolic: $q\in J^{+}\left(p\right)$ if and only if there is a non--positive Legendrian isotopy (not necessarily consisting of skies) connecting $S\left(p\right)$ and $S\left(q\right)$.

\subsection{Twisted null curves}\label{sec:Twisted}
 
A key element to show unnecessary the condition of non-refocussing is the notion of twisted null curve. These curves, roughly speaking, are null curves non-geodesic at any point. The precise definition is as follows. 

\begin{definition}\label{def-twisted-null} 
A continuous curve $\mu:\left[a,b\right]\rightarrow M$ is called a \emph{twisted null curve} if 
\begin{enumerate}
\item[i.] $\left. \mu \right|_{\left(a,b\right)}$ is differentiable.
\item[ii.] $\mathbf{g}\left(\mu'\left(s\right),\mu'\left(s\right)\right)=0$ for all $s\in \left(a,b\right)$.
\item[iii.] $\mu'\left(s\right)$ and $\frac{D\mu'}{ds}\left(s\right)$ are linearly independent for all $s\in \left(a,b\right)$.
\end{enumerate}

If $\mu:\left[a,b\right]\rightarrow M$ is continuous and there exists a partition $a=s_0 < s_1 <\ldots < s_k = b$ such that $\left. \mu \right|_{\left(s_{i-1},s_i\right)}$ is a twisted null curve for every $i=1,\ldots,k$, then $\mu$ is called a \emph{piecewise twisted null curve}. 
Moreover, we say that $\mu$ is future--directed (respectively past--directed) if $\mu\mid_{\left(s_{i-1},s_i\right)}$ is future--directed (respectively past--directed) for all $i = 1,\ldots,k$.  
\end{definition}

Observe that, in the definition of twisted null curve, we are considering that $\mu$ contains its endpoints because the objective is to connect points with this kind of curves. A general definition of these curves without endpoints can also be used for other purposes. 

It is clear that these curves $\mu$ can only exist when $\mathrm{dim}(M)\geq 3$ because, if $M$ is $2$--dimensional, $\mu$ could not be twisted so that $\mu'$ remains null. According to this, we state the following result.

\begin{theorem}\textbf{(Twisted null curve theorem)}\label{mu-teorema}
Let $p,q\in M$ such that $q\in I^+(p)$, then there exists a future--directed piecewise twisted null curve $\mu$ joining $p$ to $q$.
\end{theorem}

\begin{proof}[Sketch of the proof]
The details of the demonstration are in \cite[Thm. 9]{Ba15}. The first step is to show that, for the $3$--dimensional case, it is possible to connect locally two points lying in a future--directed timelike geodesic $\lambda$ by a future--directed twisted null curve. It is specially helpful to consider a synchronous coordinate system in a neighbourhood of $\lambda$ (see \cite[Def. 7.13]{Pe72} or \cite[Sec. 99]{La71}). Next, it is possible to generalize the previous step for general dimension $m\geq 3$ if we apply it to some $3$--dimensional submanifold. Then, if $q\in I^{+}(p)$ then there is a timelike curve $\beta$ with endpoints at $p$ and $q$. If we consider a finite covering $\{W_i\}$ of the timelike segment $\beta$ such that each $W_i\subset M$ is causally convex, then every pair of points in $\beta\cap W_i$ can be connected by a timelike geodesic segment. Therefore, by gluing a finite amount of these timelike geodesic segment, we obtain a future--directed broken timelike geodesic. Finally, we can apply the previous steps to the resultant curve to conclude the proof.
\qed
\end{proof}

The theorem \ref{mu-teorema} is a result independent of the space of light rays but, since twisted null curves can be described in terms of $\mathcal{N}$ and $\Sigma$, it can be written in the context of section \ref{sec:Causalcurves}. 

In sections \ref{sec:Topology} and \ref{sec:Reconstruction}, we have introduced some subsets of $T\mathcal{N}$ such that their vectors are tangent to skies, for example,  in theorem \ref{theorem1} we have shown that $\widehat{U}=\bigcup_{X\in U}\widehat{T}X$ is a regular submanifold of $T\mathcal{N}$ where $V\subset M$ is a causally convex, normal, globally hyperbolic open set and $U=S(V)\subset \Sigma$.
This kind of vectors can characterize twisted null curves. 

\begin{definition}\label{def-celestial}
A non-zero tangent vector $\langle J \rangle\in \widehat{T}_\gamma\mathcal{N}$ will be called a \emph{celestial vector} if there exists a sky $X\in \Sigma$ such that $\langle J \rangle \in \widehat{T}_\gamma X$, where $\widehat{T}X=TX-\{0\}$ as in (\ref{eq-widehat-T}).  Consequently, a smooth curve $\Gamma \colon I \to \mathcal{N}$ is called a \emph{celestial curve} if $\Gamma'(s)$ is a celestial vector for all $s \in I$.
\end{definition}

The conditions to build a variation of light rays in $M$ defining a celestial curve in $\mathcal{N}$ are stated in the following proposition.

\begin{proposition}\label{proposition4.3} 
If $\Gamma :\left[ 0,1\right]\rightarrow \mathcal{N}$ is a celestial curve such that $\Gamma \left( s\right) =\gamma _{s}\in \mathcal{N}$, then there exists a null curve $\mu :\left[ 0,1\right] \rightarrow M$ such that $\gamma_{s}\left( \tau \right) =\exp _{\mu \left( s\right) }\left( \tau \sigma\left( s\right) \right) $ where $\sigma \left( s\right) \in \mathbb{N}^{+}_{\mu\left( s\right) }$ is a smooth curve proportional to $\mu ^{\prime}\left( s\right) $ wherever $\mu $ is regular.
\end{proposition}

\begin{proof}
Let $\Gamma \colon [0,1] \to \mathcal{N}$ be a celestial curve with $\Gamma (s) = \gamma_s$. Consider a local chart $(\widehat{U}, \overline{\varphi})$, with $\overline{\varphi} = (\mathbf{x}, \mathbf{u}, \mathbf{v})$ as in (\ref{overlinevarphi}) at some point $\Gamma'(z)\in T_{\gamma_{z}} S(\gamma_{z}(t))$ for $z \in [0,1]$ and $t\in \mathbb{R}$. We can take $(V,\varphi)$ the coordinate chart in $V\subset M$ such that $\varphi=\mathbf{x}$ according to remark \ref{remark-carta-U-hat}.
Consider a neighbourhood $I\subset [0,1]$ of $z$ such that $\Gamma'(s) \in \widehat{U}$ for all $s \in I$.  Then
\[
\overline{\varphi} (\Gamma'(s)) = (\mathbf{x}(\Gamma'(s)), \mathbf{u}(\Gamma'(s)), \mathbf{v}(\Gamma'(s))) 
\]
is a smooth curve.   Notice that the curve $\mu (s) = \varphi^{-1} \circ \mathbf{x}(\Gamma'(s)) \in M$ is also smooth and describes the points to whose skies $\Gamma(s)$ are tangent. Then, the coordinates $\mathbf{u}$ of $\Gamma'(s)$ permit to define the smooth curve $\sigma\subset \mathbb{N}^{+}$ such that if $\mathbf{u}=\left(u^i\right)$ then
\[
\sigma (s) = \mathbf{E}_1(\mu(s)) + u^2(\Gamma'(s)) \mathbf{E}_2(\mu(s)) + \cdots + u^m(\Gamma'(s)) \mathbf{E}_m (\mu(s)) \in \mathbb{N}^{+}_{\mu(s)} \, .
\]
where $\{\mathbf{E}_1, \ldots , \mathbf{E}_m \}$ is the orthonormal frame associated to the local chart $(\widehat{U}, \overline{\varphi})$. Hence, the geodesic variation $\mathbf{f}(s,\tau ) = \exp_{\mu(s)} (\tau \sigma (s)) = \overline{\gamma}_s (\tau) $ defines the curve $\Gamma (s)$.

By lemma \ref{lemmaDC92}, the Jacobi field $\overline{J}_s$ along $\overline{\gamma}_s$ of the variation $\mathbf{f}$ satisfies that $\overline{J}_s(0) = \mu'(s)$ and, since $\Gamma'(s)\in T_{\Gamma(s)}S\left(\mu(s)\right)$, then  $\overline{J}_s(0) = \alpha_s \overline{\gamma}'_s(0)$ for some $\alpha_s \in \mathbb{R}$.  Then, $\mu'$ must be proportional to $\overline{\gamma}'_s(0)$, therefore also proportional to $\sigma(s)$.

The extension of $\mu$ and $\sigma$ to the interval $[0,1]$ follows from the compactness of $\Gamma$. 
\qed
\end{proof}

It is proved in \cite[Cor. 7]{Ba15} that, given a celestial curve $\Gamma:\left[ 0,1\right]\rightarrow \mathcal{N}$ such that $\Gamma'\left(s_0\right)\in \widehat{T}S\left(p_0\right)$, then the curve $\mu: \left[ 0,1\right]\rightarrow M$ of the previous proposition verifying $\mu\left(s_0\right)=p_0\in M$ is unique. The proof can be obtained by \emph{reductio ad absurdum}, assuming there are two such curves $\mu_1,\mu_2$ and showing that the set $A=\left\{s\in \left[0,1\right]:\mu_1\left(s\right)=\mu_2\left(s\right)  \right\}$ is closed, open and not empty, therefore $A=[0,1]$.

The unique curve $\mu$ (in the sense of proposition \ref{proposition4.3}) passing by $p = S^{-1}(X)$ of a given celestial curve $\Gamma$ will  be called the \emph{dust} or \emph{trail} of $\Gamma$ by $X$ and denoted by $\mu_{X}^\Gamma$.   
The following proposition shows that the dust of a celestial curve is a twisted null curve. See \cite[Lem. 8]{Ba15} for proof.

\begin{proposition}\label{mu-lemma}
Let $\Gamma:\left[0,1\right]\rightarrow \mathcal{N}$ be a celestial curve such that $\Gamma'\left(0\right)\in \widehat{T}X_0$ with $X_0\in \Sigma$. 
Then there exists a unique curve $\chi_{X_0}^{\Gamma}\colon \left[0,1\right]\rightarrow \Sigma$, continuous in Low's topology, verifying  $\chi_{X_0}^{\Gamma}\left(0\right)=X_0$ and $\Gamma'\left(s\right)\in \widehat{T}\chi_{X_0}^{\Gamma}\left(s\right)$. Moreover, the dust curve $\mu^{\Gamma}_{X_0}$ is a piecewise twisted null curve in $M$ along the image of $S^{-1} \circ \chi_{X_0}^{\Gamma}$.

Conversely, given a regular twisted null curve $\mu\colon \left[0,1\right]\rightarrow M$ such that $\mu\left(0\right)=x_0=S^{-1}\left(X_0\right)$ and $\mu'(0) \neq 0 \neq \mu'(1)$, then the variation of null geodesics  
\[
f\left(s,t\right)=\mathrm{exp}_{\mu\left(s\right)}\left(t\mu'\left(s\right)\right)=\left.\Gamma\left(s\right)\right|_{t}
\]
with $s\in\left[0,1\right]$, defines a celestial curve $\Gamma:\left[0,1\right]\rightarrow \mathcal{N}$ such that $\Gamma'\left(0\right)\in \widehat{T}X_0$ and $\chi_{X_0}^{\Gamma}\left(s\right)=S\left(\mu\left(s\right)\right)$.
\end{proposition}

Previous proposition \ref{mu-lemma} says that celestial curves in $\mathcal{N}$ and twisted null curves in $M$ are two sides of the same coin, and consequently, it permits to interpret theorem \ref{mu-teorema} according to proposition \ref{prop00300} as the next result. 

\begin{corollary}\label{cor-celestial-curve}
$q\in I^{+}(p)$ if and only if there exists a celestial curve $\Gamma:\left[0,1\right]\rightarrow \mathcal{N}$ and a continuous curve of skies $\chi^{\Gamma}:\left[0,1\right]\rightarrow \Sigma$ such that
\begin{enumerate}
\item $\chi^{\Gamma}$ is a non--positive Legendrian isotopy with $\chi^{\Gamma}\left(0\right)= S\left(p\right)$ and $\chi^{\Gamma}\left(1\right)=S\left(q\right)$.
\item $\Gamma ' \left(s\right)\in T_{\Gamma\left(s\right)}\chi^{\Gamma}\left(s\right)$.
\end{enumerate} 
\end{corollary}

So, the connection by a signed Legendrian isotopy between the skies $S(p),S(q)\in \Sigma$ of any pair of chronological related points $q\in I^{+}(p)$, ensured by proposition \ref{prop00300}, can be achieved by the skies $\chi^{\Gamma}$ pointed by a celestial curve $\Gamma$.

\subsection{Legendrian linking}\label{sec:Linking}

Related to causality in the framework of the space of light rays, we find the topic of \emph{sky--linking}. 
This interesting subject is out of the scope of this review, but we think it is worth introducing the general problem and some references. 

By the equivalence between the conformal manifold $M$ and its space of skies $\Sigma$ stated in theorem \ref{teo-sky-difeo} and according section \ref{sec:isotopies}, we have that any (smooth) curve $\lambda:I\rightarrow M$ defines a Legendrian isotopy of skies $\{S\left(\lambda(s)\right)\}_{s\in I}$ in $\mathcal{N}$. Then all skies in $\Sigma$, trivially, can be connected by a Legendrian isotopy. 

At his PhD thesis \cite{Lo88}, Low noticed that the isotopy classes for pairs of skies $S(p)\sqcup S(q)$ depend on the causal relation between $p,q\in M$. 
Later, Penrose suggested this conjecture to Nat\'{a}rio et al. as a problem to be solved motivating \cite{Na04}. In this reference, the authors proposed a modification \cite[Conj. 6.4]{Na04} of the Low's conjecture to catch the past or future orientation of the skies of the link $S(p)\sqcup S(q)$, so they consider Legendrian isotopies.

In the literature, a pair $S(p)\sqcup S(q)$ in $\mathcal{N}$ is called a \emph{link} and it is shown, for example in \cite[Lem. 4.3]{Ch10}, that all links $S(p)\sqcup S(q)$ of skies of causally unrelated points $p,q\in M$ lies in the same class of (Legendrian) isotopies. So, in this case, it is said that $S(p)\sqcup S(q)$ is \emph{(Legendrian) unlinked} or also \emph{(Legendrian) trivially linked}. There might be a bit of confusion with the use of the word ``link", since a link $S(p)\sqcup S(q)$ might be unlinked, but we use this terminology as in the references. 

The development of this topic done by Low in \cite{Lo88}, \cite{Lo90}, \cite{Lo94}, \cite{Lo00}, by Field and Low \cite{Fi98}, by Nat\'{a}rio et al. \cite{Na00}, \cite{Na02}, \cite{Na04} and by Chernov et al. \cite{Ch08} led Chernov and Nemirowski to prove both conjectures in \cite{Ch10}.

\begin{theorem}[Low's conjecture]
Let $M$ be a $3$--dimensional spacetime, globally hyperbolic with Cauchy surface diffeomorphic to a subset of $\mathbb{R}^{2}$.
Two points $p,q\in M$ are causally related if and only if $S(p)\sqcup S(q)$ are linked.
\end{theorem}

\begin{theorem}[Legendrian Low conjecture]
Let $M$ be a globally hyperbolic $m$--dimensional spacetime with Cauchy surface diffeomorphic to a subset of $\mathbb{R}^{m-1}$.
Then two points $p,q\in M$ are causally related in $M$ if and only if their skies either intersect or are Legendrian linked in $\mathcal{N}$.
\end{theorem}

Some recent work on the extension of these results for  more general spacetimes are \cite{Ch18b} and \cite{Ch18c}.

\section{\texorpdfstring{$L$--boundary}{L--boundary}}\label{sec:LBoundary}

As it was mentioned in the introduction, we try to point out a possible way to construct a new causal boundary for $M$, called the \emph{light boundary} or \emph{$L$--boundary}, introduced by R. Low in \cite{Lo06}. 
It will be a conformal extension of the conformal manifold and its construction requires the use of geometrical structures of the space of light rays.
A first study of this new boundary is done in \cite{Ba17} including some initial results and examples comparing it with the \emph{c-boundary} developed in \cite{GKP68}. Later, a more complete article \cite{Ba18} was release.

\subsection{Preliminary: the causal boundary}\label{sec:causal-boundary}

In \cite{GKP68}, the \emph{$c$--boundary} or \emph{causal boundary} is introduced to improve others already existing boundaries (Geroch's \cite{Ge68}, Schmidt's \cite{Sc71}, etc.) It is conformally invariant and defined intrinsically adding endpoints, as subsets in the spacetime $M$, to causal inextensible curves. 

\begin{definition}
We will say that $W\subset M$ is an \emph{indecomposable past set} or \emph{IP} if $W$ is open, non--empty and a past set, that is $I^{-}(W)=W$, which can not be expressed as the union of two proper subsets satisfying the same previous properties.

If there exists $p\in M$ such that $I^{-}(p)=W$, we will say that $W$ is a \emph{proper IP} or \emph{PIP}. We will name $W$ as a \emph{terminal IP} or \emph{TIP} whenever such $p\in M$ does not exist.

The \emph{future causal boundary} or \emph{future $c$--boundary} is defined by the set of all TIP contained in $M$. 

In an analogous way, the \emph{past $c$--boundary} is defined from \emph{indecomposable future sets}, \emph{PIFs} and \emph{TIFs}.
\end{definition}

The following proposition characterizes the TIPs as chronological past of future inextendible timelike curves and states that points in the $c$--boundary can also be defined by causal curves.
Analogue results can be stated for TIFs.
See \cite[Prop. 6.8.1]{HE} and \cite[Prop. 3.32(i)]{Fl11} for proofs. 

\begin{proposition}\label{Prop-TIPs}
Let $M$ be a strongly causal spacetime. 
\begin{enumerate}
\item A set $W\subset M$ is a TIP if and only if there exists a future inextendible timelike curve $\lambda\subset M$ such that $I^{-}(\lambda)=W$. \item If $\lambda\subset M$ is a future inextendible causal curve, then $I^{-}(\lambda)$ is a TIP.
\end{enumerate}
\end{proposition}

An immediate consequence of proposition \ref{Prop-TIPs} is the following corollary.

\begin{corollary}\label{cor-TIPs}
Let $M$ be a strongly causal spacetime. 
$I^{-}(\gamma)$ is a TIP for any light ray $\gamma\in\mathcal{N}$.
\end{corollary}

Observe that the points $p\in M$ can be identified by its corresponding PIPs $I^{-}(p)$ or PIFs $I^{+}(p)$ and the gaps and points at infinity of $M$ are identified with TIPs $I^{-}(\lambda)$ or TIFs $I^{+}(\lambda)$ for some future inextensible timelike curve $\lambda$. Therefore, by corollary \ref{cor-TIPs}, light rays also define points in the $c$--boundary.

\begin{figure}[h]
  \centering
\begin{tikzpicture}[scale=1]
\fill[rounded corners,color=gray!5]
(-6,0) -- (-5,2) -- (-2,2.3) -- (3,2.1) -- (5.5,1.9) -- (6,0) -- (5,-2) -- (4,-2.1)--(-2,-1.9)--(-5.5,-2.2)--(-6,0);
\fill[color=blue!15] (-5.9,0.2)--(-4,2.1)--(0.065,-1.965)--(-2,-1.9)--(-5.4,-2.14)--(-5.525,-2.04)--(-6,0)--(-5.9,0.2);
\fill[color=blue!15,rounded corners] (-5.9,0.2)--(-4,2.1)--(0.065,-1.965)--(-2,-1.9)--(-5.5,-2.2)--(-5.9,0.2);
\fill[color=emerald!15] (1.5,-0.5)--(3.06,-2.06)--(0.0465,-1.965)--(1.5,-0.5);
\fill[color=red!15] (3.425,-2.075)--(5.9,0.4)-- (5.95,0.2)--(4.9,-1.9)--(3.425,-2.075);
\fill[color=red!15,rounded corners] (3.425,-2.075)--(5.9,0.4)-- (5.95,0.2)-- (6,0) -- (5,-2) -- (4,-2.1)--(3.425,-2.075);
\draw[dashed,rounded corners]
(-6,0) -- (-5,2) -- (-2,2.3) -- (3,2.1) -- (5.5,1.9) -- (6,0) -- (5,-2) -- (4,-2.1)--(-2,-1.9)--(-5.5,-2.2)--(-6,0);
\draw[thick,color=red] (5.2,-1.2) to[out=88, in=-92] (5.9,0.4);  
\draw[color=red] (5.5,-0.4) node[anchor=east] {$\lambda$};       
\draw[fill=white] (5.9,0.4) node[anchor=west] {$\widetilde{q}$} circle (2pt);
\draw[thick,color=blue] (-4,0.5) to[out=120, in=-60] (-4,2.1);       
\draw[thick,color=blue] (-4.1,1.2) node[anchor=east] {$\mu$};        
\draw[thick,color=blue] (0.065,-1.965)--(-4,2.1);       
\draw[color=blue] (-2,0.1) node[anchor=north east] {$\gamma$};  
\draw[fill=white] (-4,2.1) node[anchor=south] {$\widetilde{r}$} circle (2pt);
\draw[thick,color={rgb:red,0;green,128;blue,128}] (1,-1.5) to[out=88, in=-92] (1.5,-0.5);   
\draw[thick,color={rgb:red,0;green,128;blue,128}] (1.3,-1) node[anchor=west] {$\beta$};     
\draw[fill=white] (1.5,-0.5) node[anchor=south] {$\widetilde{p}$} circle (2pt);
\draw[fill=white] (5.3,1.7) node[anchor=north east] {$M$};
\draw[color=blue] (-3.8,-0.8) node {$I^{-}(\mu)$};
\draw[color=red] (4.6,-1.7) node {$I^{-}(\lambda)$};
\draw[color={rgb:red,0;green,128;blue,128}] (1.8,-1.7) node {$I^{-}(\beta)$};
\end{tikzpicture}
  \caption{The ideal points $\widetilde{r}$, $\widetilde{p}$ and $\widetilde{q}$ are respectively identified with the chronological past of the future inextensible timelike curves $\mu$, $\beta$ and $\lambda$. Observe that the point $\widetilde{p}$ has been removed from $M$, so it is a \emph{gap} in the spacetime. The light ray $\gamma$ also defines the same TIP as $\mu$ because $I^{-}(\gamma)=I^{-}(\mu)$.}
  \label{TIPs}
\end{figure}
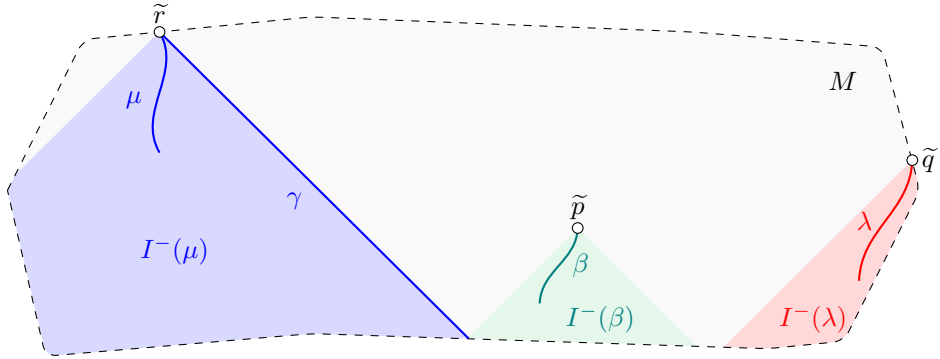

\begin{remark}\label{remark-light-ray-no-accessible}
As proposition \ref{Prop-TIPs} suggests in its statements, not every TIP can be defined by light rays because such ideal points may be accessible only by timelike curves. For example, let us consider $M=\{(t,x)\in\mathbb{M}^2 : t<-2\vert x \vert \}$ as submanifold of the $2$--dimensional Minkowski spacetime with the standard metric $\mathbf{g}=-dt\otimes dt + dx \otimes dx$. The ideal point $(0,0)$ is not accessible by light rays contained in $M$. In fact, the TIP defined by the curve $\lambda(s)=(s,0)$ for $s\in(-\epsilon, 0)$ is the whole spacetime $M$, that is $(0,0)\sim I^{-}(\lambda)=M$ (see figure \ref{figura-accesible-TIP}). This example can be easily generalized to higher dimensions.
\end{remark}

\begin{figure}[h]
  \centering
\begin{tikzpicture}[scale=1]
\draw[-latex] (-1.5,0) -- (1.5,0) node[anchor=south west] {$x$};
\draw[-latex] (0,-2) -- (0,0.5) node[anchor=east] {$t$};
\draw[dashed,fill=blue!20,opacity=0.3] (-1,-2) -- (0,0)--(1,-2);
\draw[red,very thick] (0,-0.8) node[anchor=east] {$\lambda$} -- (0,0);
\fill[green!30,opacity=0.3] (-0.2,-2)--(0.6,-1.2)--(1,-2) -- cycle;
\draw[teal,very thick] (-0.2,-2)--(0.4,-1.4) node[anchor=south east] {$\gamma$} --(0.6,-1.2);
\draw[blue] (-0.5,-2) node[anchor=south] {$M$};
\draw (-1.2,-0.5) node {$\mathbb{M}^2$};
\draw[color=teal] (0.9,-1.8) node[anchor=south west] {$I^{-}(\gamma)$};
\draw[fill=white] (0,0) node[anchor=south west] {$(0,0)$} circle (1pt);
\draw[fill=white] (0.6,-1.2)  node[anchor=south west] {$\widetilde{p}$} circle (1pt);
\end{tikzpicture}
  \caption{The point $(0,0)$ is not accessible by light rays in $M$. It corresponds to the TIP $I^{-}(\lambda)=M$ defined by a timelike curve such $\lambda$. Light rays define ideal points such as $\widetilde{p}$ which correspond to TIPs like $I^{-}(\gamma)$ in green.}
  \label{figura-accesible-TIP}
\end{figure}
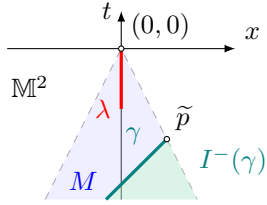

The identification $p\sim I^{-}(p)\sim I^{+}(p)$ presents some inconsistencies affecting even to the topology of the spacetime's completion and many authors have tried to solve this problem (Budic et al. \cite{Bu74}, R\'{a}cz \cite{Ra87}, Szabados \cite{Sz88}, \cite{Sz89}, Harris \cite{Ha98}, \cite{Ha00}, Marolf et al. \cite{Ma02}, \cite{Ma03} among others) but finally, Flores et al. \cite{Fl11} stated a consistent definition of the $c$--boundary. Anyway, this matter is beyond the scope of this review, so, for our purposes, we will only need the classical definition of the causal boundary above.
See \cite{GaP05}, \cite{Sa09} and \cite{Fl11} for an overview on this subject.

\subsection{Low's idea}\label{sec:LBoundary-idea}

As indicated in \cite{Lo06}, Low introduces a new idea for the construction of a causal and conformal boundary in $M$. 
It consists in the addition of \emph{skies at infinity}, in virtue of the equivalence between events and skies by the sky map (see section \ref{sec:lightrays-skies}). 

It is expected that a sky at infinity will be the \emph{limit} of the skies of the points of a fixed light ray when the parameter tends to the future/past endpoint of its domain of definition. 
Low wonders if this is a new way to obtain the $c$--boundary using the geometry of the space of light rays.

In order to build these boundary points, we will take some future--directed inextensible null geodesic $\gamma:\left(a,b\right)\rightarrow M$, then we will consider the curve $\widetilde{\gamma}:\left(a,b\right)\rightarrow \mathrm{Gr}^{m-2}\left(\mathcal{H}_{\gamma}\right)$ defined by 
\[
\widetilde{\gamma}\left(s\right)= T_{\gamma}S\left(\gamma\left(s\right) \right) 
\]
where $S(\gamma(s))\in \Sigma$ is the sky of $\gamma(s)\in M$.  
Since each $S(p)$ is diffeomorphic to the sphere $\mathbb{S}^{m-2}$, then $T_{\gamma}S\left(\gamma\left(s\right) \right)$ is contained in the Grassmannian manifold $\mathrm{Gr}^{m-2}\left(\mathcal{H}_{\gamma}\right)$ of $\left(m-2\right)$--dimensional subspaces of $\mathcal{H}_{\gamma}\subset T_{\gamma}\mathcal{N}$.
Then, we can define endpoints of the curve $\widetilde{\gamma}$ by
\begin{equation}\label{boundary-field}
\begin{tabular}{l}
$\ominus_{\gamma} = \lim_{s\mapsto a^{+}}\widetilde{\gamma}\left(s\right)\in \mathrm{Gr}^{m-2}\left(\mathcal{H}_{\gamma}\right)$ , \vspace{3mm} \\
$\oplus_{\gamma} = \lim_{s\mapsto b^{-}}\widetilde{\gamma}\left(s\right)\in \mathrm{Gr}^{m-2}\left(\mathcal{H}_{\gamma}\right)$ ,
\end{tabular}
\end{equation}
when the limits exist. 
In general, the existence of $\ominus_{\gamma}$ and $\oplus_{\gamma}$ is not clear, although the compactness of $\mathrm{Gr}^{m-2}\left(\mathcal{H}_{\gamma}\right)$ assures the existence of accumulation points when $s\mapsto a^{+},b^{-}$. 
Moreover, even in case of the existence of the limits, we wonder if $\ominus,\oplus:\mathcal{N}\rightarrow \mathrm{Gr}^{m-2}\left(\mathcal{H}\right)$ are smooth distributions in $\mathcal{N}$. 
Low defines the future/past endpoint of the light ray $\gamma\subset M$ for this new boundary of $M$ as the integral manifold of the distributions $\oplus / \ominus$, which will comprise all light rays \emph{arriving at/emerging from} the same point at infinity than $\gamma$, see \cite{Lo06}.

Recall that, at the end of section \ref{sec:lightrays-contact}, we have seen that $\left(\mathcal{H}_{\gamma},\left.\omega\right|_{\mathcal{H}_{\gamma}}\right)$ is a symplectic vector space for any  $\gamma\in \mathcal{N}$, where $\left.\omega\right|_{\mathcal{H}_{\gamma}}$ satisfies the expression (\ref{eq-dalpha-expression}). 
Then, it is easy to show that for any sky $X\in \Sigma$ such that $\gamma\in X$, then $T_{\gamma}X$ is its own symplectic orthogonal vector space, that is, 
\[
T_{\gamma}X=\left(T_{\gamma}X\right)^{\perp}\equiv \{ \langle J \rangle\in \mathcal{H}_{\gamma}:\left.\omega\right|_{\mathcal{H}_{\gamma}}\left(\langle J \rangle,\langle K \rangle\right)=0 \text{ for all } \langle K \rangle\in T_{\gamma}X \}
\]
therefore $T_{\gamma}X$ is a \emph{lagrangian subspace} of $\mathcal{H}_{\gamma}$.
So, if we denote by $\mathscr{L}\left(\mathcal{H}\right)\subset \mathrm{Gr}^{m-2}\left(\mathcal{H}\right)$ the \emph{manifold of Lagrange grassmannian subspaces} in $\mathcal{H}$ and by $\mathscr{L}\left(\mathcal{H}_{\gamma}\right)\subset \mathrm{Gr}^{m-2}\left(\mathcal{H}_{\gamma}\right)$ the submanifold of Lagrange grassmannian subspaces in $\mathcal{H}_{\gamma}$, then $\Lambda\in \mathscr{L}\left(\mathcal{H}_{\gamma}\right)$ if and only if $\dim \Lambda = m-2$ and $\left.\omega\right|_{\Lambda}=0$. 
So, the image of the maps $\ominus,\oplus$ can be restricted to $\mathscr{L}\left(\mathcal{H}\right)$ by
\[
\ominus,\oplus:\mathcal{N}\rightarrow \mathscr{L}\left(\mathcal{H}\right) ,
\]
as well as $\widetilde{\gamma}\subset \mathscr{L}\left(\mathcal{H}_{\gamma}\right)$ holds.

The distributions $\oplus$ and $\ominus$ are independent of each other and they permit to build the future and past boundaries respectively. Therefore, the construction of one of these boundaries is also independent of that of the other.  
Thus, we will describe the construction of the future boundary from the distribution $\oplus$, taking into account that the process to obtain the past boundary from $\ominus$ is analogous.

We propose the following hypotheses for the general case $\dim M \geq 3$:

\begin{enumerate}[start=1,label={\bfseries H\arabic* }]
\item \label{itm:hypotheses-general-1} $\left(M,\mathcal{C}\right)$ is strongly causal, null--pseudo convex, light non--conjugate and sky--separating.
\item \label{itm:hypotheses-general-2} The distribution $\oplus : \mathcal{N} \rightarrow \mathscr{L}\left(\mathcal{H}\right)$, defined by $\oplus_{\gamma}=\lim_{s\mapsto b^{-}}T_{\gamma}S\left(\gamma\left(s\right)\right)$ for any maximally and future--directed parametrized light ray $\gamma:\left(a,b\right)\rightarrow M$, is differentiable and regular.
\end{enumerate}

\begin{definition}
A Lorentz conformal manifold $M$ satisfying conditions \ref{itm:hypotheses-general-1} and \ref{itm:hypotheses-general-2} is said to be a \emph{$L$--spacetime}.
\end{definition}

Notice that \ref{itm:hypotheses-general-1} is required for $\mathcal{N}$ having good topological and differentiable properties and \ref{itm:hypotheses-general-2} are basic conditions on $\oplus$ so that the $L$--boundary could be built.

For the sake of simplicity, we do not use the labels \emph{future} and \emph{past} in the definition of $L$--spacetime as the property \ref{itm:hypotheses-general-2} is verified for $\oplus$ or $\ominus$ respectively. We understand that this may be a bit ambiguous, but in this way we will avoid too many adjectives in later definitions. Therefore, as in what follows, we will build only the future boundary, \emph{$L$--spacetime} should be understood as \emph{future $L$--spacetime}.

\subsection{Hypotheses for the 3--dimensional case}\label{sec:LBoundary-3D}

As a first approximation to the general case, we will study the $L$--boundary for $3$--dimensional conformal manifolds following the original Low's idea. 
Observe that in such case $\mathcal{N}$ is also $3$--dimensional and the Lagrangian grassmannian manifold $\mathscr{L}\left(\mathcal{H}\right)$ becomes $\mathbb{P}\left(\mathcal{H}\right)$ (in fact, $\mathrm{Gr}^{1}\left(\mathcal{H}\right)=\mathscr{L}\left(\mathcal{H}\right)=\mathbb{P}\left(\mathcal{H}\right)$). So the curve $\widetilde{\gamma}$ is contained in $\mathbb{P}\left(\mathcal{H}_{\gamma}\right)\simeq\mathbb{S}^1$.

Notice that, if $M$ is light non--conjugate, then the curve $\widetilde{\gamma}\left(s\right)=T_{\gamma}S\left(\gamma\left(s\right)\right)\in \mathbb{P}\left(\mathcal{H}_{\gamma}\right)\simeq \mathbb{S}^{1}$ is injective and therefore the continuity of $\widetilde{\gamma}$ would imply that the limits $\oplus_{\gamma}$ and $\ominus_{\gamma}$ exist consisting in lines in $\mathbb{P}\left(\mathcal{H}_{\gamma}\right)$.

\begin{figure}[h]
  \centering
\begin{tikzpicture}[scale=1]
\draw[dashed] (1.4,0) arc[start angle=0, end angle=360,radius=1.4cm];
\draw (0,0) node {$\mathbb{P}\left(\mathcal{H}_{\gamma}\right)$};
\draw[-|,thick] (1.4,0) arc[start angle=0, end angle=60,radius=1.4cm];
\draw[-|,thick] (1.4,0) arc[start angle=0, end angle=-210,radius=1.4cm];
\draw[->,thick] (0,-1.2) arc[start angle=-90, end angle=-70,radius=1cm];
\draw (0,-1.2) arc[start angle=-90, end angle=-110,radius=1cm];
\draw (-1.3,0.9) node[anchor=east] {$\ominus_{\gamma}$};
\draw (1,1.3) node[anchor=south] {$\oplus_{\gamma}$};
\draw[-|,thick] (1.4,0) arc[start angle=0, end angle=0,radius=1.4cm];
\draw (1.6,0) node[anchor=west] {$T_{\gamma}S\left(\gamma(s)\right)$};
\draw (0,-1.4) node[anchor=north] {$\widetilde{\gamma}$};
\end{tikzpicture}
  \caption{The curve $\widetilde{\gamma}\subset \mathbb{P}\left(\mathcal{H}_{\gamma}\right)\simeq\mathbb{S}^1$.}
  \label{diapositiva7}
\end{figure}
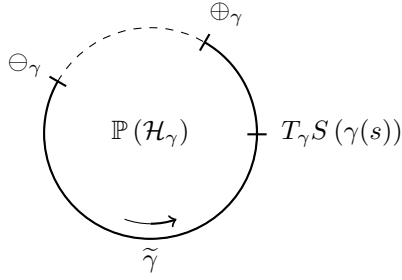

So, when $M$ is $3$--dimensional, the hypothesis \ref{itm:hypotheses-general-2} becomes

\begin{enumerate}[start=2,label={\bfseries H\arabic* }]
\item The distribution $\oplus: \mathcal{N} \rightarrow \mathbb{P}\left(\mathcal{H}\right)$, defined by $\oplus_{\gamma}=\lim_{s\mapsto b^{-}}T_{\gamma}S\left(\gamma\left(s\right)\right)$ for any maximally and future--directed parametrized light ray $\gamma:\left(a,b\right)\rightarrow M$, is differentiable and regular.
\end{enumerate}

The rest of this section \ref{sec:LBoundary} is devoted to the construction of the future $L$--boundary for $m=3$, but we also believe that a similar way can be travelled in order to get the $L$--boundary for any dimension $m\geq 3$. 

\subsection{The space \texorpdfstring{$\widetilde{\mathcal{N}}$}{Ñ} of tangent spaces to skies}\label{sec:LBoundary-Ntilde}

First, let us consider the following natural map
\begin{equation}\label{difeo-sigma}
\begin{tabular}{rrcl}
$\sigma:$ & $\mathbb{PN}$ & $\rightarrow$ & $\mathbb{P}\left(\mathcal{H}\right)$ \\
& $\left[u\right]$ & $\mapsto$ & $T_{\gamma_{\left[u\right]}}S\left(\pi^{\mathbb{PN}}_{M}\left(\left[u\right]\right)\right)$
\end{tabular}
\end{equation}
where $\pi^{\mathbb{PN}}_{M}:\mathbb{PN}\rightarrow M$ is the canonical projection. 
The assumption of $M$ being light non--conjugate, by \cite[Lem. 2.5]{Ba17}, gives us the injectivity of $\sigma$.
But, moreover, this map permits to embed $M$ (by its bundle $\mathbb{PN}$ of null directions) into the geometry of $\mathcal{N}$ (by the bundle $\mathbb{P}\left(\mathcal{H}\right)$ of lines in the contact structure $\mathcal{H}$). 

\begin{proposition}\label{prop-diffeo-sigma}
The map $\sigma:\mathbb{PN}\rightarrow\mathbb{P}\left(\mathcal{H}\right)$ defined in (\ref{difeo-sigma}) is a diffeomorphism onto its image.
Moreover, $\widetilde{\mathcal{N}}=\mathrm{Im}\left(\sigma \right)$ is an open submanifold of $\mathbb{P}\left(\mathcal{H}\right)$.
\end{proposition}

\begin{proof}[Sketch of  the proof]
This proposition in proven in \cite[Prop. 5.1]{Ba18} in two steps. 
First, the differentiability of $\sigma$ is shown \cite[Lem. 5.1]{Ba18} by the construction of $\sigma\left( \left[u\right] \right)\in \mathbb{P}\left(\mathcal{H}\right)$ for $\left[u\right]\in \mathbb{PN}$ by differentiable composition of maps, using null geodesics variations which fix the point $\pi^{\mathbb{PN}}_{M}([u])\in M$. 
Second, the differential $d\sigma_{\left[u\right]}$ is an isomorphism. 
In this second step, it is necessary to prove that the curve $\widetilde{\gamma}\left(t\right)=\sigma\left(\left[\gamma'\left(t\right)\right]\right)$ is regular whenever the parameter $t$ is affine \cite[Lem. 5.2]{Ba18}. 
Finally, since both $\mathbb{PN}$ and $\mathbb{P}\left(\mathcal{H}\right)$ are $4$--dimensional then the image $\sigma\left(\mathbb{PN}\right)=\widetilde{\mathcal{N}}$ is open in $\mathbb{P}\left(\mathcal{H}\right)$.
\qed 
\end{proof}

Observe that, by construction, the manifold $\widetilde{\mathcal{N}}$ is the space of all tangent spaces to skies of points of $M$. 
In \cite[Sec. 3.1]{Ba18}, $\widetilde{\mathcal{N}}$ is named as the \emph{blow up} of $M$. 

Notice that, if $\gamma=\gamma\left(t\right)$ is a null geodesic, then $\sigma\left( \left[\gamma'\left(t\right) \right] \right)=T_{\gamma}S\left(\gamma\left(t\right)\right) \in \mathbb{P}\left(\mathcal{H}_{\gamma}\right)$. 
So, the endpoints of the curve 
\begin{equation}\label{def-gammatilde}
\widetilde{\gamma}\left(t\right) \equiv \sigma\left( \left[\gamma'\left(t\right) \right] \right) = T_{\gamma}S\left(\gamma\left(t\right)\right)
\end{equation}
define the distributions $\oplus$ and $\ominus$.
Assuming the hypotheses \ref{itm:hypotheses-general-1}, we have the following implications
\begin{align*}
\widetilde{\gamma}\left(t_1\right)=\widetilde{\gamma}\left(t_2\right) & \Rightarrow T_{\gamma}S\left(\gamma\left(t_1\right)\right) = T_{\gamma}S\left(\gamma\left(t_2\right)\right) & \text{ (by definition)} \\
& \Rightarrow S\left(\gamma\left(t_1\right)\right) = S\left(\gamma\left(t_2\right)\right) &\text{ (by light non--conjugate)} \\
& \Rightarrow \gamma\left(t_1\right) = \gamma\left(t_2\right) &\text{ (by injectiveness of $S$)} \\
& \Rightarrow t_1 = t_2  &\text{ (by injectiveness of $\gamma$)}
\end{align*}
then any $\widetilde{\gamma}$ is also an injective curve.

In \cite[Prop. 6.1]{Ba18}, it is shown that $\widetilde{\mathcal{N}}\subset \mathbb{P}\left(\mathcal{H}\right)$ is a submanifold with boundary under the hypotheses \ref{itm:hypotheses-general-1}, \ref{itm:hypotheses-general-2} for both $\oplus$ and $\ominus$ with $\oplus_{\gamma}\neq \ominus_{\gamma}$ for all $\gamma\in \mathcal{N}$ . This last hypothesis is not critical, it is just a technical condition to simplify the construction.
It is possible to show a more general statement only under conditions \ref{itm:hypotheses-general-1} and \ref{itm:hypotheses-general-2} with the same procedure used in \cite{Ba18}, which is sufficient for the construction of the future boundary.
Notice that, if $\mathcal{N}_{U}$ is the open set of all light rays passing through the globally hyperbolic, normal and causally convex open $U\subset M$ according to remark \ref{rmk-basic-neighb}, we can consider two smooth spacelike Cauchy surfaces $C,C_{-}\subset U$ such that $C_{-}\subset I^{-}(C)$. 
Observe that if $\mathcal{H}_U=\bigcup_{\gamma\in \mathcal{N}_U}\mathcal{H}_{\gamma}$ then  
\[
\mathbb{P}\left(\mathcal{H}_U\right)\simeq \mathcal{N}_{U} \times \mathbb{S}^{1} \simeq C \times \mathbb{S}^{1}\times \mathbb{S}^{1}  .
\]
So, we can choose a coordinate system $\widetilde{\psi}=(x,y,\theta, \phi)$ where $(x,y,\theta)$ are coordinates for $C\times \mathbb{S}^{1}$ and $\phi\in\left[0,2\pi\right)$ is a coordinate for the fibres $\mathbb{P}\left(\mathcal{H}_{\gamma}\right)\simeq \mathbb{S}^{1}$ that can be built from the initial values of the Jacobi fields at the Cauchy surface $C\subset U$. If $P\in\mathbb{P}\left(\mathcal{H}_{\gamma}\right)$ is a line such that $P=\mathrm{span}\{\langle J \rangle\}$, then we can choose a representative $J$ such that $J(0),J'(0)\in T_{\gamma(0)}C$. Since  $\mathbf{g}\left(J(0),\gamma'(0)\right)=0$ and $\mathbf{g}\left(J'(0),\gamma'(0)\right)=0$, then it is possible to write 
\begin{equation*}\label{initial-vectors-J}
J\left(0\right)=w \cdot \mathbf{e}, \quad \text{ and } \quad J'\left(0\right)=v \cdot \mathbf{e}
\end{equation*} 
where $T_{\gamma(0)}C \cap \left\{\gamma'\left(0\right)\right\}^{\perp}=\mathrm{span}\{\mathbf{e}\}$ with $\left\{\gamma'\left(0\right)\right\}^{\perp} = \left\{u\in T_{\gamma(0)}M: \mathbf{g}\left( \gamma'\left(0\right), u \right) = 0 \right\}$.  
Since for any $0\neq\alpha\in\mathbb{R}$ we have $P=\mathrm{span}\{\langle \alpha J \rangle\}=\mathrm{span}\{\langle J \rangle\}$ then the homogeneous coordinate 
\begin{equation}\label{eq-homogeneus-coord}
\phi=\left[w:v\right]
\end{equation}
or, even the polar coordinate $\phi=\arctan (w/v)$, determines the line $P\in\mathbb{P}\left(\mathcal{H}_{\gamma}\right)$.

For any $\gamma\in \mathcal{N}_U$ there exists $s_{\gamma}\in \mathbb{R}$ smoothly depending on $\gamma$ such that $\gamma\left(s_{\gamma}\right)\in C_{-}$\footnote{This can be shown if we notice that, since $U$ is globally hyperbolic, by \cite[Thm. 3.78]{Mi08}, there exists a diffeomorphism $h:C\times \mathbb{R}\rightarrow U$ such that $C_{\lambda}=h\left(C\times\{\lambda\}\right)$ is a smooth spacelike Cauchy surface in $U$. Without any lack of generality, we can assume that $C=C_0$ and $C_{-}=C_c$ for some $c\in\mathbb{R}$. Observe that any light ray can be parametrized by $\gamma(s)=\mathrm{exp}_{\gamma(0)}\left(s\cdot \gamma'(0)\right)$ then, if $p_2:C\times \mathbb{R}\rightarrow \mathbb{R}$ is the canonical projection, then the equation $p_2 \circ h^{-1}\left(\gamma(s)\right)=c$ can be written in coordinates by an equation such that $F(x,y,\theta,s)=c$. Then, the existence of $s_{\gamma}=s_{\gamma}(x,y,\theta)$ follows from the Theorem of implicit function. }.
Moreover, since for all $\gamma\in \mathcal{N}$, every curve $\sigma\left( [\gamma'(s)] \right)=\widetilde{\gamma}(s)$ is injective, then we can assume that 
\[
0<\phi\left(\sigma\left( [\gamma'(s_{\gamma})] \right)\right) < \phi\left(\sigma\left( [\gamma'(s)] \right)\right) < \phi\left(\oplus_{\gamma}\right) < 2\pi
\]
for all $\gamma\in \mathcal{N}_U$ and $s>s_{\gamma}$ in the domain of $\gamma$. Due to $\oplus$ is a smooth distribution then the function $\phi_{\oplus}=\phi\circ \oplus:\mathcal{N}_U\rightarrow \left[0,2\pi\right)$ is smooth, therefore the future boundary of $\widetilde{\mathcal{N}}_U=\mathbb{P}\left(\mathcal{H}_U\right)\cap \widetilde{\mathcal{N}}$ corresponding to $\oplus$ can be locally written by 
\[
\partial^{+} \widetilde{\mathcal{N}}_U =\left\{ \phi = \phi_{\oplus} \right\}   
\]

which can be seen as the graph of $\oplus$. Then $\oplus:\mathcal{N}\rightarrow \mathbb{P}\left(\mathcal{H}\right)$ is a diffeomorphism onto its image and the future boundary of $\widetilde{\mathcal{N}}$ is 
\[
\partial^{+}\widetilde{\mathcal{N}} = \oplus\left(\mathcal{N}\right) .
\]

The past boundary $\partial^{-} \widetilde{\mathcal{N}}$ can be shown to be smooth analogously. It is summarized in the following proposition.

\begin{proposition}\label{prop-boundaries-Ntilde}
Let $M$ be a $3$--dimensional conformal manifold under the hypotheses \ref{itm:hypotheses-general-1} and \ref{itm:hypotheses-general-2}. Then $\oplus:\mathcal{N}\rightarrow \partial^{+}\widetilde{\mathcal{N}}$ is a diffeomorphism and $\partial^{+}\widetilde{\mathcal{N}}=\oplus\left(\mathcal{N}\right)$ is a smooth manifold embedded in the boundary of $\widetilde{\mathcal{N}}$. 
\end{proposition}

If both $\oplus$ and $\ominus$ satisfy the condition \ref{itm:hypotheses-general-2} with $\oplus_{\gamma}\neq \ominus_{\gamma}$ for all $\gamma\in \mathcal{N}$, a trivial corollary can be stated.

\begin{corollary}
Let $M$ be a $3$--dimensional conformal manifold under the hypotheses \ref{itm:hypotheses-general-1} and \ref{itm:hypotheses-general-2} for both $\oplus$ and $\ominus$ and $\oplus_{\gamma}\neq \ominus_{\gamma}$ for all $\gamma\in \mathcal{N}$. Then the closure of $\widetilde{\mathcal{N}}$ is a smooth manifold with boundary $\partial \widetilde{\mathcal{N}}= \partial^{+} \widetilde{\mathcal{N}} \cup  \partial^{-} \widetilde{\mathcal{N}}$ where $\partial^{+} \widetilde{\mathcal{N}}=\oplus\left(\mathcal{N}\right)$ and  $\partial^{-} \widetilde{\mathcal{N}}=\ominus\left(\mathcal{N}\right)$ are embedded in $\mathbb{P}\left(\mathcal{H}\right)$ . 
\end{corollary}

It is known that a projectivity from $\mathbb{S}^{1}$ to $\mathbb{R}\cup \{\infty\}$ can be defined by choosing three points $s_1,s_2,s_3\in\mathbb{S}^{1}$ and assigning them the corresponding value in $\mathbb{R}\cup \{\infty\}$. 
Then we can define, in a smooth way, a projectivity in each fibre $\mathbb{P}\left(\mathcal{H}_{\gamma}\right)\simeq \mathbb{S}^{1}$. 
If $U\subset M$ is a globally hyperbolic, causally convex and normal open set and $C,C_{-}\subset U$ are smooth spacelike Cauchy surfaces such that $C_{-}\subset I^{-}(C)$, we can observe that $\oplus$ is a section of the bundle $\pi^{\mathbb{P}\left(\mathcal{H}\right)}_{\mathcal{N}}:\mathbb{P}\left(\mathcal{H}_U\right)\rightarrow \mathcal{N}_U$ and we can consider two disjoint sections $\sigma\left(\mathbb{PN}(C)\right)\in \mathbb{P}\left(\mathcal{H}_U\right)$ and $\sigma\left(\mathbb{PN}(C_{-})\right)\in \mathbb{P}\left(\mathcal{H}_U\right)$, then they permit to define a projectivity $\mathbf{t}$ in such a way that
\[
\mathbf{t}\left(\sigma\left(\mathbb{PN}\left(C_{-}\right)\right)\right)=-1, \qquad \mathbf{t}\left(\sigma\left(\mathbb{PN}\left(C\right)\right)\right)=0, \qquad \mathbf{t}\left(\partial^{+} \widetilde{\mathcal{N}}_U \right)=1 .
\]
We can see $\mathbf{t}$ as a new coordinate for the fibres of $\mathbb{P}\left(\mathcal{H}_{U}\right)$. It is easy to check that $\mathbf{t}$ is related to $\phi$ by
\begin{equation}\label{eq-proyectivo-homogeneo}
\mathbf{t}\left(P_{\gamma}\right)=\frac{\left(\phi^{\vee}-\phi^{\wedge}\right)\left(\phi\left(P_{\gamma}\right)-\phi^0\right)}{\left(2\phi^0 -\left(\phi^{\wedge}+\phi^{\vee}\right)\right)\phi\left(P_{\gamma}\right) +\left(2\phi^{\wedge}\phi^{\vee}-\phi^0\left(\phi^{\wedge}+\phi^{\vee}\right)\right)}  
\end{equation} 
for $P_{\gamma}\in \mathbb{P}\left(\mathcal{H}_{\gamma}\right)-\widetilde{\infty}_{\gamma}$ where $\widetilde{\infty}$ is the section corresponding to the infinite of $\mathbf{t}$ which verifies $\widetilde{\infty}\cap \overline{\widetilde{\mathcal{N}}_U} = \varnothing$ and where we have denoted $\phi^{\vee}=\phi\circ \sigma\left([\gamma'(s_{\gamma})]\right)$, $\phi^{\wedge}=\phi\circ\oplus (\gamma)$ and $\phi^0=\phi\circ\sigma\circ \xi^{-1}(\gamma)$ with $\xi:\mathbb{PN}(C)\rightarrow \mathcal{N}_U$ the diffeomorphism of diagram (\ref{diagram-charts}).

So, we can express the trivialization $ \mathbb{P}\left(\mathcal{H}_U\right)\simeq \mathcal{N}_{U} \times \mathbb{S}^{1}$ by a map 
\begin{equation}\label{eq-varepsilon}
\begin{tabular}{rrcl}
$\varepsilon:$ & $\mathcal{N}_{U}\times \mathbb{R}$ & $\rightarrow$ & $\mathbb{P}\left(\mathcal{H}_{U}\right)-\widetilde{\infty}$ \\
& $\left(\gamma,\mathbf{t}\right)$ & $\mapsto$ & $\widetilde{\gamma}\left(\mathbf{t}\right)$
\end{tabular}
\end{equation}
which is a diffeomorphism because its expression in coordinates is just $\left((x,y,\theta),\mathbf{t}\right)\mapsto (x,y,\theta,\mathbf{t})$ and $\mathbf{t}$ is said to be a \emph{projective parameter}.
Now it is clear that $\varepsilon\left(\mathcal{N}_{U}\times (-1,1)\right)\subset \widetilde{\mathcal{N}}_U $ with 
\[
\sigma\left(\mathbb{PN}\left(C_{-}\right)\right)=\varepsilon\left(\mathcal{N}_{U}\times \{-1\}\right), \qquad \sigma\left(\mathbb{PN}\left(C\right)\right)=\varepsilon\left(\mathcal{N}_{U}\times \{0\}\right), \qquad \partial^{+} \widetilde{\mathcal{N}}_U =\varepsilon\left(\mathcal{N}_{U}\times \{1\}\right)
\]
and, for any light ray $\gamma\in \mathcal{N}_U$, the curve $\sigma\left(\left[\gamma'(\mathbf{t})\right]\right)=\widetilde{\gamma}\left(\mathbf{t}\right)=\varepsilon\left(\gamma, \mathbf{t}\right)$ is defined for $\mathbf{t}\in\left(-1,1\right)$, but it is naturally extended to any $\mathbf{t}\in \mathbb{R}$ by $\widetilde{\gamma}\left(\mathbf{t}\right)= \varepsilon\left(\gamma,\mathbf{t}\right)$.
Moreover, the tangent vector to $\widetilde{\gamma}$ satisfies
\begin{equation}\label{extension-gamma-prima-tilde}
\widetilde{\gamma}'\left(\mathbf{t}\right)=\left( \frac{\partial}{\partial \mathbf{t}} \right)_{\widetilde{\gamma}\left(\mathbf{t}\right)}  .
\end{equation}

The map $\varepsilon$ defined by choosing Cauchy surfaces $C,C_{-}\subset U$ is called a \emph{local projective synchronization} and it will permit us to build the future $L$--boundary. The construction of the past boundary can be done by reversing the time. Notice that, if $\oplus_{\gamma}\neq \ominus_{\gamma}$ for all $\gamma\in\mathcal{N}$, the section $\sigma\left(\mathbb{PN}(C_{-})\right)\in \mathbb{P}\left(\mathcal{H}_U\right)$ can be replaced by $\partial^{-}\widetilde{\mathcal{N}}_U$ in order to define the projective parameter. In this case, we have that $\varepsilon:\mathcal{N}_U\times (-1,1)\rightarrow \widetilde{\mathcal{N}}_U$ is a diffeomorphism (see \cite[Prop. 5.2]{Ba18}).

\begin{remark}
If $M$ is globally hyperbolic conformal manifold with Cauchy surface $C\subset M$ diffeomorphic to $\mathbb{R}^{2}$, then $\mathcal{N}\simeq C \times \mathbb{S}^{1}$ and therefore the projective parameter $\mathbf{t}$ can be defined for $\mathbb{P}\left(\mathcal{H}\right)$ since $\mathbb{P}\left(\mathcal{H}\right)\simeq \mathcal{N}\times \mathbb{R}\cup\{\infty\}\simeq C\times \mathbb{S}^{1}\times \mathbb{R}\cup\{\infty\}$ where $C$ is a global Cauchy surface. 
In this case we will call \emph{universal projective parameter} to the parameter $\mathbf{t}\in \mathbb{R}$.
\end{remark}

\begin{remark}\label{remark-light-non-conj}
The bijectivity of $\varepsilon$ implies the injectivity of the curves $\widetilde{\gamma}$. This generalizes the condition of light non--conjugation of $M$, because  $\widetilde{\gamma}(\mathbf{t})=\varepsilon(\gamma,\mathbf{t})$ extends $\widetilde{\gamma}(\mathbf{t})=\sigma([\gamma'(\mathbf{t})])$ outside of $\widetilde{\mathcal{N}}$. 
\end{remark}

\subsection{Distributions in \texorpdfstring{$\widetilde{\mathcal{N}}$}{Ñ}}\label{sec:LBoundary-distrib}

In order to simplify, we will work only with the boundary $\partial^{+}\widetilde{\mathcal{N}}$ because all the construction for $\partial^{-}\widetilde{\mathcal{N}}$ can be done analogously.

Now, we will define two distributions in $\overline{\widetilde{\mathcal{N}}}=\widetilde{\mathcal{N}}\cup \partial^{+}\widetilde{\mathcal{N}}$: the former will be $\mathcal{D}^{\sim}$ in $\widetilde{\mathcal{N}}$ and the latter $\partial^{+}\mathcal{D}^{\sim}$ in $\partial^{+}\widetilde{\mathcal{N}}$. 
We will determine the conditions so that the union $\overline{\mathcal{D}^{\sim}}=\mathcal{D}^{\sim}\cup \partial^{+}\mathcal{D}^{\sim}$ is a smooth distribution in $\overline{\widetilde{\mathcal{N}}}$.   
The orbits of the distribution $\partial^{+}\mathcal{D}^{\sim}$ will corresponds to the points of the future boundary. 

First, let us call $\mathcal{P}$ to the regular distribution in $\mathbb{PN}$ defined by the fibres $\mathbb{PN}_q$ with $q\in M$ then, trivially, the map $\zeta:M\rightarrow \mathbb{PN}/\mathcal{P}$ defined by $\zeta\left(q\right)=\mathbb{PN}_q$ is a diffeomorphism.
Then, passing the distribution $\mathcal{P}$ to $\widetilde{\mathcal{N}}$ by the diffeomorphism $\sigma$, we obtain the distribution $\mathcal{D}^{\sim}$. 

Observe that the orbits of $\mathcal{D}^{\sim}$ are 
\[
\sigma\left(\mathbb{PN}_{q}\right)=\{ \sigma\left(\left[v\right]\right)\in\widetilde{\mathcal{N}}:\left[v\right]\in \mathbb{PN}_{q}  \}
\]
being $1$--dimensional compact submanifolds when $\mathrm{dim}~M=3$, then  $\mathcal{D}^{\sim}$ is a regular distribution and $\widetilde{\mathcal{N}}/\mathcal{D}^{\sim}$ is a differentiable manifold. 
Moreover the quotient map $\widetilde{\pi}:\widetilde{\mathcal{N}} \rightarrow \widetilde{\mathcal{N}}/\mathcal{D}^{\sim}$ is a submersion. 

Now, we have the following diagram
\begin{equation}\label{diagram-distrib-D}
\begin{tikzpicture}[every node/.style={midway}]
\matrix[column sep={6em,between origins},
        row sep={2em}] at (0,0)
{ ; &  \node(PN)   { $\mathbb{PN}$}  ; & \node(Nt)   { $\widetilde{\mathcal{N}}$} ; \\
 \node(M)   { $M$}; &  \node(PN-P)   { $\mathbb{PN}/ \mathcal{P}$}  ; & \node(Nt-D)   { $\widetilde{\mathcal{N}}/\mathcal{D}^{\sim}$} ;  \\} ; 
\draw[->] (PN) -- (Nt) node[anchor=south]  {$\sigma$};
\draw[->] (PN) -- (PN-P) node[anchor=east]  {$\kappa$};
\draw[->] (Nt)   -- (Nt-D) node[anchor=west] {$\widetilde{\pi}$};
\draw[->] (M)   -- (PN-P) node[anchor=north] {$\zeta$};
\draw[->] (PN-P)   -- (Nt-D) node[anchor=north] {$\widetilde{\sigma}$}; 
\end{tikzpicture}
\end{equation}
where $\widetilde{\sigma}$ is the map defined by $\widetilde{\sigma}\left(\mathbb{PN}_q\right)=\sigma\left(\mathbb{PN}_q\right)\in \widetilde{\mathcal{N}}/\mathcal{D}^{\sim}$, the maps $\kappa$ and $\widetilde{\pi}$ are the corresponding submersions, and $\sigma$, $\zeta$ and $\widetilde{\sigma}$ are diffeomorphisms. 
Therefore, we can observe that 
\begin{equation}\label{diffeo-S}
\widetilde{S}=\widetilde{\sigma}\circ \zeta: M \rightarrow \widetilde{\mathcal{N}}/\mathcal{D}^{\sim}
\end{equation}
is a diffeomorphism. 
There is a different proof of this result in \cite[Prop. 2.6]{Ba17}.

\vspace{5mm}

The second distribution to be defined is $\partial^{+} \mathcal{D}^{\sim}$. 
Recall that, by hypotheses, $\oplus$ is a $1$--dimensional regular distribution in $\mathcal{N}$, then $\oplus$ is integrable and the orbits of $\oplus$ define a regular foliation. 
By proposition \ref{prop-boundaries-Ntilde}, $\oplus: \mathcal{N}\rightarrow\partial^{+}\widetilde{\mathcal{N}}$ is a diffeomorphism, then the images of the orbits of $\oplus$ define a regular foliation in $\partial^{+}\widetilde{\mathcal{N}}$ corresponding to the distribution $\partial^{+} \mathcal{D}^{\sim}$.  
Of course, the distribution $\partial^{-}\mathcal{D}^{\sim}$ in $\partial^{-}\widetilde{\mathcal{N}}$ needed to build the past boundary can be defined by $\ominus$.                        

\vspace{5mm}

The next step is to describe the distribution $\overline{\mathcal{D}^{\sim}}=\mathcal{D}^{\sim}\cup \partial^{+} \mathcal{D}^{\sim}$ in order to study its smoothness. 
For this purpose, we will construct explicitly the orbits of $\mathcal{D}^{\sim}$.

Fix some auxiliary metric $\mathbf{g}\in \mathcal{C}$ and some globally hyperbolic, normal and causally convex open $U\subset M$ with $C\subset U$ a smooth spacelike Cauchy surface in $U$ as in remark \ref{rmk-basic-neighb}.
We denote by $\mathcal{N}_{U}\subset \mathcal{N}$ the open set of all light rays passing by $U$ and hence $\mathcal{N}_{U}$ is diffeomorphic to $C\times \mathbb{S}^1$ and then, we can consider all light rays $\gamma\in \mathcal{N}_{U}$ parametrized such that $\gamma'\left(0\right)\in \Omega\left(C\right)=\{ u\in \mathbb{N}^{+}\left(U\right): \mathbf{g}\left(u,T\right)=-1 \}$ for some future--directed timelike vector field $T\in\mathfrak{X}\left(M\right)$.
Since $M$ is strongly causal, by \cite[Prop. 6.4.7]{HE}, we can assume without any lack of generality, there is no imprisoned light ray in the closure $\overline{U}$ where $U$ is assumed to be relatively compact, so $U$ can be chosen such that $\gamma\cap U$ has only one connected component for all $\gamma \in\mathcal{N}_{U}$.

Let us consider an orthonormal frame $\{ E_1, E_2, E_3 \}$ on the local Cauchy surface $C$ such that $E_2, E_3$ are tangent to $C$ and $E_1$ is future--directed timelike related to the conformal structure $\left(M,\mathcal{C}\right)$.  
For a light ray $\gamma\in \mathcal{N}_{U}$ such that $\gamma\simeq\left(c,\theta\right)\in C\times \mathbb{S}^{1}$, we define $\left\{\mathbf{E}_i\left(\gamma,\mathbf{t}\right)\right\}_{i=1,2,3}$ the extension of the frame $\left\{E_i\left(c\right)\right\}_{i=1,2,3}$ by parallel transport to $\gamma\left(\mathbf{t}\right)$ along $\gamma$ related to the metric $\mathbf{g}$, where $\mathbf{t}$ is the projective parameter defined in section \ref{sec:LBoundary-Ntilde} by a local projective synchronization $\varepsilon$. 

The regular dependence on parameters of the solutions of initial value problems of ODEs \cite[Ch.~5]{Ha} assures the smooth dependence of the frames $\left\{\mathbf{E}_i\left(\gamma,\mathbf{t}\right)\right\}_{i=1,2,3}$ on $\left(\gamma,\mathbf{t}\right)$.

If $\theta\in \left[0,2\pi\right)\simeq \mathbb{S}^{1}$, then we define the lightlike vector
\[
V\left(\gamma,\mathbf{t},s\right) = \mathbf{E}_1\left(\gamma,\mathbf{t}\right) + \cos\left(\theta + s \right)\mathbf{E}_2\left(\gamma,\mathbf{t}\right)+\sin\left(\theta + s \right)\mathbf{E}_3\left(\gamma,\mathbf{t}\right)\in \mathbb{N}
\]
which depends smoothly on $\left(\gamma,\mathbf{t}\right)$ and defines the line 
\[
\Lambda\left(\gamma,\mathbf{t},s\right)=\left[V\left(\gamma,\mathbf{t},s\right)\right] = \mathrm{span}\{V\left(\gamma,\mathbf{t},s\right)\}\in \mathbb{PN}  .
\]

By means of the diffeomorphisms $\sigma$ and $\varepsilon$ of section \ref{sec:LBoundary-Ntilde}, and the canonical projections $p_1:\mathcal{N}\times \left(-1,1\right)\rightarrow \mathcal{N}$ and $p_2:\mathcal{N}\times \left(-1,1\right)\rightarrow \left(-1,1\right)$, we define the following differentiable maps
\begin{equation}\label{eq-def-distribution}
\begin{tabular}{l}
$\widetilde{X}\left(\gamma,\mathbf{t},s\right)=\sigma\left( \Lambda\left(\gamma,\mathbf{t},s\right)  \right) \in \widetilde{\mathcal{N}}$ \\
$X\left(\gamma,\mathbf{t},s\right)= p_1 \circ \varepsilon^{-1}\left(\widetilde{X}\left(\gamma,\mathbf{t},s\right) \right) \in \mathcal{N}$ \\
$\tau\left(\gamma,\mathbf{t},s\right)=p_2 \circ \varepsilon^{-1}\left(\widetilde{X}\left(\gamma,\mathbf{t},s\right) \right) \in \left(-1,1\right)$ .
\end{tabular}
\end{equation}
Observe that, for fixed $\left(\gamma,\mathbf{t}\right)\in\mathcal{N}_{U}\times\left(-1,1\right)$, the curve $X_{\left(\gamma,\mathbf{t}\right)}\left(s\right)=X\left(\gamma,\mathbf{t},s\right)$ is a parametrization of the $1$--dimensional submanifold $S\left(\gamma\left(\mathbf{t}\right)\right)\cap \mathcal{N}_U$. 
We will denote
\[
\gamma_{\left(\mathbf{t},s\right)}= X\left(\gamma,\mathbf{t},s\right)\in \mathcal{N}
\] 
when we will need to use a parameter for the light ray.
Also, the function $\tau_{\left(\gamma,\mathbf{t}\right)}\left(s\right)=\tau\left(\gamma,\mathbf{t},s\right)$ is the value of the parameter of $\gamma_{\left(\mathbf{t},s\right)}$ at the point $\gamma\left(\mathbf{t}\right)$ from $C$, then the identity
\begin{equation}\label{eq-gamma-tau}
\gamma_{\left(\mathbf{t},s\right)}\left( \tau\left(\gamma,\mathbf{t},s\right)\right)=\gamma\left(\mathbf{t}\right)
\end{equation}  
holds.
Moreover, $\widetilde{X}_{\left(\gamma,\mathbf{t}\right)}\left(s\right)=\widetilde{X}\left(\gamma,\mathbf{t},s\right)$ is the curve of lines of classes of Jacobi fields tangent to the light ray $X_{\left(\gamma,\mathbf{t}\right)}\left(s\right)$ at the point $\gamma\left(\mathbf{t}\right)$.

Recall that the map $\xi:\mathbb{PN}(C) \rightarrow \mathcal{N}_U$ of diagram (\ref{diagram-charts}) is a diffeomorphism, so the curve defined by 
\begin{equation}\label{eq-curva-c}
c_{\left(\gamma,\mathbf{t}\right)}\left(s\right)=\pi^{\mathbb{PN}(C)}_{C}\circ \xi^{-1}\left(\gamma_{(\mathbf{t},s)}\right)
\end{equation}
is smooth, but it can also be written by 
\[
c_{\left(\gamma,\mathbf{t}\right)}\left(s\right)=\gamma_{(\mathbf{t},s)}\cap C =\gamma_{\left(\mathbf{t},s\right)}\left(0\right)\in C.
\]

Now, we replace the parameter $s$ for the arc--length parameter. 
Fix some auxiliary metric $\mathbf{g}\in \mathcal{C}$ in $M$, since the Cauchy surface $C$ is differentiable and spacelike, the restriction $\left.\mathbf{g}\right|_{TC\times TC}$ is a Riemannian metric on $C$. 
If consider any $\langle J_{\left(\gamma,\mathbf{t},s\right)}\rangle \in \widetilde{X}\left(\gamma,\mathbf{t},s\right)$, since $M$ is assumed to be light non--conjugate, then for any $\mathbf{t}>0$ we have that any representative $J_{\left(\gamma,\mathbf{t},s\right)}\in \langle J_{\left(\gamma,\mathbf{t},s\right)}\rangle$ satisfies
\[
 J_{\left(\gamma,\mathbf{t},s\right)}\left(0\right)\neq 0~(\mathrm{mod} \gamma'_{\left(\mathbf{t},s\right)}\left(0\right)) 
\] 
and by lemma \ref{lemmaDC92}, 
\[
c'_{\left(\gamma,\mathbf{t}\right)}\left(s\right) \neq 0
\]
therefore we can parametrize the curves $c_{\left(\gamma,\mathbf{t}\right)}$ with the arc--length parameter $\mathbf{s}$ defined in $C$ by the restriction of $\mathbf{g}$.

\begin{lemma}\label{lem-curva-c-epsilon}
For every $\gamma_0\in \mathcal{N}$ and every $\delta\in(0,1)$ there exists $\epsilon>0$ and a neighbourhood $\mathcal{N}_U^{\epsilon}\subset \mathcal{N}$ of $\gamma_0$ such that the curves $c_{\left(\gamma,\mathbf{t}\right)}\left(\mathbf{s}\right)\in C$ can be defined for $\left(\gamma,\mathbf{t},\mathbf{s}\right)\in \mathcal{N}_{U}^{\epsilon}\times (1-\delta,1) \times (-\epsilon,\epsilon)$, where $\mathbf{s}$ is the arc--length parameter.
\end{lemma}

\begin{proof}
Fix some auxiliary $\mathbf{g}\in \mathcal{C}$. 
Let us consider a neighbourhood $\mathcal{N}_U$ of $\gamma_0$ such that $U\subset M$ is relatively compact and globally hyperbolic with Cauchy surface $C\in U$. 
For $p\in C$ and $r>0$, we will denote by $B_r(p)$ the ball in $C$ of radius $r$ centered at $p$ determined by the restriction of the metric $\mathbf{g}$ to $C$.
Now consider some $\epsilon > 0$ such that the ball $B_{2\epsilon}(\gamma_0\cap C)$ is fully contained in $C$. 
So, let us call $\mathcal{N}_U^{\epsilon}=\{ \gamma\in \mathcal{N}_U: \gamma\cap C\in B_{\epsilon}(\gamma_0 \cap C)\}\subset \mathcal{N}_U$. 

For any $(\gamma,\mathbf{t})\in \mathcal{N}_U^{\epsilon} \times \left(1-\delta,1\right)$ there exist $a_{(\gamma,\mathbf{t})},b_{(\gamma,\mathbf{t})}>0$ such that the maximal domain of definition of the segment of the curve $c_{(\gamma,\mathbf{t})}$ contained in $B_{2\epsilon}(\gamma_0\cap C)$ is the interval $I_{(\gamma,\mathbf{t})}=\left( -a_{(\gamma,\mathbf{t})},b_{(\gamma,\mathbf{t})} \right)\subset \mathbb{R}$.

Let us assume that $b_{(\gamma,\mathbf{t})}<\epsilon$. If we take a sequence $\{\mathbf{s}_n\}\subset I_{(\gamma,\mathbf{t})}$ such that $\mathbf{s}_n\mapsto b_{(\gamma,\mathbf{t})}$, since $\mathbb{PN}(C)$ is relatively compact, then the sequence 
\[
\left[ u_n \right] = \left[ \gamma'_{(\mathbf{t},\mathbf{s}_n)}(0) \right]\in\mathbb{PN}(C)
\]
has a convergent subsequence.  Assuming that this subsequence is $\{[u_n]\}$ itself, then 
\[
\left[ u_n \right]\mapsto \left[ u \right]\in\mathbb{PN}(C)  .
\] 
Moreover, for any $\left[ u_n \right]$ there exists $\left[ v_n \right]\in \mathbb{PN}_{\gamma(\mathbf{t})}$ such that $\gamma_{\left[ v_n \right]}=\gamma_{\left[ u_n \right]}\in\mathcal{N}$. 
Since $\mathbb{PN}_{\gamma(\mathbf{t})}$ is compact, then we can consider that $\left[ v_n \right]\mapsto \left[ v \right]\in \mathbb{PN}_{\gamma(\mathbf{t})}$. 
Now, because $\mathcal{N}$ is Hausdorff, then $\gamma_{\left[ v \right]}=\gamma_{\left[ u \right]}\in\mathcal{N}$ and hence we have $\gamma_{\left[ u \right]}\in S\left(\gamma(\mathbf{t})\right)$. 
This implies that $c_{(\gamma,\mathbf{t})}(\mathbf{s})$ exists for $\mathbf{s}=b_{(\gamma,\mathbf{t})}$ and the interval $I_{(\gamma,\mathbf{t})}$ is not maximal. 
This fact contradicts the assumption of $b_{(\gamma,\mathbf{t})}<\epsilon$. 

Therefore, $b_{(\gamma,\mathbf{t})}\geq\epsilon$ for all $(\gamma,\mathbf{t})\in \mathcal{N}_U^{\epsilon} \times \left(1-\delta,1\right)$. The case $a_{(\gamma,\mathbf{t})}\geq\epsilon$ can be shown analogously. 
Then $c(\gamma,\mathbf{t},\mathbf{s})=c_{(\gamma,\mathbf{t})}(\mathbf{s})$ can be defined by
\[
c:\mathcal{N}_U^{\epsilon} \times \left(1-\delta,1\right)\times  \left(-\epsilon,\epsilon\right)\longrightarrow B_{2\epsilon}(\gamma_0\cap C)\subset C  
\]
\qed
\end{proof}

The arc--length parameter $\mathbf{s}\in(-\epsilon,\epsilon)$ of the curves $c_{(\gamma,\mathbf{t})}$ can replace the previous variable $s$ in the maps $\widetilde{X}$, $X$ and $\tau$. 
We will denote again $\widetilde{X}\left(\gamma,\mathbf{t},\mathbf{s}\right)$, $X\left(\gamma,\mathbf{t},\mathbf{s}\right)$, $\tau\left(\gamma,\mathbf{t},\mathbf{s}\right)$ the corresponding maps of (\ref{eq-def-distribution}) with the new variable $\mathbf{s}$ and defined as
\begin{equation}\label{eq-Xtilde-X-tau}
\left\{
\begin{tabular}{l}
$\widetilde{X}:\mathcal{N}^{\epsilon}_{U}\times\left(1-\delta,1\right)\times \left(-\epsilon,\epsilon\right)\longrightarrow \widetilde{\mathcal{N}}_{U}$ \\
$X: \mathcal{N}^{\epsilon}_{U}\times\left(1-\delta,1\right)\times \left(-\epsilon,\epsilon\right)\longrightarrow  \mathcal{N}_{U}$ \\
$\tau: \mathcal{N}^{\epsilon}_{U}\times\left(1-\delta,1\right)\times \left(-\epsilon,\epsilon\right)\longrightarrow \left(-1,1\right)$
\end{tabular}
\right.
\end{equation}

Observe that, by construction, we have
\begin{equation} \label{eq-Xtilde-X}
\widetilde{X}\left(\gamma,\mathbf{t},\mathbf{s}\right)=\varepsilon\left(X\left(\gamma,\mathbf{t},\mathbf{s}\right), \tau\left(\gamma,\mathbf{t},\mathbf{s}\right) \right)
\end{equation}
and moreover
\begin{equation}\label{eq-X-tau-cero}
\left\{
\begin{tabular}{l}
$X\left(\gamma,\mathbf{t},0\right)=\gamma_{\left(\mathbf{t},0\right)}=\gamma$ \\
$\tau\left(\gamma,\mathbf{t},0\right)=\mathbf{t}$
\end{tabular}  
\right. 
\end{equation}
holds for all $\mathbf{t}\in\left(1-\delta,1\right)$.

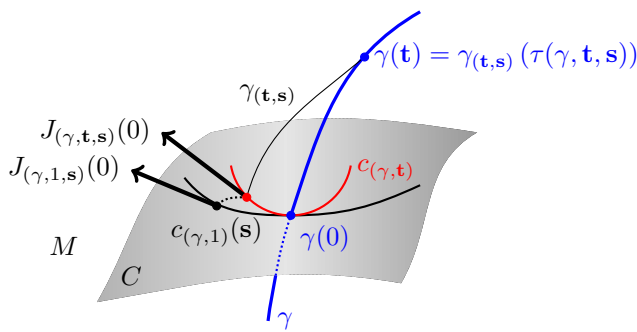
\begin{figure}[h]
  \centering
\begin{tikzpicture}[scale=1]
\fill[left color=gray!30!gray!60, right color=gray!15!gray!80, middle color=gray!20]
(-0.25,-0.25) to[out=60, in=-150] (1.25,2) to[out=12, in=171] (4.75,2) to[out=-150, in=60] (3.75,0) to[out=170, in=10] (-0.25,-0.25);
\draw[thick] (0.9,1.4) to[out=-52, in=180] (2.3,0.9);
\draw[thick] (2.3,0.9) to[out=0, in=-155] (4,1.3);
\draw[thick,color=red] (2.3,0.9) arc(-90:-10: 0.8 and 0.8);
\draw[thick,color=red] (2.3,0.9) arc(-90:-170: 0.8 and 0.8);
\draw[thick,densely dotted] (1.72,1.14) to[out=-180, in=50] (1.325,1.025);
\draw[very thick,color=blue] (2,-0.5) to[out=83, in=-100] (2.1,0.1);
\draw[thick,densely dotted,color=blue] (2.1,0.1) to[out=80, in=-108] (2.3,0.9);
\draw[very thick,color=blue] (2.3,0.9) to[out=72, in=-150] (4,3.6);
\draw (3.27,3) to[out=-140, in=75] (1.72,1.14);
\draw[->,ultra thick] (1.72,1.14) -- (0.6,2) node[anchor=east] {$J_{(\gamma,\mathbf{t},\mathbf{s})}(0)$};
\draw[->,ultra thick] (1.325,1.025) -- (0.2,1.5) node[anchor=east] {$J_{(\gamma,1,\mathbf{s})}(0)$};
\filldraw[color=blue] (2.3,0.9) node[anchor=north west] {$\gamma(0)$} circle (1.5pt);
\filldraw[color=blue] (3.27,3) node[anchor=west] {$\gamma(\mathbf{t})=\gamma_{(\mathbf{t},\mathbf{s})}\left(\tau(\gamma,\mathbf{t},\mathbf{s})\right)$} circle (1.5pt);
\filldraw[color=red] (1.72,1.14)  circle (1.5pt);
\filldraw (1.325,1.025) circle (1.5pt);
\draw (1.325,1) node[anchor=north] {$c_{(\gamma,1)}(\mathbf{s})$};
\draw (0.2,0.1) node {$C$};
\draw (-1,0.5) node[anchor=west] {$M$};
\draw[color=blue] (2,-0.5) node[anchor=west] {$\gamma$};
\draw (2,2.5) node {$\gamma_{(\mathbf{t},\mathbf{s})}$};
\draw[color=red] (3.05,1.5) node[anchor=west] {$c_{(\gamma,\mathbf{t})}$};
\end{tikzpicture}
  \caption{The curves $c_{\left(\gamma,\mathbf{t}\right)}$ in $C\subset M$.}
  \label{diapositiva4}
\end{figure}

It is important to remark that the curve $\widetilde{X}_{\left(\gamma,\mathbf{t}\right)}(\mathbf{s})=\widetilde{X}\left(\gamma,\mathbf{t},\mathbf{s}\right)$ is a parametrization of the orbit of the distribution $\mathcal{D}^{\sim}$ passing through $\widetilde{\gamma}\left(\mathbf{t}\right)$, that is the submanifold $\widetilde{S\left(\gamma\left(\mathbf{t}\right)\right)} \subset \widetilde{\mathcal{N}}$. 
This implies that $\mathcal{D}^{\sim}$ is generated by the tangent vectors $\frac{\partial \widetilde{X}}{\partial s}\left(\gamma,\mathbf{t},0\right) \in T\widetilde{\mathcal{N}}$ so, by (\ref{eq-Xtilde-X}), we have
\begin{equation} \label{eq-partial-Xtilde-X}
\frac{\partial \widetilde{X}}{\partial s}\left(\gamma,\mathbf{t},0\right)=\left(d\varepsilon\right)_{\left(\gamma,\mathbf{t}\right)}\left( \frac{\partial X}{\partial s}\left(\gamma,\mathbf{t},0\right) , \frac{\partial \tau}{\partial s}\left(\gamma,\mathbf{t},0\right) \right) 
\end{equation} 
for all $\left(\gamma,\mathbf{t}\right)\in \mathcal{N}^{\epsilon}_{U}\times\left(1-\delta,1\right)$.

\subsection{Smoothness of the distribution \texorpdfstring{$\overline{\mathcal{D}^{\sim}}=\partial^{+}\mathcal{D}^{\sim}\cup \mathcal{D}^{\sim}$}{\overline{D}=dD U D}: current status}\label{sec:Lboundary-Smoothness}

We have defined the distributions $\mathcal{D}^{\sim}$ and $\partial^+ \mathcal{D}^{\sim}$ separately in $\widetilde{\mathcal{N}}$ and in $\partial^+ \widetilde{\mathcal{N}}$ respectively, so we have that $\overline{\mathcal{D}^{\sim}}=\partial^{+}\mathcal{D}^{\sim}\cup \mathcal{D}^{\sim}$ is a distribution that, a priori, could even be non--continuous. In this section, we will study the conditions under which the distribution $\overline{\mathcal{D}^{\sim}}$ is smooth.  

The following theorem \ref{theorem-extension} clarifies a loophole in \cite[Thm. 7.1]{Ba18} and allows to establish equivalent conditions to the differentiability of $\overline{\mathcal{D}^{\sim}}$.

Recall that when we say that a smooth map $f:A\times (a,b) \rightarrow B$ can be smoothly (or differentiably) extended to $A\times (a,b]$ we mean that there exists $\epsilon>0$ and a smooth map $\overline{f}:A\times (a,b+\epsilon) \rightarrow B$ such that $\overline{f}=f$ in $A\times (a,b)$.


\begin{theorem}\label{theorem-extension}
Under the hypotheses \ref{itm:hypotheses-general-1} and \ref{itm:hypotheses-general-2} the following conditions are equivalent:
\begin{enumerate}
\item \label{itm:theo-cond-1} $\overline{\mathcal{D}^{\sim}}$ is smooth in $\overline{\widetilde{\mathcal{N}}_U}$.
\item \label{itm:theo-cond-2} $\widetilde{X}$ can be smoothly extended to $\mathbf{t}=1$.
\item \label{itm:theo-cond-3} $X$ can be smoothly extended to $\mathbf{t}=1$.
\item \label{itm:theo-cond-4} $\tau$ can be smoothly extended to $\mathbf{t}=1$.
\item \label{itm:theo-cond-5} $h(\gamma, \mathbf{t})=\frac{\partial \tau}{\partial \mathbf{s}}\left(\gamma,\mathbf{t},0\right)$ can be smoothly extended to $\mathbf{t}=1$.
\end{enumerate}
\end{theorem}

\begin{proof}

\ref{itm:theo-cond-1}) $\Rightarrow$ \ref{itm:theo-cond-2})

Assuming $\overline{\mathcal{D}^{\sim}}$ is smooth, since it is $1$--dimensional, then there exists a non--zero vector field $\widetilde{\Phi}\in \mathfrak{X}\left(\overline{\widetilde{\mathcal{N}}}\right)$ such that $\overline{\mathcal{D}^{\sim}}=\mathrm{span}\{\widetilde{\Phi}\}$. For any $\widetilde{\gamma}(1)\in\partial^{+}\widetilde{\mathcal{N}}$ there exists a flow box of $\widetilde{\Phi}$, that is, a smooth map $\widetilde{F}:\widetilde{\mathcal{U}}\times (-\overline{\epsilon},\overline{\epsilon}) \rightarrow \overline{\widetilde{\mathcal{N}}}$ such that $u\mapsto \widetilde{F}_{\widetilde{\gamma}(\mathbf{t})}(u)=\widetilde{F}\left(\widetilde{\gamma}(\mathbf{t}),u\right)$ is an integral curve of $\widetilde{\Phi}$. Making smaller the neighbourhood of definition of $\widetilde{F}$ and since $\varepsilon$ is a diffeomorphism, we can define the map
\[
\begin{tabular}{rccl}
$\overline{\widetilde{X}}:$ & $\mathcal{U} \times \left(1-\delta,1\right] \times \left(-\overline{\epsilon},\overline{\epsilon}\right)$ & $\longrightarrow$ & $\overline{\widetilde{\mathcal{N}}}$ \\
& $\left(\gamma,\mathbf{t},u\right)$ & $\longmapsto$ & $\overline{\widetilde{X}}\left(\gamma,\mathbf{t},u\right)=\widetilde{F}\left(\varepsilon(\gamma,\mathbf{t}),u\right)$ .
\end{tabular}
\]
with $\mathcal{U}\subset \mathcal{N}$. Then we have that 
\[
\left(\overline{X}\left(\gamma,\mathbf{t},u\right), \overline{\tau}\left(\gamma,\mathbf{t},u\right)\right)=\varepsilon^{-1}\left(\overline{\widetilde{X}}\left(\gamma,\mathbf{t},u\right)\right)\in \mathcal{U} \times \left(1-\delta,1\right]
\]
where $\overline{X}$ and $\overline{\tau}$ are differentiable maps for $\mathbf{t}\leq 1$. 

By lemma \ref{lem-curva-c-epsilon}, we can reparametrize the maps $\overline{\widetilde{X}}$, $\overline{X}$ and $\overline{\tau}$ by arc--length of the curves $z_{(\gamma,\mathbf{t})}(u)=\overline{X}\left(\gamma,\mathbf{t},u\right)\cap C$ and we obtain the maps (called in the same way)
\[
\left\{
\begin{tabular}{l}
$\overline{\widetilde{X}}\left(\gamma,\mathbf{t},\mathbf{s}\right)\in \overline{\widetilde{\mathcal{N}}}$ \\
$\overline{X}\left(\gamma,\mathbf{t},\mathbf{s}\right)\in \mathcal{N}_{U}^{\epsilon}$ \\
$\overline{\tau}\left(\gamma,\mathbf{t},\mathbf{s}\right)\in \left(1-\delta,1\right]$
\end{tabular}
\right.
\]
for $\mathbf{s}\in(-\epsilon,\epsilon)$. Since the curves $\mathbf{s}\mapsto \overline{\widetilde{X}}(\gamma,\mathbf{t},\mathbf{s})$ describe the orbits of $\overline{\mathcal{D}^{\sim}}$ and they are parametrized by the same arc--length parameter then
\[
\overline{\widetilde{X}}=\widetilde{X} \quad , \qquad \overline{X}=X \quad , \qquad \overline{\tau}=\tau
\]
for $\mathbf{t}<1$. Therefore $\overline{\widetilde{X}}$ is a smooth extension of $\widetilde{X}$.

\bigskip

\ref{itm:theo-cond-2}) $\Rightarrow$ \ref{itm:theo-cond-3})

Trivially, since $\varepsilon$ is a diffeomorphism, then if $\overline{\widetilde{X}}$ is a smooth extension of $\overline{\widetilde{X}}$ to $\mathbf{t} = 1$ we have
\[
\overline{\widetilde{X}}(\gamma,\mathbf{t},\mathbf{s})= \varepsilon\left(\overline{X}(\gamma,\mathbf{t},\mathbf{s}),\overline{\tau}\left(\gamma,\mathbf{t},\mathbf{s}\right) \right) \quad \Longleftrightarrow \quad \varepsilon^{-1}\circ \overline{\widetilde{X}}(\gamma,\mathbf{t},\mathbf{s})= \left(\overline{X}(\gamma,\mathbf{t},\mathbf{s}),\overline{\tau}\left(\gamma,\mathbf{t},\mathbf{s}\right) \right)
\]
for $\mathbf{t} \leq 1$ and therefore $\overline{X}$ (and also $\overline{\tau}$) is a smooth extension to $\mathbf{t}=1$.

\bigskip

\ref{itm:theo-cond-3}) $\Rightarrow$ \ref{itm:theo-cond-4})

Let us assume that $\overline{X}(\gamma,\mathbf{t},\mathbf{s})$ is a smooth extension of $X(\gamma,\mathbf{t},\mathbf{s})$ to $\mathbf{t}=1$, then we have that
\[
\frac{\partial \overline{X}}{\partial \mathbf{s}}\left(\gamma,\mathbf{t},\mathbf{s}\right)\in T_{\gamma_{(\mathbf{t},\mathbf{s})}}S\left( \gamma_{(\mathbf{t},\mathbf{s})}\left(\tau(\gamma,\mathbf{t},\mathbf{s})\right) \right) \subset \mathcal{H}_{\gamma_{\left(\mathbf{t},\mathbf{s}\right)}}  
\]
for $\mathbf{t}<1$.

Since the curves $c_{\left(\gamma,\mathbf{t}\right)}(\mathbf{s})=\overline{X}\left(\gamma,\mathbf{t},\mathbf{s}\right)\cap C$ are smooth and parametrized by arc--length as in equation (\ref{eq-curva-c}), by continuity we have that 
\[
\vert c'_{\left(\gamma,1\right)}(\mathbf{s})\vert=\lim_{\mathbf{t}\mapsto 1} \vert c'_{\left(\gamma,\mathbf{t}\right)}(\mathbf{s})\vert =1
\]
where $\vert \cdot \vert$ denotes the norm related to the restriction of the metric $\mathbf{g}\in\mathcal{C}$ to the local Cauchy surface $C\subset U$, hence we have that $c'_{\left(\gamma,1\right)}(\mathbf{s})\neq 0$. 
Then $\frac{\partial \overline{X}}{\partial \mathbf{s}}\left(\gamma,\mathbf{t},\mathbf{s}\right)\neq 0$ for all $\mathbf{t}\leq 1$ and so we obtain 
\[
\left[\frac{\partial \overline{X}}{\partial \mathbf{s}}\left(\gamma,\mathbf{t},\mathbf{s}\right)\right]=\mathrm{span}\left\{ \frac{\partial \overline{X}}{\partial \mathbf{s}}\left(\gamma,\mathbf{t},\mathbf{s}\right) \right\}\in \mathbb{P}\left(\mathcal{H}_{\gamma_{\left(\mathbf{t},\mathbf{s}\right)}}\right)  .
\]
Notice that, since $X\left(\gamma,\mathbf{t},\mathbf{s}\right)=\overline{X}\left(\gamma,\mathbf{t},\mathbf{s}\right)$ for $\mathbf{t}<1$ and $\widetilde{X}\left(\gamma,\mathbf{t},\mathbf{s}\right)=\mathrm{span}\left\{ \frac{\partial X}{\partial \mathbf{s}}\left(\gamma,\mathbf{t},\mathbf{s}\right) \right\}$ for $\mathbf{t}<1$, then we have $\widetilde{X}\left(\gamma,\mathbf{t},\mathbf{s}\right)= \mathrm{span}\left\{ \frac{\partial \overline{X}}{\partial \mathbf{s}}\left(\gamma,\mathbf{t},\mathbf{s}\right) \right\}$ for $\mathbf{t}<1$. 

By the diffeomorphism $\varepsilon$, we have that
\[
\left( \overline{X}(\gamma,\mathbf{t},\mathbf{s}), \overline{\tau}(\gamma,\mathbf{t},\mathbf{s}) \right)=\varepsilon^{-1}\left( \left[\frac{\partial \overline{X}}{\partial \mathbf{s}}\left(\gamma,\mathbf{t},\mathbf{s}\right)\right] \right)
\]
therefore $\overline{\tau}$ is a smooth extension of $\tau$.

\bigskip 

\ref{itm:theo-cond-4}) $\Rightarrow$ \ref{itm:theo-cond-5})

Let $\overline{\tau}$ be a smooth extension of $\tau$ to $\mathbf{t}=1$. Trivially, $\overline{h}(\gamma,\mathbf{t})=\frac{\partial \overline{\tau}}{\partial \mathbf{s}}\left(\gamma,\mathbf{t},0\right)$ is an smooth extension of $h(\gamma,\mathbf{t})=\frac{\partial \tau}{\partial \mathbf{s}}\left(\gamma,\mathbf{t},0\right)$ to $\mathbf{t}=1$.

\bigskip

\ref{itm:theo-cond-5}) $\Rightarrow$ \ref{itm:theo-cond-1})

First, let us show that $\overline{h}(\gamma,1)=\frac{\partial \overline{\tau}}{\partial \mathbf{s}}\left(\gamma,1,0\right)=0$. 
Observe that, by equation (\ref{eq-X-tau-cero}), $\tau(\gamma,\mathbf{t},0)=\mathbf{t}$ and since $\overline{\tau}$ is continuous, we have $\overline{\tau}(\gamma,1,0)=1$. 
Moreover, since $\overline{\tau}(\gamma,\mathbf{t},\mathbf{s})< 1$ for all $(\gamma,\mathbf{s})$ and $\mathbf{t}<1$, then $\overline{\tau}(\gamma,1,\mathbf{s})\leq 1$. So, since for every $\gamma$, the function $f(\mathbf{s})=\overline{\tau}(\gamma,1,\mathbf{s})$ reaches its maximum at $\mathbf{s}=0$, then, the smoothness of $\overline{\tau}$ brings $\frac{\partial \overline{\tau}}{\partial \mathbf{s}}\left(\gamma,1,0\right)=f'(0)=0$.

\medskip

Now, let us show that $\Phi\left(\gamma,\mathbf{t}\right)=\frac{\partial X}{\partial \mathbf{s}}\left(\gamma,\mathbf{t},0\right)$ can be smoothly extended to $\mathcal{N}^{\epsilon}_{U} \times \left(1-\delta,1\right]$. 
Notice that in (\ref{eq-def-distribution}) and (\ref{eq-Xtilde-X-tau}), we have defined $X(\gamma,\mathbf{t},\mathbf{s})=\gamma_{(\mathbf{t},\mathbf{s})}$ for $\mathbf{t}<1$ such that $c_{(\gamma,\mathbf{t})}(\mathbf{s})=\gamma_{(\mathbf{t},\mathbf{s})}(0)\in C$ is arc--length parametrized and, by lemma \ref{lemmaDC92}, the tangent vector $\langle J_{(\gamma,\mathbf{t})} \rangle=\frac{\partial X}{\partial \mathbf{s}}\left(\gamma,\mathbf{t},0\right) $ of the variation of light rays $X(\gamma,\mathbf{t},\mathbf{s})$ at $\mathbf{s}=0$ can be chosen such that $J_{(\gamma,\mathbf{t})}(0)=c'_{(\gamma,\mathbf{t})}(0)$.

We can consider the fibre bundle $\pi:\mathcal{H}\rightarrow \mathbb{P}(\mathcal{H})$ and any smooth non--zero local section $\overline{\omega}:\widetilde{\mathcal{U}}\subset\mathbb{P}\left(\mathcal{H}\right)\rightarrow \mathcal{H}$ in some neighbourhood $\widetilde{\mathcal{U}}$ of some $\widetilde{\gamma}_0\left(1\right)\in\partial^{+}\widetilde{\mathcal{N}}$.  
Without any lack of generality, we can assume that $\widetilde{\mathcal{U}}=\varepsilon\left(\mathcal{N}^{\epsilon}_{U}\times (1-\delta,1+\delta)\right)$.

Since $\pi\left(\overline{\omega}(\widetilde{\gamma}(\mathbf{t}))\right)=\widetilde{\gamma}(\mathbf{t})$, then we have 
\[
\overline{\omega}(\widetilde{\gamma}(\mathbf{t}))\in\widetilde{\gamma}(\mathbf{t})   .
\]

We denote by $\langle Z_{(\gamma,\mathbf{t})} \rangle$ the class of Jacobi fields along $\gamma\in\mathcal{N}^{\epsilon}_{U}$ defined by $\overline{\omega}(\widetilde{\gamma}(\mathbf{t}))$, that is 
\[
\overline{\omega}(\widetilde{\gamma}(\mathbf{t}))=\langle Z_{(\gamma,\mathbf{t})} \rangle\in \mathcal{H}_{\gamma}  .
\]

For any $\mathbf{t}\neq 0$, any representative $Z_{(\gamma,\mathbf{t})}\in \langle Z_{(\gamma,\mathbf{t})} \rangle$ verifies that 
\[
Z_{(\gamma,\mathbf{t})}(0) \neq 0~(\mathrm{mod}~\gamma'(0))
\]
because, in other case, we will have $\langle Z_{(\gamma,\mathbf{t})}\rangle \in \widetilde{\gamma}(0)\cap \widetilde{\gamma}(\mathbf{t})$. 
But this is not possible since $\widetilde{\gamma}$ is an injective curve for all $\mathbf{t}\in \mathbb{R}$, as noted in remark \ref{remark-light-non-conj}.

In fact, by locality of $\overline{\omega}$, we can assume that $Z_{(\gamma,\mathbf{t})}(0)$ is far from $0$ because $Z_{(\gamma_0,1)}(0)\neq 0~(\mathrm{mod}~\gamma_0'(0))$, it means that for all $(\gamma,\mathbf{t})\in\mathcal{N}^{\epsilon}_{U}\times (1-\delta,1+\delta)$ we have 
\[
\vert Z_{(\gamma,\mathbf{t})}(0) \vert^2=\mathbf{g}\left( Z_{(\gamma,\mathbf{t})}(0), Z_{(\gamma,\mathbf{t})}(0)\right)\geq\epsilon_0 > 0  
\]
for some $\epsilon_0 > 0$.
We can call $f(\gamma,\mathbf{t})=\vert Z_{(\gamma,\mathbf{t})}(0) \vert$ the smooth function which does not annihilate for all $(\gamma,\mathbf{t})\in\mathcal{N}^{\epsilon}_{U}\times (1-\delta,1+\delta)$. 

If we define  
\[
Y_{(\gamma,\mathbf{t})}= \frac{1}{f(\gamma,\mathbf{t})} \cdot Z_{(\gamma,\mathbf{t})}
\]
we have that 
\[
\omega\left(\widetilde{\gamma}(\mathbf{t})\right)=\frac{1}{f(\gamma,\mathbf{t})} \cdot \overline{\omega}\left(\widetilde{\gamma}(\mathbf{t})\right)=\langle Y_{(\gamma,\mathbf{t})} \rangle \in\widetilde{\gamma}(\mathbf{t})
\]
is another smooth non--zero local section defined for all  $(\gamma,\mathbf{t})\in\mathcal{N}^{\epsilon}_{U}\times (1-\delta,1+\delta)$ verifying  
\[
\mathbf{g}\left( Y_{(\gamma,\mathbf{t})}(0), Y_{(\gamma,\mathbf{t})}(0)\right)=1   .
\]
Take into account that, since $\widetilde{\gamma}(\mathbf{t})=T_{\gamma}S\left(\gamma\left(\mathbf{t}\right)\right)$ is $1$--dimensional, then the initial vectors $Y_{\left(\gamma,\mathbf{t}\right)}\left(0\right)$ determine the value of the section $\omega$.  
In fact, if $\overline{Y}$ is a another Jacobi field along $\gamma$ such that $\langle \overline{Y} \rangle \in \widetilde{\gamma}(\mathbf{t})$ with the same initial vector $\overline{Y}\left(0\right)=Y_{(\gamma,\mathbf{t})}(0)$, then $K=Y_{(\gamma,\mathbf{t})}-\overline{Y}$ is also a Jacobi field along $\gamma$ verifying $K\left(0\right)=0~\left(\mathrm{mod}~\gamma'(0)\right)$ and so, $K\in \widetilde{\gamma}\left(0\right)\cap \widetilde{\gamma}\left(\mathbf{t}\right)$.
Since every curve $\widetilde{\gamma}$ is injective, then $K=0$ and therefore $\langle \overline{Y} \rangle=\langle Y_{(\gamma,\mathbf{t})} \rangle$.

Recall that the curves $c_{(\gamma,\mathbf{t})}$ have been parametrized by arc--length, so 
\[
\mathbf{g}\left( c'_{(\gamma,\mathbf{t})}(0), c'_{(\gamma,\mathbf{t})}(0)\right)=1
\]
and we can take representatives $Y_{(\gamma,\mathbf{t})}$ such that $Y_{(\gamma,\mathbf{t})}(0)=c'_{(\gamma,\mathbf{t})}(0)$, therefore by construction of $X(\gamma,\mathbf{t},\mathbf{s})$ and $\omega$, then    
\begin{equation}
\omega\left(\widetilde{\gamma}\left(\mathbf{t}\right)\right)=\langle Y_{\left(\gamma,\mathbf{t}\right)}\rangle= \langle J_{\left(\gamma,\mathbf{t}\right)}\rangle =\frac{\partial X}{\partial \mathbf{s}}\left(\gamma,\mathbf{t},0\right) 
\end{equation}
holds for $\mathbf{t}<1$.

Therefore $\overline{\Phi}\left(\gamma,\mathbf{t}\right)=\omega\left(\varepsilon(\gamma,\mathbf{t})\right)=\omega\left(\widetilde{\gamma}\left(\mathbf{t}\right)\right)$ is a smooth extension of $\Phi\left(\gamma,\mathbf{t}\right)=\frac{\partial X}{\partial \mathbf{s}}\left(\gamma,\mathbf{t},0\right)$ defined in $\mathcal{N}^{\epsilon}_{U} \times \left(1-\delta, 1+\delta\right)$.

Now, the expression of equation (\ref{eq-partial-Xtilde-X}) can be extended as 
\begin{equation*}
\overline{\widetilde{\Phi}}\left(\gamma,\mathbf{t}\right)=\left(d\varepsilon\right)_{\left(\gamma,\mathbf{t}\right)}\left( \overline{\Phi}\left(\gamma,\mathbf{t}\right) , \overline{h}\left(\gamma,\mathbf{t}\right) \right) 
\end{equation*} 
for all $\left(\gamma,\mathbf{t}\right)\in \mathcal{N}^{\epsilon}_{U}\times\left(1-\delta,1\right]$ such that 
\begin{equation*} 
\frac{\partial \widetilde{X}}{\partial s}\left(\gamma,\mathbf{t},0\right)=\overline{\widetilde{\Phi}}\left(\gamma,\mathbf{t}\right)
\end{equation*} 
because $\overline{\Phi}\left(\gamma,\mathbf{t}\right)=\frac{\partial X}{\partial s}\left(\gamma,\mathbf{t},0\right)$ and $\overline{h}\left(\gamma,\mathbf{t}\right)=\frac{\partial \tau}{\partial s}\left(\gamma,\mathbf{t},0\right)$ for $\mathbf{t}<1$.

Notice that a curve $\Gamma\left(s\right)\in \mathcal{N}$ is an integral curve of $\oplus:\mathcal{N}\rightarrow \mathbb{P}\left(\mathcal{H}\right)$ if $\Gamma'\left(s\right)\in \oplus_{\Gamma\left(s\right)}$.
So, the curve $\widetilde{\Gamma}\left(s\right)=\varepsilon\left(\Gamma\left(s\right),1\right)$ is a leaf of the distribution $\partial^{+}\mathcal{D}^{\sim}$ if $\Gamma'\left(s\right)\in \oplus_{\Gamma\left(s\right)}$, that is  
\begin{equation}\label{eq-curva-integral}
\widetilde{\Gamma}'\left(s\right)=\left(d\varepsilon\right)_{\left(\Gamma\left(s\right),1\right)}\left( \Gamma'\left(s\right) , 0 \right)\in \partial^{+}\mathcal{D}^{\sim} \Longleftrightarrow  \Gamma'\left(s\right)\in \oplus_{\Gamma\left(s\right)}  .
\end{equation}

Now, since $\frac{\partial \overline{X}}{\partial s}\left(\gamma,\mathbf{t},0\right)\in \widetilde{\gamma}(\mathbf{t})$ for $1-\delta <\mathbf{t}<1$, by continuity we have
\[
\frac{\partial \overline{X}}{\partial s}\left(\gamma,1,0\right)\in \widetilde{\gamma}(1)=\oplus_{\gamma}
\]
and also
\begin{align*}
\frac{\partial \overline{\widetilde{X}}}{\partial s}\left(\gamma,1,0\right)& =\left(d\varepsilon\right)_{\left(\gamma,1\right)}\left( \frac{\partial \overline{X}}{\partial s}\left(\gamma,1,0\right) , \frac{\partial \overline{\tau}}{\partial s}\left(\gamma,1,0\right) \right) = \\
& =\left(d\varepsilon\right)_{\left(\gamma,1\right)}\left( \frac{\partial \overline{X}}{\partial s}\left(\gamma,1,0\right) , 0 \right)   .
\end{align*}
Then, by (\ref{eq-curva-integral}), $\frac{\partial \overline{\widetilde{X}}}{\partial s}\left(\gamma,1,0\right)\in \partial^{+}\mathcal{D}^{\sim}$. 

Therefore, $\overline{\widetilde{\Phi}}\left(\gamma,\mathbf{t}\right)=\frac{\partial \overline{\widetilde{X}}}{\partial s}\left(\gamma,\mathbf{t},0\right)$ is a smooth vector field defining $\overline{\mathcal{D}^{\sim}}=\mathcal{D}^{\sim} \cup \partial^{+}\mathcal{D}^{\sim}$ for all $1-\delta <\mathbf{t}\leq 1$, so $\overline{\mathcal{D}^{\sim}}$ is a differentiable distribution (see figure \ref{diapositiva5}).
\qed
\end{proof}

\begin{figure}[h]
  \centering
\begin{tikzpicture}[scale=0.9]
\fill[left color=gray!30!gray!60, right color=gray!15!gray!80, middle color=gray!20] 
(0.25,0.25) to[out=60, in=-150] (1.25,1.5) to[out=12, in=171] (4.75,1.5) to[out=-150, in=60] (3.75,0) to[out=170, in=10] (0.25,0.25);
\fill[left color=gray!30!gray!60, right color=gray!15!gray!80, middle color=gray!20] 
(0.25,4.25) to[out=60, in=-150] (1.25,5.5) to[out=12, in=171] (4.75,5.5) to[out=-150, in=60] (3.75,4) to[out=170, in=10] (0.25,4.25);
\draw[color=red] (3,1) arc(-35:60: 0.8 and 0.4);
\draw[color=red] (3,1) arc(-35:-100: 0.8 and 0.4);
\draw[color=red] (3,3) arc(-80:10: 0.8 and 0.4);
\draw[color=red] (3,3) arc(-80:-170: 0.8 and 0.4);
\draw[color=red] (3,5.1) arc(-150:-240: 0.8 and 0.4);
\draw[color=red] (3,5.1) arc(-150:-60: 0.8 and 0.4);
\draw[thick,color=blue] (3,-0.3) node[anchor=north] {$\widetilde{\gamma}$} -- (3,0.12);
\draw[densely dotted,color=blue] (3,0.12) -- (3,1);
\draw[thick,color=blue] (3,1) -- (3,4.12);
\draw[densely dotted,color=blue] (3,4.12) -- (3,5.1);
\draw[->,thick] (3,1) -- (3.8,1.8) node[anchor=south west] {$\overline{\widetilde{\Phi}}$};
\draw[->,thick] (3,3) -- (4,3.1) node[anchor=west] {$\overline{\widetilde{\Phi}}$};
\draw[->,thick] (3,5.1) -- (3.6,4.7) node[anchor=north east] {$\overline{\widetilde{\Phi}}$};
\filldraw[color=blue] (3,1) node[anchor=north west] {$\widetilde{\gamma}(0)$} circle (1pt);
\filldraw[color=blue] (3,5.1) node[anchor=north east] {$\oplus_{\gamma}=\widetilde{\gamma}(1)$} circle (1pt);
\filldraw[color=blue] (3,3) node[anchor=north east] {$\widetilde{\gamma}(\mathbf{t})$} circle (1pt);
\draw[color=red] (1.7,0.9) node {$\widetilde{X}_{(\gamma,0)}$};
\draw[color=red] (1.6,3.3) node {$\widetilde{X}_{(\gamma,\mathbf{t})}$};
\draw[color=red] (3.5,6) node {$\widetilde{X}^{+}_{\gamma}$};
\draw (3.75,0) node[anchor=west] {$\varepsilon\left(\mathcal{N}\times \{0\}\right)$};
\draw (4.75,5.5) node[anchor=west] {$\varepsilon\left(\mathcal{N}\times \{1\}\right)=\partial^{+}\widetilde{\mathcal{N}}$};
\draw (0,2.8) node[anchor=east] {$\widetilde{\mathcal{N}}\subset \mathbb{P}\left(\mathcal{H}\right)$};
\end{tikzpicture}
  \caption{The vector field $\overline{\widetilde{\Phi}}$ in $\overline{\widetilde{\mathcal{N}}}$ defining the distribution $\overline{\mathcal{D}^{\sim}}$.}
  \label{diapositiva5}
\end{figure}
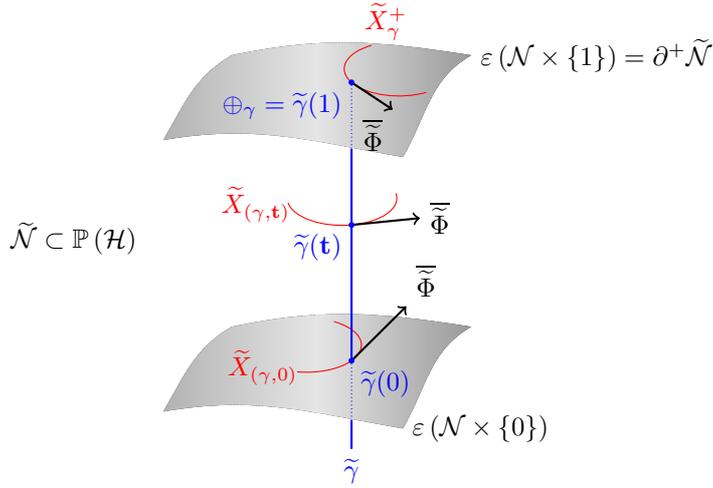

As an immediate consequence we have the following result.

\begin{corollary}
The distribution $\overline{\mathcal{D}^{\sim}}$ in $\overline{\widetilde{\mathcal{N}}}$ is smooth if for each light ray $\gamma_0\in\mathcal{N}$ there exist $U\subset M$ such that $\gamma_0\in \mathcal{N}_U=\{ \gamma\in \mathcal{N} : \gamma\cap U \neq \varnothing \}$ satisfying the hypotheses of theorem \ref{theorem-extension}.  
\end{corollary}

At this point, it has not yet been possible to prove whether the hypotheses \ref{itm:hypotheses-general-1} and \ref{itm:hypotheses-general-2} are sufficient conditions to establish the existence of the extension of the function $h(\gamma,\mathbf{t})=\frac{\partial \tau}{\partial \mathbf{s}}\left(\gamma,\mathbf{t},0\right)$ to $\mathbf{t}=1$ or whether there is some counterexample of this. So, this is an open question that should be studied in future work.

\subsection{The canonical future extension of \texorpdfstring{$M$}{M}}\label{sec:Lboundary-canonical}

Let us assume that $\overline{\mathcal{D}^{\sim}}$ is a smooth distribution. 
Notice that the leaves of $\mathcal{D}^{\sim}$ are always compact but the ones of $\partial^{+}\mathcal{D}^{\sim}$ could be not, making of $\overline{\mathcal{D}^{\sim}}$ a non--regular distribution. 
When, moreover, $\overline{\mathcal{D}^{\sim}}$ becomes a regular distribution then the quotient    
\[
\overline{\widetilde{\mathcal{N}}} / \overline{\mathcal{D}^{\sim}} =  \widetilde{\mathcal{N}} / \mathcal{D}^{\sim}  \cup \partial^{+}\widetilde{\mathcal{N}} / \partial^{+}\mathcal{D}^{\sim} 
\]
is a differentiable manifold.  
In this case, by the diffeomorphism $\widetilde{S}:M\rightarrow \widetilde{\mathcal{N}} / \mathcal{D}^{\sim}$ of equation (\ref{diffeo-S}), and because $\partial^{+}\widetilde{\mathcal{N}}$ is the boundary of $\overline{\widetilde{\mathcal{N}}}$, then $ \partial^{+}\widetilde{\mathcal{N}} / \partial^{+}\mathcal{D}^{\sim}$ is the boundary of $\overline{\widetilde{\mathcal{N}}} / \overline{\mathcal{D}^{\sim}}$.
Then we can extend $\widetilde{S}$ by
\[
\widetilde{S}:\overline{M}\rightarrow \overline{\widetilde{\mathcal{N}}} / \overline{\mathcal{D}^{\sim}}
\]
where $\overline{M} = M \cup \partial^{+} M $ with 
\[
\partial^{+} M = \partial^{+}\widetilde{\mathcal{N}} / \partial^{+}\mathcal{D}^{\sim}  
\]
in such a way that $\left.\widetilde{S}\right|_{\partial^{+} M}$ is the identity map. Then there exists a differentiable structure in $\overline{M}$, compatible with the one in $M$, such that the extension $\widetilde{S}$ is a diffeomorphism inducing in $\partial^{+} M$ a differentiable structure.

This is summarized in the following corollary.

\begin{corollary}\label{corollary-main-theorem}
If $\overline{\mathcal{D}^{\sim}}$ is a smooth and regular distribution then the quotient $\overline{\widetilde{\mathcal{N}}} / \overline{\mathcal{D}^{\sim}} $ 
is a differentiable manifold with boundary  $ \partial^{+} M= \partial^{+}\widetilde{\mathcal{N}} / \partial^{+}\mathcal{D}^{\sim}$. Moreover, there exists an extension 
\[
\widetilde{S}:\overline{M}=M \cup \partial M\rightarrow \overline{\widetilde{\mathcal{N}}} / \overline{\mathcal{D}^{\sim}}
\]
of the diffeomorphism $\widetilde{S}:M\rightarrow \widetilde{\mathcal{N}} / \mathcal{D}^{\sim}$ of equation (\ref{diffeo-S}) such that $\left.\widetilde{S}\right|_{\partial^{+} M}$ is the identity map and it induces in $\partial^{+} M$ a differentiable structure such that the extended map $\widetilde{S}$ is a diffeomorphism.
\end{corollary}

\begin{definition}
The extension $\overline{M}$ constructed in the corollary \ref{corollary-main-theorem} is called the \emph{canonical future extension of} $\left(M,\mathcal{C}\right)$ and the boundary $\partial^{+} M$ is the \emph{future $L$--boundary}.
\end{definition}

Again, we only focus on future boundary, but the construction of the \emph{canonical past extension} $M\cup \partial^{-}M$ and the past boundary $\partial^{-} M$ can be done using the distribution $\ominus$ analogously.

In virtue of maps (\ref{eq-varepsilon}) and (\ref{diffeo-S}), we obtain a double fibration extended to the boundaries which is analogous to the double fibration (\ref{double-fibration-1}) of twistor theory. So, we have  
\begin{equation}\label{double-fibration-3}
\begin{tikzpicture}[every node/.style={midway}]
\matrix[column sep={6em,between origins},
        row sep={2em}] at (0,0)
{ ; &  \node(Ntilde)   { $\overline{\widetilde{\mathcal{N}}}$}  ; &  ; \\
 \node(N)   { $\mathcal{N}$}; &    ; & \node(M)   { $\overline{M}$} ;  \\} ; 
\draw[->] (Ntilde) -- (N) node[anchor=south east]  {$\pi^{\mathbb{P}\left(\mathcal{H}\right)}_{\mathcal{N}}$};
\draw[->] (Ntilde) -- (M) node[anchor=south west]  {$\rho$};
\end{tikzpicture}
\end{equation}
where $\pi^{\mathbb{P}\left(\mathcal{H}\right)}_{\mathcal{N}}$ is the canonical projection that can be expressed by $\pi^{\mathbb{P}\left(\mathcal{H}\right)}_{\mathcal{N}}=p_1\circ \varepsilon^{-1}$ as in equation (\ref{eq-def-distribution}) and $\rho$ is the extension of the submersion given by $\rho=\widetilde{S}^{-1}\circ \widetilde{\pi}$ where $\widetilde{\pi}$ is the quotient map given in diagram (\ref{diagram-distrib-D}).

\begin{remark}
The regularity of $\overline{\mathcal{D}^{\sim}}$ is achieved, for example when the leaves of $\partial^{+}\mathcal{D}^{\sim}$ are compact, because then, all the leaves of $\overline{\mathcal{D}^{\sim}}$ are compact. This holds when $M$ is globally hyperbolic with compact Cauchy surface $C$. In this case we have that $\mathcal{N}\simeq \mathbb{PN}(C)$ is compact and since, by hypothesis, $\oplus$ is a regular distribution, then the leaves of $\partial^{+}\mathcal{D}^{\sim}$ must be compact, therefore $\overline{\mathcal{D}^{\sim}}$ is regular and the canonical extension  $\overline{M}$ is a differentiable manifold. This is the case of de Sitter $d\mathcal{S}^{m}$ spacetimes and Robertson-Walker models without initial or final singularity.  

If the leaves of $\partial^{+}\mathcal{D}^{\sim}$ are not compact, the canonical extension can still exist, as for Minkowski $\mathbb{M}^m$ spacetimes. 
\end{remark}

In example \ref{example-M3-block}, the canonical extension of Minkowski is built for $\mathrm{dim}(M)=3$. 
For higher dimension it can be computed integrating the distribution $\oplus$ in $\mathcal{N}$ as done in \cite[Sec. IV.B]{Ba17}. 
In this reference, the $L$--boundary of $\mathbb{M}^3$ is constructed by restriction of the one of $\mathbb{M}^4$.  
This can be done because there exists an embedding $\mathbb{M}^3\hookrightarrow \mathbb{M}^4$ such that the maximal null geodesic in $\mathbb{M}^3$ are maximal null geodesics in $\mathbb{M}^4$. 
Moreover, any Cauchy surface $\overline{C}\subset \mathbb{M}^4$ defines a Cauchy surface $C=\overline{C}\cap \mathbb{M}^3$ in $\mathbb{M}^3$ and, since we have an embedding $T\mathbb{M}^3\hookrightarrow T\mathbb{M}^4$, then $\mathbb{PN}(C)\hookrightarrow \mathbb{PN}(\overline{C})$ is an embedding. Therefore, $\mathcal{N}_{\mathbb{M}^3}\hookrightarrow \mathcal{N}_{\mathbb{M}^4}$ is an embedding and this implies that $T\mathcal{N}_{\mathbb{M}^3}\hookrightarrow T\mathcal{N}_{\mathbb{M}^4}$ is another embedding. Then, denoting $\overline{\mathcal{H}}$ and $\overline{\oplus}$ the contact structure and the future limit distribution related to $\mathbb{M}^4$ and $\mathcal{H}$ and $\oplus$ the same geometric objects related to $\mathbb{M}^3$, then it is easy to see that 
\[
\mathcal{H}_{\gamma}= \overline{\mathcal{H}}\cap T_{\gamma}\mathcal{N}_{\mathbb{M}^3}, \qquad \oplus_{\gamma}= \overline{\oplus}\cap T_{\gamma}\mathcal{N}_{\mathbb{M}^3} \qquad \text{ for } \gamma\in\mathcal{N}_{\mathbb{M}^3}\subset \mathcal{N}_{\mathbb{M}^4} .
\]

We will use this procedure in example \ref{example-de-Sitter} for the $3$--dimensional de Sitter spacetime which is embedded in $\mathbb{M}^4$. 

\begin{example}\label{example-M3-block}
Consider the $3$--dimensional Minkowski spacetime block 
\[
\mathbb{M}^3_{(a,b)} = \left\{ (t,x,y)\in\mathbb{M}^3: a<t<b \right\}   .
\]
To simplify, we will only compute the future $L$--boundary for $-\infty <a<0$ to avoid $\oplus = \ominus$ when $b=\infty$ and to keep $C\equiv\{t=0\}$ as Cauchy surface for initial values of the null geodesics. 
First, observe that, by (\ref{eq-Jacobi-3M-initial}), the homogeneous coordinate $\phi$ in eq. (\ref{eq-homogeneus-coord}) is 
\[
\phi=[-s:1]=\left[1:\frac{1}{-s}\right]
\]
when $s\neq 0$ then, by equation (\ref{eq-proyectivo-homogeneo}), we can consider the projective parameter $\mathbf{t}(s)=\frac{(b-a)s}{(b+a)s-2ab}$ and whence $s(\mathbf{t})=\frac{2ab\mathbf{t}}{(b+a)\mathbf{t}-(b-a)}$.
In example \ref{example-contact-3M}, we have seen that $\mu(\theta,s,\tau)\in C$ whenever $\tau=-s$ and this value does not depend on $\theta\in\left[0,2\pi\right)$. Moreover, in virtue of remark \ref{remark-contact-skies} and examples \ref{ex-contact-4M} and \ref{example-contact-3M}, the tangent space of the sky $S\left(\gamma\left(\mathbf{t}\right)\right)\in\Sigma$ at $\gamma$ can be written as
\begin{equation}\label{tangent-Mink-sky}
T_{\gamma}S\left(\gamma\left(\mathbf{t}\right)\right) =  \mathrm{span}\left\{   \textstyle{   s(\mathbf{t})\left( \sin\theta_0 \left( \frac{\partial}{\partial x} \right)_{\gamma} - \cos\theta_0 \left( \frac{\partial}{\partial y} \right)_{\gamma} \right) + \left( \frac{\partial}{\partial \theta} \right)_{\gamma}   }    \right\}  .
\end{equation}
Therefore, the orbit of the distribution $\overline{\mathcal{D}^{\sim}}$ in $\overline{\widetilde{\mathcal{N}}}$ passing through $\widetilde{\gamma}(\mathbf{t})\simeq(x_0,y_0,\theta_0,\mathbf{t}_0)$ corresponds to the integral curve $c\left(r\right)=\left(x\left(r\right),y\left(r\right),\theta\left(r\right),\mathbf{t}\left(r\right)\right)$ of the vector field 
\[
\widetilde{\Phi}=s(\mathbf{t})\left(\sin\theta \frac{\partial}{\partial x} - \cos\theta  \frac{\partial}{\partial y}\right)+\frac{\partial}{\partial \theta} \in\mathfrak{X}\left(\widetilde{\mathcal{N}}\right)   .
\] 
So, after integration, it can be written by
\begin{align*}
 c\left(r\right)&=\textstyle{\left(x_0 + s(\mathbf{t}_0) \left[\cos \theta_0 -\cos \left(\theta_0+r\right)\right]  \, ,   \right. } \\
& \textstyle{ \left.  y_0 + s(\mathbf{t}_0) \left[\sin \theta_0 -\sin \left(\theta_0+r\right)\right] \, , \, \theta_0+r \, , \, \mathbf{t}_0\right) }  
\end{align*}
and verifies 
\begin{equation}\label{eq-orbit-boundary-3m}
\left\{
\begin{array}{l}
\left(x-x_0 - s(\mathbf{t}_0) \cos \theta_0\right)^2+ \left(y-y_0 - s(\mathbf{t}_0) \sin \theta_0\right)^2=s^2(\mathbf{t}_0)\\
\theta= \theta_0+r \\
\mathbf{t}=\mathbf{t}_0
\end{array}
\right. .
\end{equation}

Since $\lim_{\mathbf{t_0}\mapsto 1^{-}}s(\mathbf{t}_0)=b$ then, whenever $b<\infty$, the light rays in $\mathbb{M}^{3}_{(a,b)}$ of the orbit of $\partial^{+}\mathcal{D}^{\sim}$ passing through $\widetilde{\gamma}(1)\simeq (x_0,y_0,\theta_0, 1)$ corresponds to the sky in $\mathbb{M}^{3}$ of the point $p=\left(b, x_0 + b \cos \theta_0 , y_0 + b \sin \theta_0\right)$, see figure \ref{fig-L-boundary-M3-a}. 
Then the future boundary $\partial^{+} M$ can be identified with the topological boundary of $\mathbb{M}^{3}_{(a,b)}$ as a set in $\mathbb{M}^{3}$, that is 
\[
\partial^{+} M = \{ (t,x,y)\in \mathbb{M}^3 : t=b\}  .
\]

In case of $b=\infty$, we can develop the squares of the first equation in (\ref{eq-orbit-boundary-3m}) and divide by $s(\mathbf{t}_0)$ to obtain 
\begin{equation*}
\left\{
\begin{array}{l}
\frac{1}{s(\mathbf{t}_0)}\left(x-x_0\right)^2 - 2\left(x-x_0\right) \cos \theta_0 + \frac{1}{s(\mathbf{t}_0)}\left(y-y_0\right)^2 - 2\left(y-y_0\right) \sin \theta_0=0\\
\theta= \theta_0+r \\
\mathbf{t}=\mathbf{t}_0
\end{array}
\right. 
\end{equation*}
then, taking the limit $\lim_{\mathbf{t_0}\mapsto 1^{-}}s(\mathbf{t}_0)=+\infty$, we have that the orbit of $\partial^{+}\mathcal{D}^{\sim}$ passing through the point with coordinates $(x_0,y_0,\theta_0,1)\in \overline{\widetilde{\mathcal{N}}}$ verifies 
\begin{equation*}
\left\{
\begin{array}{l}
\left(x-x_0\right) \cos \theta_0 + \left(y-y_0\right) \sin \theta_0=0\\
\theta= \theta_0 \\
\mathbf{t}=1
\end{array}
\right. .
\end{equation*}
Therefore, as it is illustrated in figure \ref{fig-L-boundary-M3-b}, the orbit consists of the light rays with tangent vector $v=\left(1, \cos \theta_0, \sin \theta_0 \right)$ intersecting the Cauchy surface $C\equiv \{t=0\}$ at the points of the straight line 
\[
\left\{
\begin{array}{l}
\left(x-x_0\right) \cos \theta_0 + \left(y-y_0\right) \sin \theta_0 =0 \\
t= 0    .
\end{array}
\right.
\] 
   
It is trivial to see that for any $\beta\in \mathcal{N}$ in the above orbit of $\oplus$, the chronological past of $\beta$ is
\[
I^{-}\left(\beta\right) = \left\{ \left(t,x,y\right)\in \mathbb{M}^3: t< \left(x-x_0\right) \cos \theta_0 + \left(y-y_0\right) \sin \theta_0  \right\} 
\]
then, any point in the future $L$-boundary corresponds to a point in the future $c$-boundary.

In both cases, the $L$--boundary coincides with the part of $c$--boundary of $\mathbb{M}^{3}_{(a,b)}$ accessible by light rays, but if $b=\infty$, it is not possible for the $L$--boundary to obtain the points of the timelike $c$--boundary because there exist inextensible timelike curves with chronological past which is not the chronological past of any light ray (see remark \ref{remark-light-ray-no-accessible}).

\end{example}

\begin{figure}[h]
\centering
\begin{subfigure}[t]{0.40\textwidth}
\centering
\begin{tikzpicture}[scale=1]
\draw[->] (0,0) -- (-1.5,-1.5) node[anchor=north west] {$x$};
\draw[->] (0,0) -- (0,1.5) node[anchor=south west] {$t$};
\draw[->] (0,0) node[anchor=south west] {$0$} -- (4.5,0) node[anchor=north] {$y$};
\fill[color=red!10] (2,-0.5)  arc (90:270:1.2cm and 0.3cm);
\fill[color=red!10] (2,-0.5)  arc (90:-90:1.2cm and 0.3cm);
\fill[color=red!70, shading=axis, opacity=0.15]
	 (0.8,-0.8) -- (2,0.4) -- (3.2,-0.8) -- (3.2,-0.8) arc (0:-180:1.2cm and 0.3cm);
\fill[color=blue!80, opacity=0.2] (0,1.2) -- (-1.5,-0.3) -- (3,-0.3)--(4.5,1.2)--cycle;,
\draw[thick,color=blue] (2,0.4)  -- (1.2,-1);       
\draw[thick,color=blue!30] (0.8,-1.7)  -- (1.2,-1);  
\draw[thick,color=red] (2,0.4)  -- (2.8,-1);        
\draw[thick,color=red!30] (3.2,-1.7)  -- (2.8,-1);  
\draw[densely dotted] (2,0.4)--(2,-0.8) -- (1.2,-1)--(0.9,-1.3);
\draw[densely dotted] (1,-1.2) arc (-120:25:0.45cm and 0.25cm);
\draw (1.6,-1.4) node {$\theta_0$};
\fill (1.2,-1) circle (1.5pt);
\draw (1.15,-1.1) node[anchor=east] {$q$};
\draw[thick,fill=white] (2,0.4) node[anchor=south] {$\widetilde{p}_{\gamma}$} circle (1.5pt);
\draw (-0.1,1.2) node[anchor=east] {$b$} -- (0.1,1.2);
\draw (4.2,1.2) node[anchor=north east] {$\{t=b\}$};
\draw[color=blue] (1.7,-0.5) node {$\gamma$};
\draw[color=red] (2.3,-0.5) node {$\beta$};
\end{tikzpicture}
 \caption{Case $0<b<\infty$. The orbit of $\oplus$ passing through $\gamma$ consists of all light rays $\beta$ arriving at $\widetilde{p}_{\gamma}\notin \mathbb{M}^3_{(a,b)}$. It forms the half of the cone in the figure which can be identified with the point $\widetilde{p}_{\gamma}\in\{t=b\}$. }
  \label{fig-L-boundary-M3-a}
\end{subfigure}
\hspace{5pt}
\begin{subfigure}[t]{0.40\textwidth}
\centering
\begin{tikzpicture}[scale=1]
\draw[->] (0,0) -- (-1.5,-1.5) node[anchor=north west] {$x$};
\draw[->] (0,0) -- (0,1.5) node[anchor=south west] {$t$};
\draw[->] (0,0) node[anchor=south west] {$0$} -- (4.5,0) node[anchor=north] {$y$};
\draw[densely dashed] (1.2,-1)  arc (-135:225:0.8cm and 0.14cm);
\draw[densely dashed] (1.2,-1)  arc (-135:225:1cm and 0.22cm);
\draw[densely dashed] (1.2,-1)  arc (-135:225:1.2cm and 0.3cm);
\draw[densely dashed] (1.2,-1)  arc (-135:225:1.4cm and 0.38cm);
\draw[densely dashed] (1.2,-1)  arc (225:150:1.8cm and 0.44cm);
\draw[densely dashed] (1.2,-1)  arc (-135:-10:1.8cm and 0.44cm);
\draw[densely dotted] (1.2,-1)--(2.134,0.595)--(2.134,-0.666) -- (1.2,-1)--(0.9,-1.3);
\draw[densely dotted] (1,-1.2) arc (-120:25:0.45cm and 0.25cm);
\draw[thick,color=blue]  (1.2,-1)--(2.4,1.1);
\draw[thick,color=blue!30] (0.8,-1.7)  -- (1.2,-1);
\draw (1.6,-1.4) node[anchor=north east] {$\theta_0$};
\draw[color=blue]  (-0.5,-0.5) -- (1.2,-1)--(2.9,-1.5) ;
\draw[densely dotted] (-0.5,-0.5)--(0.434,1.095)--(0.434,-0.166) -- (-0.5,-0.5);
\draw[densely dotted] (2.9,-1.5)--(3.834,0.095)--(3.834,-1.166) -- (2.9,-1.5);
\fill[color=red!70,opacity=0.15] (-0.5,-0.5) -- (0.434,1.095)--(3.834,0.095)--(2.9,-1.5) --cycle;
\draw[thick,color=red] (0.35,-0.75) -- (1.55,1.45);
\draw[thick,color=red!30] (-0.05,-1.45)  -- (0.35,-0.75);
\draw[color=blue] (1.85,0.17) node[anchor=east] {$\gamma$};
\draw[color=red] (1,0.42) node[anchor=west] {$\beta$};
\fill (1.2,-1) circle (1.5pt);
\draw (1.15,-1.1) node[anchor=east] {$q$};
\end{tikzpicture}
 \caption{Case $b=\infty$. The orbit of $\oplus$ passing through $\gamma$ consists of all light rays $\beta$ in the same direction that $\gamma$ intersecting the Cauchy surface $C=\{t=0\}$ in a straight line orthonormal to the direction of $\gamma$.}
  \label{fig-L-boundary-M3-b}
\end{subfigure}
 \caption{The $L$--boundaries of $\mathbb{M}^3_{(a,b)}$}
  \label{fig-L-boundary-M3}
\end{figure}
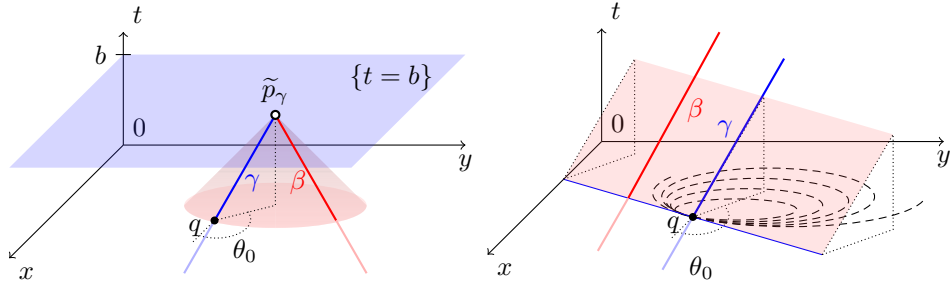

\begin{example}\label{example-de-Sitter}
We consider the $3$--dimensional \emph{de Sitter spacetime} $d\mathcal{S}^{3}$ embedded in $\mathbb{M}^4$ as the set verifying 
\begin{equation} \label{Secuacion}
-t^{2}+x^{2}+y^{2}+z^{2}=1 \, .
\end{equation}

As seen in example \ref{ex-contact-4M}, a null geodesic in $\mathbb{M}^4$ can be written by 
\[
\gamma(s)=\left(s,x_0+s \cos\theta_0\sin\phi_0,y_0+s \sin\theta_0\sin\phi_0,z_0+s \cos\phi_0  \right)\in \mathbb{M}^{4}
\]
where $\gamma(0)\in \overline{C}\equiv\{t=0\}$. Then $\gamma$ is a null geodesic in $d\mathcal{S}^{3}$ if equation (\ref{Secuacion}) is satisfied, so 
\[
-s^2 + \left(x_0+s\cos \theta_0 \sin \phi_0 \right)^{2} + \left(y_0 +s\sin \theta_0 \sin \phi_0\right)^{2} + \left(z_0+s\cos \phi_0\right)^{2} = 1 
\]
and simplifying we get
\begin{equation}  \label{cond_v1}
\left(x_0\cos \theta_0 +y_0\sin \theta_0\right)\sin \phi_0+z_0\cos \phi_0=0 \quad \Rightarrow \quad \cot \phi_0=-\frac{x_0 \cos \theta_0+y_0 \sin \theta_0}{z_0}  .
\end{equation}

Taking $t=0$ in equation (\ref{Secuacion}), we obtain a Cauchy surface $C\subset d\mathcal{S}^{3}$, in fact it is $2$--sphere, restriction of the Cauchy surface $\overline{C}\equiv\{t=0\}\subset \mathbb{M}^4$ to $d\mathcal{S}^{3}$. We can parametrize $C$ by 
\begin{equation}\label{Cparam}
\left\{ 
\begin{array}{l}
x=\cos u \sin w \\ 
y=\sin u \sin w \\ 
z=\cos w
\end{array}
\right.   
\end{equation}
where $(x_0,y_0,z_0)=(\cos u_0\sin w_0,\sin u_0 \sin w_0,\cos w_0)$.

By remark \ref{remark-contact-skies}, the tangent space to the sky of $\gamma(s)$ related to $\mathbb{M}^4$ is 
\begin{align*}
T_{\gamma}\overline{S}\left(\gamma(s)\right) & = \textstyle{ \mathrm{span}\left\{ s \left(\sin\theta_0 \sin\phi_0\left( \frac{\partial}{\partial x} \right)_{\gamma} - \cos\theta_0 \sin\phi_0\left( \frac{\partial}{\partial y} \right)_{\gamma}\right)+ \left( \frac{\partial}{\partial \theta} \right)_{\gamma},  \right. } \\
& \textstyle{ \left. s\left(-\cos\theta_0 \cos \phi_0 \left( \frac{\partial}{\partial x} \right)_{\gamma} - \sin\theta_0 \cos \phi_0 \left( \frac{\partial}{\partial y} \right)_{\gamma}  + \sin \phi_0 \left( \frac{\partial}{\partial z} \right)_{\gamma}\right)+\left( \frac{\partial}{\partial \theta} \right)_{\gamma}  \right\}  }
\end{align*}
and then 
\begin{align*}
\overline{\oplus}_{\gamma}&=\lim_{s\mapsto\infty} T_{\gamma}S\left(\gamma(s)\right)  = \\ 
& =\textstyle{ \mathrm{span}\left\{\sin\theta_0 \sin\phi_0\left( \frac{\partial}{\partial x} \right)_{\gamma} - \cos\theta_0 \sin\phi_0\left( \frac{\partial}{\partial y} \right)_{\gamma},  \right. } \\
& \textstyle{ \left. -\cos\theta_0 \cos \phi_0 \left( \frac{\partial}{\partial x} \right)_{\gamma} - \sin\theta_0 \cos \phi_0 \left( \frac{\partial}{\partial y} \right)_{\gamma}  + \sin \phi_0 \left( \frac{\partial}{\partial z} \right)_{\gamma} \right\}  }  .
\end{align*}

Integrating $\overline{\oplus}$, we obtain that its orbit passing through $(x_0,y_0,z_0,\theta_0,\phi_0)\in\mathcal{N}_{\mathbb{M}^4}$ is defined by
\begin{equation*} 
\left\{ 
\begin{array}{l}
x(\tau,\eta) = x_0+\tau\sin \theta_0 \sin \phi_0 - \eta \cos \theta_0 \cos \phi_0 \\ 
y(\tau,\eta) = y_0-\tau\cos \theta_0 \sin \phi_0 - \eta \sin \theta_0 \cos \phi_0 \\ 
z(\tau,\eta) = z_0+ \eta \sin \phi_0 \\ 
\theta(\tau,\eta) = \theta_0 \\
\phi(\tau,\eta) = \phi_0%
\end{array}%
\right.   
\end{equation*}%
which verifies the equation
\begin{equation} \label{eq-l-boundary-4M}
\left\{
\begin{array}{l}
\cos \theta_0 \sin \phi_0 \cdot\left(x-x_0\right)+ \sin \theta_0 \sin \phi_0 \cdot\left(y-y_0\right)+ \cos \phi_0 \cdot \left(z-z_0\right)=0 \\ 
\theta = \theta_0 \\
\phi = \phi_0  .
\end{array}%
\right.   
\end{equation}%

Substituting (\ref{cond_v1}) and (\ref{Cparam}) in (\ref{eq-l-boundary-4M}), we obtain the expression of the orbit of the field $\oplus$ in $\mathcal{N}_{d\mathcal{S}^3}$ as the restriction of the orbit of $\overline{\oplus}$ to $d\mathcal{S}^3$. So, we have
\begin{align*}
&  \cos \theta_0 \sin \phi_0 \cos u \sin w + \sin \theta_0 \sin \phi_0 \sin u \sin w+ \cos \phi_0 \cos  w=0   \\ 
\Rightarrow &  \tan w \left(\cos \theta_0 \cos u + \sin \theta_0  \sin u \right)=-\cot \phi_0 \\ 
\Rightarrow  &  \tan w \cos\left( u -\theta_0 \right)=-\cot \phi_0    \\
\text{ (by (\ref{cond_v1}))}\Rightarrow  & \tan w \cos\left( u -\theta_0 \right)=\tan w_0 \cos\left( u_0 -\theta_0 \right) 
\end{align*}%
and therefore, the orbits of $\oplus$ must satisfy the equation
\begin{equation}\label{eq-cielo-desitter}
\tan w \cos\left( u -\theta_0 \right)=\tan w_0 \cos\left( u_0 -\theta_0 \right)  .
\end{equation}

For any $\gamma_0\in\mathcal{N}_{d\mathcal{S}^3}$ with coordinates $(u_0,w_0,\theta_0)\simeq\gamma_0$, it is straightforward to check that $(u,w)$ in eq. (\ref{eq-cielo-desitter}) corresponds with a maximal circumference on the Cauchy surface $C=\mathbb{S}^2$. This circumference is the limiting curve of the intersection of the lights rays in $S(\gamma_0(s))$ with $C$ whenever $s\mapsto \infty$ (see figure \ref{figura-orbita-de-Sitter-1}). 
For fixed $\theta=\theta_0$ we obtain all maximal circumferences passing through the points $(u,w)=\left(\theta_0\pm\frac{\pi}{2},\frac{\pi}{2}\right)$ (see figure \ref{figura-orbita-de-Sitter-2}).
Then, moving $\theta_0$ in $\left[0,2\pi\right)$, we get all maximal circumferences in the sphere $C$. 
Observe that each maximal circumference can be obtained twice as solutions of eq. (\ref{eq-cielo-desitter}) for two different light rays with the same values of $(u_0,w_0)$ and the antipodal values $\theta_0$ and $\theta_0+\pi$. So, although the space of maximal circumferences in $C$ is diffeomorphic to the projective space $\mathbb{P}(\mathbb{R}^3)$, the orbits of $\oplus$ do not coincide for antipodal values of $\theta_0$.
Therefore, we have that the space of orbits of $\oplus$, that is, the future $L$--boundary is $\partial^{+}(d\mathcal{S}^3)\simeq \mathbb{S}^2$. Notice that $c$--boundary and $L$--boundary of $d\mathcal{S}^3$ coincide.  
\end{example}

\begin{figure}[h]
\centering
\begin{subfigure}[t]{0.40\textwidth}
\centering
\begin{tikzpicture}[scale=1]
  \begin{axis}[xtick = {-1,0,1},
    ytick = {-1,0,1},ztick = {-1,0,1},xlabel = $x$,ylabel = $y$,zlabel = {$z$},
    ticklabel style = {font = \tiny},axis equal image,colormap={custom}{rgb255(0cm)=(240,240,240);
            rgb255(1cm)=(230,230,230);rgb255(2cm)=(220,220,220);rgb255(3cm)=(210,210,210);},    
    view={135}{25}]
    \addplot3 [surf,domain=0:360,samples=60,y domain=0:180,samples y=30,line join=round,opacity=10,shader=interp,variable=\t,point meta={y}]
   ( {cos(t)*sin(y)},{sin(t)*sin(y)},{cos(y)} );
   \addplot3 [thick,red!20,domain=145:327,samples=60,y domain=0:180,samples y=0,line join=round,variable=\t] ( {0},{sin(t)},{cos(t)} );
   \addplot3 [blue!20,domain=200:325,samples=60,y domain=0:180,samples y=0,line join=round,variable=\t] ( {(2/5)*(1-cos(t))},{((4/5)^0.5)*sin(t)},{(4/5)*(0.25+cos(t))} );
   \addplot3 [blue!20,domain=170:350,samples=60,y domain=0:180,samples y=0,line join=round,variable=\t] ( {(4/17)*(1-cos(t))},{((16/17)^0.5)*sin(t)},{(16/17)*(0.0625+cos(t))} );
  \addplot3 [blue,domain=0:360,samples=60,y domain=0:180,samples y=0,line join=round,variable=\t] ( {(4/17)*(1-cos(t))},{((1/17)^0.5)*sin(t)},{(1/17)*(16+cos(t))} );
  \addplot3 [blue,domain=0:360,samples=60,y domain=0:180,samples y=0,line join=round,variable=\t] ( {(2/5)*(1-cos(t))},{((1/5)^0.5)*sin(t)},{(1/5)*(4+cos(t))} );
  \addplot3 [blue,domain=0:360,samples=60,y domain=0:180,samples y=0,line join=round,variable=\t] ( {(1/2)*(1-cos(t))},{((1/2)^0.5)*sin(t)},{(1/2)*(1+cos(t))} );
  \addplot3 [blue,domain=-35:200,samples=60,y domain=0:180,samples y=0,line join=round,variable=\t] ( {(2/5)*(1-cos(t))},{((4/5)^0.5)*sin(t)},{(4/5)*(0.25+cos(t))} );
  \addplot3 [blue,domain=-10:170,samples=60,y domain=0:180,samples y=0,line join=round,variable=\t] ( {(4/17)*(1-cos(t))},{((16/17)^0.5)*sin(t)},{(16/17)*(0.0625+cos(t))} );
  \addplot3 [thick,red,domain=-33:145,samples=60,y domain=0:180,samples y=0,line join=round,variable=\t] ( {0},{sin(t)},{cos(t)} );
\addplot3[fill=black] ( 0,0,1) node[anchor=south west] {$\gamma(0)$} circle (1pt);
\addplot3[draw=black] ( 1,1,-0.9) node[anchor=south]  {$C\subset d\mathcal{S}^3$};
\end{axis}
\end{tikzpicture}
\caption{Skies of a light ray in $d\mathcal{S}^3$. Each blue curve describes the trace  on the Cauchy surface $C=\mathbb{S}^2$ of the sky of a point $\gamma(s)$ along a fixed light ray $\gamma$. The red circle corresponds to the limiting sky of the point of $\gamma$ at the $L$--boundary.}
  \label{figura-orbita-de-Sitter-1}
\end{subfigure}
\hspace{5pt}
\begin{subfigure}[t]{0.40\textwidth}
\centering
\begin{tikzpicture}[scale=1]
  \begin{axis}[xtick = {-1,0,1},ytick = {-1,0,1},ztick = {-1,0,1},xlabel = $x$,ylabel = $y$,zlabel = {$z$},
    ticklabel style = {font = \tiny},axis equal image,colormap={custom}{rgb255(0cm)=(240,240,240);
            rgb255(1cm)=(230,230,230);rgb255(2cm)=(220,220,220);rgb255(3cm)=(210,210,210);},    
    view={135}{25}]
  \addplot3 [surf,domain=0:360,samples=60,y domain=0:180,samples y=30,line join=round,opacity=10,shader=interp,variable=\t,point meta={y}]
   ( {cos(t)*sin(y)},{sin(t)*sin(y)},{cos(y)} );
\addplot3 [blue!20,domain=135:315,samples=60,y domain=0:180,samples y=0,line join=round,variable=\t] ( {cos(t)},{sin(t)},{0} );
\addplot3 [blue!20,domain=134:314,samples=60,y domain=0:180,samples y=0,line join=round,variable=\t] ( {cos(15)*cos(t)},{sin(t)},{sin(15)*cos(t)} );
\addplot3 [blue!20,domain=133:313,samples=60,y domain=0:180,samples y=0,line join=round,variable=\t] ( {cos(30)*cos(t)},{sin(t)},{sin(30)*cos(t)} );
\addplot3 [blue!20,domain=131:311,samples=60,y domain=0:180,samples y=0,line join=round,variable=\t] ( {cos(45)*cos(t)},{sin(t)},{sin(45)*cos(t)} );
\addplot3 [blue!20,domain=134:315,samples=60,y domain=0:180,samples y=0,line join=round,variable=\t] ( {cos(60)*cos(t)},{sin(t)},{sin(60)*cos(t)} );
\addplot3 [blue!20,domain=138:320,samples=60,y domain=0:180,samples y=0,line join=round,variable=\t] ( {cos(75)*cos(t)},{sin(t)},{sin(75)*cos(t)} );
\addplot3 [blue!20,domain=145:327,samples=60,y domain=0:180,samples y=0,line join=round,variable=\t] ( {0},{sin(t)},{cos(t)} );
\addplot3 [blue!20,domain=143:325,samples=60,y domain=0:180,samples y=0,line join=round,variable=\t] ({cos(15)*cos(t)},{sin(t)},{-sin(15)*cos(t)});
\addplot3 [blue!20,domain=145:335,samples=60,y domain=0:180,samples y=0,line join=round,variable=\t] ({cos(30)*cos(t)},{sin(t)},{-sin(30)*cos(t)});
\addplot3 [blue!20,domain=159:348,samples=60,y domain=0:180,samples y=0,line join=round,variable=\t] ({cos(45)*cos(t)},{sin(t)},{-sin(45)*cos(t)});
\addplot3 [blue!20,domain=180:360,samples=60,y domain=0:180,samples y=0,line join=round,variable=\t] ({cos(60)*cos(t)},{sin(t)},{-sin(60)*cos(t)});
\addplot3 [blue!20,domain=200:380,samples=60,y domain=0:180,samples y=0,line join=round,variable=\t] ({cos(75)*cos(t)},{sin(t)},{-sin(75)*cos(t)});

\addplot3 [blue,domain=-45:135,samples=60,y domain=0:180,samples y=0,line join=round,variable=\t] ( {cos(t)},{sin(t)},{0} );
\addplot3 [blue,domain=-46:134,samples=60,y domain=0:180,samples y=0,line join=round,variable=\t] ( {cos(15)*cos(t)},{sin(t)},{sin(15)*cos(t)} );
\addplot3 [blue,domain=-47:133,samples=60,y domain=0:180,samples y=0,line join=round,variable=\t] ( {cos(30)*cos(t)},{sin(t)},{sin(30)*cos(t)} );
\addplot3 [blue,domain=-49:131,samples=60,y domain=0:180,samples y=0,line join=round,variable=\t] ( {cos(45)*cos(t)},{sin(t)},{sin(45)*cos(t)} );
\addplot3 [blue,domain=-45:134,samples=60,y domain=0:180,samples y=0,line join=round,variable=\t] ( {cos(60)*cos(t)},{sin(t)},{sin(60)*cos(t)} );
\addplot3 [blue,domain=-40:138,samples=60,y domain=0:180,samples y=0,line join=round,variable=\t] ( {cos(75)*cos(t)},{sin(t)},{sin(75)*cos(t)} );
\addplot3 [blue,domain=-33:145,samples=60,y domain=0:180,samples y=0,line join=round,variable=\t] ( {0},{sin(t)},{cos(t)} );
\addplot3 [blue,domain=-35:143,samples=60,y domain=0:180,samples y=0,line join=round,variable=\t] ( {cos(15)*cos(t)},{sin(t)},{-sin(15)*cos(t)} );
\addplot3 [blue,domain=-25:145,samples=60,y domain=0:180,samples y=0,line join=round,variable=\t] ( {cos(30)*cos(t)},{sin(t)},{-sin(30)*cos(t)} );
\addplot3 [blue,domain=-12:159,samples=60,y domain=0:180,samples y=0,line join=round,variable=\t] ( {cos(45)*cos(t)},{sin(t)},{-sin(45)*cos(t)} );
\addplot3 [blue,domain=0:180,samples=60,y domain=0:180,samples y=0,line join=round,variable=\t] ( {cos(60)*cos(t)},{sin(t)},{-sin(60)*cos(t)} );
\addplot3 [blue,domain=20:200,samples=60,y domain=0:180,samples y=0,line join=round,variable=\t] ( {cos(75)*cos(t)},{sin(t)},{-sin(75)*cos(t)} );
\addplot3[fill=black] ( 0,1,0) circle (1pt);
\addplot3[draw=black] ( 1,1,-0.9) node[anchor=south]  {$C\subset d\mathcal{S}^3$};
\end{axis}
\end{tikzpicture}
\caption{Orbits of $\oplus$ in $d\mathcal{S}^3$.  The traces of the orbits on the Cauchy surface $C=\mathbb{S}^2$ are maximal circles. Fixed $\theta_0=0$, the orbits of $\oplus$ correspond to the maximal circles passing by the points $(0,\pm 1,0)$.}
  \label{figura-orbita-de-Sitter-2}
\end{subfigure}
\caption{Orbits of $\overline{\mathcal{D}^{\sim}}$ in $d\mathcal{S}^3$.}
  \label{figura-orbita-de-Sitter}
\end{figure}
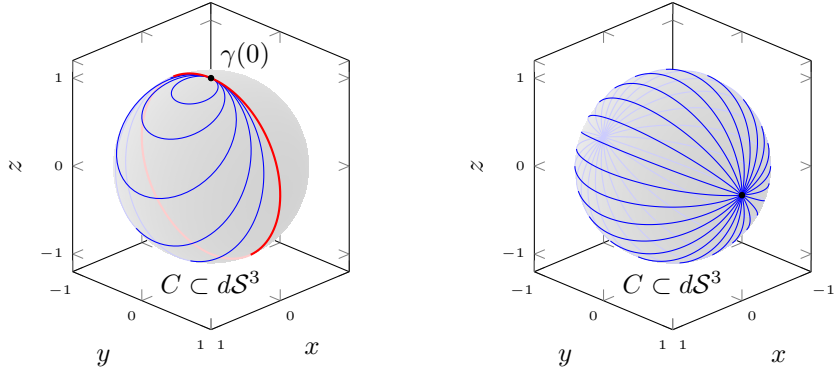

\begin{definition}
A $L$--spacetime is said to be \emph{proper} if the total distribution $\overline{\mathcal{D}^\sim}$ is smooth and regular with Hausdorff space of leaves. 
\end{definition}

If $M$ is a proper $L$--spacetime, the future $L$--boundary $\partial^{+} M = \partial^{+}\widetilde{\mathcal{N}} / \partial^{+}\mathcal{D}^{\sim}$ built in corollary \ref{corollary-main-theorem} is a smooth boundary for the manifold $\overline{M} = M \cup \partial^{+} M$.  Moreover, the hausdorffness of the quotient space $\overline{\widetilde{\mathcal{N}}}$ over the regular distribution $\overline{\mathcal{D}^\sim}$ assures that $\overline{M}$ is Hausdorff. This can be summarized in the following result. 

\begin{corollary}\label{Corollary-ext-2}
If $M$ is a proper $L$--spacetime then the canonical extension $\overline{M}$ of $\left(M,\mathcal{C}\right)$ exists.
\end{corollary}

\section{Characterization of the differentiable structure of the canonical extension: \texorpdfstring{$L$--extensions}{L--extensions}}\label{sec:Lextensions}

In this section we will give a characterization of the canonical extension of a proper $L$--spacetime $\left(M,\mathcal{C}\right)$ based on how light rays arrive at the boundary. First, let us present some properties satisfied at the $L$--boundary. 

\subsection{Properties of \texorpdfstring{$L$--boundary}{L--boundary}}\label{sec:Lextensions-properties}

In section \ref{sec:Lboundary-canonical}, we have seen that, whenever $M$ is a proper $L$--spacetime, we get a diffeomorphism $\widetilde{S}:\overline{M}\rightarrow \overline{\widetilde{\mathcal{N}}} / \overline{\mathcal{D}^{\sim}}$ and moreover, we can consider the quotient map $\widetilde{\pi}: \overline{\widetilde{\mathcal{N}}}\rightarrow \overline{\widetilde{\mathcal{N}}} / \overline{\mathcal{D}^{\sim}}$ obtaining the following commutative diagram 
\begin{equation}\label{diagram-N-tilde}
\begin{tikzpicture}[every node/.style={midway}]
\matrix[column sep={7em,between origins},
        row sep={2em}] at (0,0)
{ \node(PN1)   {$\overline{\widetilde{\mathcal{N}}}$}  ; & \node(N) {$\overline{\widetilde{\mathcal{N}}} / \overline{\mathcal{D}^{\sim}}$}; \\
; &  \node(PN2) {$\overline{M}$};                   \\};
\draw[->] (PN1) -- (N) node[anchor=south]  {$\widetilde{\pi}$};
\draw[->] (N) -- (PN2) node[anchor=west]  {$\widetilde{S}^{-1}$};
\draw[->] (PN1)   -- (PN2) node[anchor=north east] {$\mathbf{\rho}$};
\end{tikzpicture}
\end{equation}
where $\rho$ is one of the branches of the double fibration (\ref{double-fibration-3}).  

If we consider the following composition of diffeomorphisms
\[
\begin{tabular}{rcccl}
$\mathcal{N}$ & $\overset{i}{\longrightarrow}$ & $\mathcal{N}\times \{1\}$ & $\overset{\left.\varepsilon\right|_{\mathcal{N}\times \{1\}}}{\longrightarrow}$ & $\partial^{+}\widetilde{\mathcal{N}}$ \\
$\gamma$ & $\mapsto$ & $\left(\gamma,1\right)$  & $\mapsto$  & $\widetilde{\gamma}\left(1\right)$ 
\end{tabular}
\]
then, projecting on $\overline{M}$, we have that
\[
\infty^{+}=\left.\rho\right|_{\partial^{+}\widetilde{\mathcal{N}}} \circ \left.\varepsilon\right|_{\mathcal{N}\times \{1\}} \circ i : \mathcal{N}\rightarrow \partial^{+} M 
\]
is a surjective submersion.
The same construction can be done to get a surjective submersion $\infty^{-}$ onto the past $L$--boundary.
Then the canonical extension verify the following property:

\begin{enumerate}[start=1,label={\bfseries P\arabic* }]
\item \label{property-p1} The map $\infty^{+}:\mathcal{N}\rightarrow \partial^{+} M$ defined by $\infty^{+}\left(\gamma\right)=\lim_{\mathbf{t}\mapsto + 1} \overline{\gamma}\left(\mathbf{t}\right)$ is a surjective submersion. 
\end{enumerate}

Observe that a light ray can be extended to the $L$--boundary by $\overline{\gamma}(\mathbf{t})=\rho\left( \widetilde{\gamma}(\mathbf{t}) \right)$ for all $\mathbf{t}\in\left(-1,1\right]$.
Moreover, by (\ref{extension-gamma-prima-tilde}), $\widetilde{\gamma}'(\mathbf{t})=\left(\frac{\partial}{\partial \mathbf{t}}\right)_{\gamma}\neq 0$ and, by construction, the map $\overline{X}$ used in the proof of theorem \ref{theorem-extension} to extend $X$ verifies  $\frac{\partial \overline{X}}{\partial \mathbf{s}}(\gamma,\mathbf{t},0)\neq 0$. This implies that the curve $\widetilde{\gamma}$ is transversal to the orbit of $\overline{\mathcal{D}^{\sim}}$ passing through $\widetilde{\gamma}(\mathbf{t})$. 
Then, applying $\rho$, the curve $\overline{\gamma}(\mathbf{t})$ is regular. Moreover, since equation (\ref{extension-gamma-prima-tilde}) also holds  for $\mathbf{t}=+ 1$ then $\overline{\gamma}$ is transversal to $\partial^{+} M$.
So we obtain the second property of the canonical extension:

\begin{enumerate}[start=2,label={\bfseries P\arabic* }]
\item \label{property-p2} For every $\mathcal{N}_U\subset \mathcal{N}$, the map $\rho\circ\varepsilon:\mathcal{N}_U\times\left[-1,1\right]\rightarrow\overline{M}$, where $\gamma\left(\mathbf{t}\right)=\rho\circ\varepsilon\left(\gamma, \mathbf{t}\right)$ is a projective parametrization of $\gamma\in \mathcal{N}_U$ such that $\frac{\partial \rho\circ\varepsilon}{\partial \mathbf{t}}\left(\gamma,+ 1\right)\notin T_{\infty^{+}\left(\gamma\right)}\partial^{+}M$.
\end{enumerate}

For a more detailed proof of the above properties, see \cite[Prop. 7.1]{Ba18}.

The differentiable structure at the boundary is a key element for an extension and it affects to the way in which the light rays arrive at the boundary, so the features of the parametrization of a light ray by a projective parameter can change very subtly depending on if it corresponds to a  canonical--like extension or to another general extension. 
The following example illustrates this behaviour.

\begin{example}\label{example-no-admiss}
In this example we will only consider the past $L$--extension of the spacetime.
Let $\mathbb{M}^{3}_{(0,\infty)}=\left\{ \left(t,x,y\right)\in \mathbb{R}^3 : t>0 \right\}$ be the Minkowski block as seen in example \ref{example-M3-block}.  
We denote by $\overline{\mathbb{M}}^{3}_{(0,\infty)}$ the canonical past extension of $\mathbb{M}^{3}_{(0,\infty)}$, that is $\overline{\mathbb{M}}^{3}_{(0,\infty)}= \left\{ \left(t,y,z\right)\in \mathbb{R}^3 : t\geq 0 \right\}$ with the standard differentiable structure (see example \ref{example-M3-block}).
A Cauchy surface $C=\{  t=1 \}\subset \mathbb{M}$ provides a coordinate chart for $\mathcal{N}$ by fixing, for every $\gamma\in \mathcal{N}$, the point $\gamma\cap C\simeq\left(x_0,y_0\right)$ and an angle $\theta_0$ such that $\gamma\left(0\right)=\left(1,x_0,y_0\right)\in C$ and $\gamma'\left(0\right)=\left(1,\cos \theta_0,\sin \theta_0\right)\in \mathbb{N}^{+}_{\gamma\left(0\right)}$.
So, a parametrization of $\gamma$ as a null geodesic can be written by $\gamma\left(s\right)=\left(s+1,x_0+s\cos \theta_0, y_0+s\sin \theta_0 \right)$ for $s\in\left(-1,\infty\right)$. 
Then we can identify $\gamma\simeq\left(x_0,y_0,\theta_0\right)$.

The parameter $s$ is diffeomorphic to a projective parameter by $s=h(\mathbf{t})$ where $h:\left[-1,1\right)\rightarrow \left[-1,\infty\right)$ is a function such that $h'\left(\mathbf{t}\right)>0$ for all $\mathbf{t}\in \left[-1,1\right)$ and $\gamma\circ h:\left(-1,1\right)\rightarrow M$ is projective. 
Indeed, as argued in example \ref{example-M3-block}, we have $h\left(\mathbf{t}\right)=\frac{-2\mathbf{t}}{\mathbf{t}-1}$. 
Then the map  
\[
\Psi\left(\gamma, s\right)=\Psi\left(x_0,y_0,\theta_0,s\right) = \left(s+1,x_0+s\cos \theta_0, y_0+s\sin \theta_0 \right)
\]
with $s\in\left[-1,\infty\right)$ verifies the property \ref{property-p2} as $\rho \circ \varepsilon$ for the canonical past extension.

Observe that if we consider the spacetime $M=\left\{ \left(u,v,w\right)\in \mathbb{R}^3 : u>0 \right\}$ with the metric $\mathbf{\overline{g}}=-u^2 du\otimes du + dv \otimes dv + dw\otimes dw$, then $\phi: \mathbb{M}^{3}_{(0,\infty)}\rightarrow M$ defined by $\phi\left(t,x,y\right)=\left(\sqrt{2t},x,y\right)$ is an isometry. 
Then, the canonical past extension of $M$ must be $\overline{M}= \left\{ \left(u,v,w\right)\in \mathbb{R}^3 : u\geq 0 \right\}$ such that its differentiable structure verifies that the extension $\overline{\phi}\left(t,x,y\right)=\left(\sqrt{2t},x,y\right)$ of $\phi$ to $\overline{\mathbb{M}}^{3}_{(0,\infty)}$ is a diffeomorphism. 
We will call $\overline{M}$ to such extension and $\overline{M}_s$ to the same topological manifold equipped with the standard differentiable structure. 
Trivially, the identity map $\overline{\mathrm{id}}:\overline{M}\rightarrow \overline{M}_s$ is not a diffeomorphism and therefore $\overline{M}_s$ can not be the canonical extension of $M$. 

Observe that $\overline{\gamma}\left(s\right)=\phi\left(\gamma\left(s\right)\right)=\left(x_0+s\cos \theta_0, y_0+s\sin \theta_0 ,\sqrt{2\left(s+1\right)} \right)$ defines an inextensible null geodesic in $M$ for $s\in\left(-1,\infty\right)$.
If we change the parameter by $\tau^2=s+1$, we obtain a regular parameter $\tau\in\left(0,\infty\right)$ which is diffeomorphic to the canonical projective parameter for $\mathbf{t}\in\left(-1,1\right)$.
But the resulting map $\overline{\Psi}:\mathcal{N}\times\left[0,\infty\right)\rightarrow \overline{M}_s$ given in coordinates by
\[
\overline{\Psi}\left(x_0,y_0,\theta_0,\tau\right) = \left(x_0+\left(\tau^2-1\right)\cos \theta_0, y_0+\left(\tau^2-1\right)\sin \theta_0 ,\sqrt{2}\tau \right)
\]
defines regular parametrizations of light rays but not diffeomorphic to the projective parameter for $\tau\in\left[0,\infty\right)$ because it is not smooth at $\tau=0$.
\end{example}

The following definition is motivated by the example \ref{example-no-admiss}.

\begin{definition}\label{def-parametrizations}
A regular, future--directed and inextensible parametrization $\gamma:\left(a,b\right)\rightarrow M$ of a light ray $\gamma\in \mathcal{N}$ is said to be
\begin{enumerate}
\item \emph{projective} if $\widetilde{\gamma}(s)=\sigma\left( \left[ \gamma'(s) \right] \right)\in \mathbb{P}\left(\mathcal{H}_{\gamma}\right)$ defines a projectivity in the fibre $\mathbb{P}\left(\mathcal{H}_{\gamma}\right)$, and
\item  \emph{admissible} if there exists a diffeomorphism $h:\left(c,d\right]\rightarrow \left(a,b\right]$ such that $h'\left(t\right)>0$ for all $t\in \left(c,d\right]$ and $\gamma\circ h:\left(c,d\right)\rightarrow M$ is a projective parametrization.
\end{enumerate}
\end{definition}

It can be trivially observed that any projective parametrization of $\gamma\in \mathcal{N}$ is admissible, and any admissible parameter is diffeomorphic (in the sense of admissibility of definition \ref{def-parametrizations}) to the canonical projective parameter $\mathbf{t}\in\left(-1,1\right]$.

As seen in example \ref{example-no-admiss}, we have to notice that every parametrization of $\gamma\in \mathcal{N}$ can be reparametrized diffeomorphically by the canonical projective parameter $\mathbf{t}$, but it is not an admissible parametrization.

\subsection{Characterization of \texorpdfstring{$L$--boundary}{L--boundary}}\label{sec:Lextensions-Char}

Now, we would like to offer a characterization of the canonical extension based on the properties \ref{property-p1} and \ref{property-p2}  enunciated in the previous section \ref{sec:Lextensions-properties}.

The proposed characterization is given in the following definition.

\begin{definition}\label{def-L-extension}
A \emph{future $L$--extension} of a conformal manifold $\left(M,\mathcal{C}\right)$ is defined as a Hausdorff smooth manifold $\overline{M}=M \cup \partial^{+} M$ where $\partial^{+} M = \overline{M}-M$ is a closed hypersurface of $\overline{M}$ named \emph{future $L$--boundary} satisfying the following properties:
\begin{enumerate}
\item \label{L-ext-cond-0} If $\gamma:\left(a,b\right)\rightarrow M$ is a continuous parametrization of $\gamma\in \mathcal{N}$, then $\lim_{s\mapsto b^{-}}\gamma\left(s\right)=\infty^{+}_{\gamma}\in \partial^{+}M$. 
\item \label{L-ext-cond-2} The map $\infty^{+}:\mathcal{N}\rightarrow \partial^{+} M$ defined by $\infty^{+}\left(\gamma\right)=\infty^{+}_{\gamma}$ is a surjective submersion.
\item \label{L-ext-cond-1} For every $\gamma_0\in \mathcal{N}$ there exists a neighbourhood $\mathcal{U}\subset \mathcal{N}$ and a differentiable map $\Psi_{\mathcal{U}}:\mathcal{U}\times\left(a,b\right]\rightarrow\overline{M}$, where $\gamma\left(s\right)=\Psi_{\mathcal{U}}\left(\gamma, s\right)$ is an admissible parametrization of $\gamma\in \mathcal{U}$ for $s\in\left(a,b\right)$ and such that $\frac{\partial \Psi_{\mathcal{U}}}{\partial s}\left(\gamma,b\right)\notin T_{\infty^{+}\left(\gamma\right)}\partial^{+}M$.
\end{enumerate}
A \emph{past $L$--extension} $\overline{M}=M \cup \partial^{-} M$ can be defined analogously. If there exists any future or past $L$--extension of $\left(M,\mathcal{C}\right)$, then $\left(M,\mathcal{C}\right)$ is said to be \emph{future} or \emph{past $L$-extensible}.
\end{definition}

The interpretation of the condition \ref{L-ext-cond-1} of definition \ref{def-L-extension} is that any light ray $\gamma=\gamma(t)$ parametrized inextensible to the future with a projective parameter $t\in (a,b)$ is smoothly extended to $t\in (a,b]$ with $\gamma(b)=\infty^{+}(\gamma)$ and $\gamma'(b)\neq 0$ is transversal to the future $L$--boundary.

From the definition \ref{def-L-extension}, since the map $\infty^{+}:\mathcal{N}\rightarrow \partial^{+}M$ is a surjective submersion then the inverse images
\[
S\left(p\right) = \left(\infty^{+}\right)^{-1}\left(p\right)= \{ \gamma\in \mathcal{N} : p = \infty^{+}\left(\gamma\right) \} \subset \mathcal{N}  
\]
define a regular distribution $\boxplus:\mathcal{N}\rightarrow \mathbb{P}\left(T\mathcal{N}\right)$ given by $\boxplus\left(\gamma\right)= T_{\gamma}S\left(\infty^{+}\left(\gamma\right)\right)$, and the map defined by
\[
\begin{tabular}{rccl}
$S:$ & $ \partial^{+} M$ & $\rightarrow$ & $\mathcal{N}/\boxplus$  \\
 & $p$ & $\mapsto$ & $S\left(p\right)$   
\end{tabular}
\]
is a diffeomorphism.

\begin{remark}\label{remark-change-parameter}
For any map $\Psi_{\mathcal{U}}:\mathcal{U}\times\left(a,b\right]\rightarrow\overline{M}$ according to the definition of $L$--extensions, the properties of the admissible parameter $s\in\left(a,b\right]$ ensure the existence of a change of parameter $h:\mathcal{U}\times\left(-1,1\right]\rightarrow \left(a,b\right]$, globally in $\mathcal{U}\subset \mathcal{N}$ such that the map $\overline{\Psi}_{\mathcal{U}}\left(\gamma,\mathbf{t}\right)=\Psi_{\mathcal{U}}\left(\gamma,h\left(\gamma,\mathbf{t}\right)\right)$ satisfies the condition \ref{L-ext-cond-1} of definition \ref{def-L-extension} and where $\mathbf{t}\in\left(-1,1\right]$ is the canonical projective parameter.

Moreover, if $\{\mathcal{U}_{\alpha}\}_{\alpha\in I}$ is an open covering of $\mathcal{N}_{U}$ such that $\mathcal{U}_{\alpha}\subset \mathcal{N}_{U}$ for all $\alpha\in I$, since 
\[
\left.\overline{\Psi}_{\mathcal{U}_{\alpha}}\right|_{\mathcal{U}_{\alpha}\cap\mathcal{U}_{\beta}\times\left(-1,1\right]} = \left.\overline{\Psi}_{\mathcal{U}_{\beta}}\right|_{\mathcal{U}_{\alpha}\cap\mathcal{U}_{\beta}\times\left(-1,1\right]}
\]
then trivially, it is possible to define $\overline{\Psi}_{\mathcal{N}_{U}}:\mathcal{N}_{U}\times\left(-1,1\right]\rightarrow\overline{M}$ extending all $\overline{\Psi}_{\mathcal{U}_{\alpha}}$.
\end{remark}

By properties \ref{property-p1} and \ref{property-p2}, the canonical extension is a $L$-extension such that $\Psi_{\mathcal{N}_{U}}=\rho\circ \varepsilon$. This result is summarized in the following corollary.

\begin{corollary}\label{Corol-canonical-extension}
The future (resp. past) canonical extension of a proper $L$--spacetime $\left(M,\mathcal{C}\right)$, built in section \ref{sec:LBoundary}, is a future (resp. past) $L$--extension.
\end{corollary}

For any two future $L$--extensions, the corresponding skies at the endpoint of any light ray coincide, that is $S\left(\infty_{1}^{+}\left(\gamma\right)\right) = S\left(\infty_{2}^{+}\left(\gamma\right)\right)$ for all $\gamma\in \mathcal{N}$ where $\infty^{+}_i$ with $i=1,2$ are the submersions of the $L$--extensions as in definition \ref{def-L-extension}.
The details of the proof are in \cite[Lem. 8.2]{Ba18}.

Now, let us see that definition \ref{def-L-extension} gives a characterization of the canonical extension.

\begin{theorem}\label{theorem-diff-struct}
Let $M$ be a proper $L$--spacetime, $\overline{M}_1 = M \cup \partial^{+} M_1$ the canonical future $L$--extension and $\overline{M}_{2} = M \cup \partial^{+} M_2$ any other future $L$--extension of $\left(M,\mathcal{C}\right)$, then the identity map $\mathrm{id}:M\rightarrow M$ can be extended as a diffeomorphism $\overline{\mathrm{id}}:\overline{M}_1 \rightarrow \overline{M}_2$.
\end{theorem}

\begin{proof}
Since $S\left(\infty_{1}^{+}\left(\gamma\right)\right) = S\left(\infty_{2}^{+}\left(\gamma\right)\right)$ for all $\gamma\in \mathcal{N}$, then the map $\phi:\partial^{+}M_1 \rightarrow \partial^{+}M_2$ given by $\phi\left(\infty^{+}_1\left(\gamma\right)\right)=\infty^{+}_2\left(\gamma\right)$ is well--defined and a bijection. 
Thus, the diagram
\begin{equation}\label{diagram-submersions}
\begin{tikzpicture}[every node/.style={midway}]
\matrix[column sep={3em,between origins},
        row sep={1em}] at (0,0)
{ ; &  \node(N)   { $\mathcal{N}$}  ; & ; \\
  \node(M1)   { $\partial^{+}M_1$} ; &   ; & \node(M2)   { $\partial^{+}M_2$} ;       \\};
\draw[->] (N) -- (M1) node[anchor=south east]  {$\infty^{+}_{1}$};
\draw[->] (N) -- (M2) node[anchor=south west]  {$\infty^{+}_{2}$};
\draw[->] (M1)   -- (M2) node[anchor=north] {$\phi$};
\end{tikzpicture}
\end{equation}
follows. 
Since $\infty^{+}_{1}$ and $\infty^{+}_{2}$ are smooth submersions then, by \cite[Prop. 6.1.2]{BC}, $\phi$ is a diffeomorphism.

Now, if $\Psi_{\mathcal{U}}$ is the map in the definition of $L$--extension corresponding to $\overline{M}_2$, since $\left.\Psi_{\mathcal{U}}\right|_{\mathcal{U}\times \left(-1,1\right)}=\left.\rho\circ\varepsilon\right|_{\mathcal{U}\times \left(-1,1\right)}$ and $\rho\circ\varepsilon$ is a submersion (see diagram (\ref{diagram-N-tilde}) and equation (\ref{eq-varepsilon})), then $\left.\Psi_{\mathcal{U}}\right|_{\mathcal{U}\times \left(-1,1\right)}$ is a submersion. 

On the other hand, since $\left(d\Psi_{\mathcal{U}}\right)_{\left(\gamma,1\right)}\left(  \frac{\partial}{\partial \mathbf{t}}  \right)_{\left(\gamma,1\right)}=\overline{\gamma}'\left(1\right) \neq 0$ 
with $\overline{\gamma}'\left(1\right)\notin T_{\overline{\gamma}\left(1\right)} \partial^{+} M_2$ and moreover $\phi$ is a diffeomorphism and $\left.\Psi_{\mathcal{U}}\right|_{\mathcal{U}\times \{1\}}=\left.\phi\right|_{\infty^{+}_{1}\left(\mathcal{U}\right)}$, then we have that $\left(d\Psi_{\mathcal{U}}\right)_{\left(\gamma,1\right)}$ is surjective, therefore $\Psi_{\mathcal{U}}$ is a submersion.

Denoting $\overline{V}_1= \rho\circ\varepsilon\left(\mathcal{U}\times \left(-1,1\right]\right)\subset \overline{M}_1$, $\overline{V}_2=\Psi_{\mathcal{U}}\left(\mathcal{U}\times \left(-1,1\right]\right)\subset \overline{M}_2$, $V_1= \overline{V}_1 \cap M$ and $V_2=\overline{V}_2 \cap M$
then, $V=V_1=V_2=\{\gamma\left(\mathbf{t}\right)\in M: \gamma\in \mathcal{U} \}$ and we have the following diagram 
\begin{equation}\label{diagram-submersions-2}
\begin{tikzpicture}[every node/.style={midway}]
\matrix[column sep={4em,between origins},
        row sep={2em}] at (0,0)
{ ; &  \node(N)   { $\mathcal{U}\times \left(-1,1\right]$}  ; & ; \\
  \node(M1)   { $\overline{M}_1\supset\overline{V}_1$} ; &   ; & \node(M2)   { $\overline{V}_2\subset\overline{M}_2$} ;       \\};
\draw[->] (N) -- (M1) node[anchor=south east]  {$\rho\circ \varepsilon$};
\draw[->] (N) -- (M2) node[anchor=south west]  {$\Psi_{\mathcal{U}}$};
\draw[->] (M1)   -- (M2) node[anchor=north] {$\overline{\mathrm{id}}$};
\end{tikzpicture}
\end{equation}
defining $\overline{\mathrm{id}}$ as a bijection such that $\left.\overline{\mathrm{id}}\right|_{V}=\mathrm{id}:V\rightarrow V$ is the identity map. 
By \cite[Prop. 6.1.2]{BC}, $\overline{\mathrm{id}}$ is a diffeomorphism extending the identity map in $V\subset M$. 
Finally, taking a covering of $\mathcal{N}$, we can define globally $\overline{\mathrm{id}}=\overline{M}_1\rightarrow \overline{M}_2$ as a diffeomorphism. 
\qed
\end{proof}

\begin{remark}
Notice that if $\varphi: \left(M_1,\mathcal{C}_1\right) \rightarrow \left(M_2,\mathcal{C}_2\right)$ is a conformal diffeomorphism then, by the Reconstruction theorem \ref{teo-reconstruction}, there exists a diffeomorphisms $\phi:\mathcal{N}_1 \rightarrow \mathcal{N}_2$ such that for any sky $X\in \Sigma_1$ of $\mathcal{N}_1$, then $\phi\left(X\right)\in \Sigma_2$ is a sky of $\mathcal{N}_2$. 
Then, in virtue of Theorem \ref{theorem-diff-struct} and by construction of the canonical extension, whenever one of the conformal manifolds is a proper $L$--spacetime for any metric in $\mathcal{C}$, then the other is also a proper $L$--spacetime for any metric and both $L$--extensions are diffeomorphic by the extension of $\varphi$. See \cite[Cor. 8.2]{Ba18}.
\end{remark}

It is possible to show a converse result of theorem \ref{theorem-diff-struct} to conclude the characterization of the canonical extension as the unique $L$--extension (up to conformal diffeomorphism). The proof of the following proposition can be found in \cite[Prop. 8.1]{Ba18}.

\begin{proposition}
Let $M$ be a $3$--dimensional, strongly causal, light non--conjugate, sky-separating, conformal Lorentz manifold.
If $M$ admits a future $L$--extension $\overline{M}$, then the canonical field of directions $\oplus:\mathcal{N}\rightarrow \mathbb{P}\left(\mathcal{H}\right)$ defines a regular and smooth distribution. 
Moreover, the distribution $\boxplus$ defined by the L--extension verifies $\boxplus=\oplus$.
\end{proposition}

\section{Conclusion}\label{sec:Conclusion}
From the initial seed, R. Penrose's twistor theory, together with the original work of R. Low, the geometry of the space of light rays $\mathcal{N}$ of a Lorentz conformal manifold $M$ appears as a complementary model of physical universe in addition to the spacetime model. It contains all the information of $M$ but it should be treated in a different way. In this review we have introduced the geometric structures of $\mathcal{N}$, characterized the causal structure of $M$ and built the $L$--boundary for $\dim M=3$. 

But many things still need to be done. Since $\mathcal{N}$ is invariant by conformal diffeomorphisms of $M$, so conformal invariants in $M$ can be well defined in $\mathcal{N}$ and they can be determined, at least theoretically, by light rays. 
So, we can wonder what new additional geometrical objects in $\mathcal{N}$ define conformal invariants of $M$.

The question about if a new condition in the definition of $L$--spacetime should be added to ensure that the distribution $\overline{D^{\sim}}$ is smooth, is still open. 
Moreover, the particular geometry of $\mathscr{L}\left(\mathcal{H}\right)=\mathbb{P}\left(\mathcal{H}\right)$ when $\dim M=3$ allows to build the $L$--boundary explicitly. The construction for $\dim M>3$ is not free of problems, because the geometry of $\mathscr{L}\left(\mathcal{H}\right)$ is richer that the one in $\mathbb{P}\left(\mathcal{H}\right)$. It is necessary to study the additional problems arising in the general case, for example, $\widetilde{\mathcal{N}}$ is not an open submanifold of $\mathscr{L}\left(\mathcal{H}\right)$ and, under the same hypotheses, the map $\sigma$ of equation (\ref{difeo-sigma}) is not an embedding but an injective immersion. This pushes us to build the $L$--extension in a more local way. Furthermore, the projective parameter does not arise so naturally as it does in $\mathbb{P}\left(\mathcal{H}\right)$. 
The study of the conditions for an $L$--spacetime with non-compact orbits of $\partial^{+}\mathcal{D}^{\sim}$ making of $\overline{D^{\sim}}$ a regular distribution still remains unresolved. 
In any case, we believe that it is worth exploring this form of construction of the $L$--boundary for general dimension $m\geq 3$.

\section*{Acknowledgements}
\addcontentsline{toc}{section}{Acknowledgements}
The authors wish to thank two anonymous referees who provided useful comments and suggestions to improve the quality of this paper and A.~Bautista would like to thank the organizing committee of the meeting \emph{Singularity theorems, causality, and all that. A tribute to Roger Penrose} for its kind invitation.


\phantomsection

\end{document}